\documentclass[10pt,a4paper]{article}

\usepackage{fullpage}
\usepackage[usenames,dvipsnames,svgnames,table]{xcolor}
\usepackage[english]{babel}
\usepackage[bookmarksnumbered,linktocpage,hypertexnames=false,colorlinks=true,linkcolor=blue,urlcolor=blue,citecolor=blue,anchorcolor=green,breaklinks=true,pdfusetitle]{hyperref}

\usepackage{amsmath,amssymb,amsthm,mleftright}
\usepackage[capitalize]{cleveref}
\usepackage{mathtools}
\usepackage{graphicx}
\graphicspath{{./figures/}}

\usepackage{enumitem}
\usepackage{braket}

\usepackage{tikz,pgfplots}
\pgfplotsset{compat=1.18}

\usepackage{palatino}
\usepackage{mathpazo}
\hypersetup{pdfpagemode=UseNone}
\usepackage[margin=1cm]{caption}

\usepackage{thmtools}
\usepackage{thm-restate}

\newcommand{\ignore}[1]{}
\newcommand{\N}{\mathbb{N}}

\newcommand{\R}{\mathbb{R}}

\newcommand{\C}{\mathbb{C}}
\newcommand{\LB}{L}
\renewcommand{\H}{\mathcal H}

\newcommand{\HR}{\H^R}
\newcommand{\HE}{\H^E}

\DeclareMathOperator{\Tr}{Tr}

\renewcommand{\sup}{\mathrm{sup}}
\renewcommand{\min}{\mathrm{min}}
\newcommand{\argmin}{\mathrm{argmin}}
\renewcommand{\max}{\mathrm{max}}
\renewcommand{\epsilon}{\varepsilon}
\newcommand{\eps}{\epsilon}

\def\01{\{0,1\}}

\DeclareMathOperator{\polylog}{polylog}

\DeclareMathOperator{\Dom}{Dom}
\DeclareMathOperator{\var}{var}
\newcommand{\semi}{\gamma}

\numberwithin{equation}{section}
\allowdisplaybreaks[4]

\renewcommand\bra[1]{{\langle{#1}|}}
\makeatletter
\renewcommand\ket[1]{%
  \@ifnextchar\bra{\k@t{#1}\!}{\k@t{#1}}%
}
\newcommand\k@t[1]{{|{#1}\rangle}}
\makeatother

\newtheorem{defin}{Definition}[section]
\newtheorem{definition}[defin]{Definition}
\newtheorem{proposition}[defin]{Proposition}
\newtheorem{theorem}[defin]{Theorem}
\newtheorem*{theorem*}{Theorem}
\newtheorem{remark}[defin]{Remark}
\newtheorem{corollary}[defin]{Corollary}
\newtheorem{lemma}[defin]{Lemma}

\newtheorem*{claim*}{Claim}

\newtheorem*{conjecture*}{Conjecture}
\theoremstyle{definition}

\DeclarePairedDelimiter{\norm}{\lVert}{\rVert}
\DeclarePairedDelimiter{\parens}{\lparen}{\rparen}
\DeclarePairedDelimiter{\abs}{\lvert}{\rvert}
\newcommand{\vol}{\mathrm{vol}}
\DeclareMathOperator{\supp}{supp}
\DeclareMathOperator{\gap}{gap}

\newcommand{\domain}{M}
\newcommand{\manifold}{M}
\DeclarePairedDelimiterXPP{\normriem}[1]{}{\lVert}{\rVert}{_{R}}{#1}
\DeclarePairedDelimiterXPP{\normeuclz}[1]{}{\lVert}{\rVert}{_{E,z}}{#1}
\DeclarePairedDelimiterXPP{\normz}[1]{}{\lVert}{\rVert}{_z}{#1}
\DeclarePairedDelimiterXPP{\ipriem}[2]{}{\langle}{\rangle}{_{R}}{#1,#2}
\DeclarePairedDelimiterXPP{\ipeucl}[2]{}{\langle}{\rangle}{_{E}}{#1,#2}
\DeclarePairedDelimiterXPP{\ipeuclz}[2]{}{\langle}{\rangle}{_{E,z}}{#1,#2}
\DeclarePairedDelimiterXPP{\ipz}[2]{}{\langle}{\rangle}{_z}{#1,#2}
\newcommand{\Lz}{\Delta^z}

\DeclareMathOperator{\grad}{grad}

\renewcommand{\d}{\mathrm{d}}

\DeclareMathOperator{\poly}{poly}

\newcommand{\TV}{{\mathrm{TV}}}

\newcommand{\loc}{\mathrm{loc}}

\date{\today}
\begin{document}

\title{Self-concordant Schrödinger operators: spectral gaps and optimization without condition numbers}
\author{Sander Gribling\thanks{Tilburg University, Tilburg, the Netherlands. 
} \and Simon Apers\thanks{Universit\'e Paris Cit\'e, CNRS, IRIF, Paris, France.} \and Harold Nieuwboer\thanks{University of Copenhagen, Copenhagen, Denmark.} \and Michael Walter\thanks{LMU Munich, Germany; Ruhr University Bochum, Germany; University of Amsterdam, the Netherlands.}}
\date{}

\maketitle
\vspace{-2.5em}
\begin{abstract}
Spectral gaps play a fundamental role in many areas of mathematics, computer science, and physics. 
In quantum mechanics, the spectral gap of Schr\"odinger operators has a long history of study due to its physical relevance, while in quantum computing spectral gaps are an important proxy for efficiency, such as in the quantum adiabatic algorithm.
Motivated by convex optimization, we study Schr\"odinger operators associated with self-concordant barriers over convex domains and prove non-asymptotic lower bounds on the spectral gap for this class of operators.
Significantly, we find that the spectral gap does not display any condition-number dependence when the usual Laplacian is replaced by the Laplace--Beltrami operator, which uses second-order information of the barrier and hence can take the curvature of the barrier into account.
As an algorithmic application, we construct a novel quantum interior point method that applies to arbitrary self-concordant barriers and shows no condition-number dependence.
To achieve this we combine techniques from semiclassical analysis, convex optimization, and quantum annealing.
\end{abstract}

\vspace{-1.5em}
{
\footnotesize
\tableofcontents
}

\section{Introduction}

The Schr\"odinger operator
\[
\H
= -\frac12 \Delta + f,
\]
where $\Delta = \sum_{j=1}^n \partial_j^2$ is the Laplacian and $f$ a potential function, is a central object in quantum mechanics.
It describes the energies and eigenstates of a quantum system via the time-dependent Schrödinger equation, $\H \ket{\psi} = E \ket{\psi}$, as well as the system's dynamics via the time-independent Schr\"odinger equation, $i\partial_ t \ket{\psi} = \H \ket{\psi}$.
Among these energies and eigenstates, the most studied are the \emph{ground state} and its energy, as well as the \emph{spectral gap}, which is the energy gap between the ground state and the first excited state. 
Spectral gaps have long been studied for their relevance in physics and mathematics (cf.~\cite{witten1982super,simon1983semiclassical,simon1984tunneling}), and this led to a plethora of results on the asymptotics of the spectral gap in varying scenarios \cite{CyFKS1987}.
One noteworthy line of results is based on \emph{semiclassical analysis} \cite{zworski2012}, where the family $-\frac12 \Delta + \gamma^2 f$ for $\gamma \to \infty$ is considered.
Particularly relevant to this work is the insight that, for a convex potential~$f$, the spectral gap is asymptotically governed by a local quadratic approximation of~$f$~\cite{simon1983semiclassical}.

More recently, the rise of quantum computing and quantum algorithms has shown that spectral gaps are also an important proxy for efficient algorithms, as is the case for example in the quantum adiabatic algorithm \cite{farhi2000quantumAdiabatic}.
This strongly suggests to try and exploit the aforementioned line of work on spectral gaps in the design and analysis of quantum algorithms.
Unfortunately, the traditional asymptotic analysis is 
not sufficiently quantitative for this purpose, as it typically omits key dependencies on input parameters such as dimensions and condition numbers.
Indeed, many results based on semiclassical analysis seem to not have been strengthened enough to become useful for quantum algorithms.
(A notable exception are results on the ``fundamental gap'' \cite{andrews11fundamentalgap}, where the dependence is explicit.\footnote{We compare this result to ours in more detail in \cref{sec:technical overview}, after stating our result.})
What is required for algorithmics is a \textit{non-asymptotic} semiclassical analysis, which establishes a gap for an explicit choice of~$\semi>0$ and with an explicit quantitative dependence on all relevant input parameters.

In this work, we demonstrate how the theory of convex optimization can give such a strengthening and, in turn, yield new quantum algorithms for convex optimization.
More specifically, we combine the theory of \textit{self-concordance} from interior point methods~\cite{nesterov1994interior,Renegar01} with that of semiclassical analysis to prove non-asymptotic bounds on spectral gaps of Schrödinger operators with convex potentials that are self-concordant.
To begin with, we use a non-asymptotic, semiclassical analysis of the Schrödinger operator to show that its spectral gap does indeed scale with input parameters such as the dimension and the condition number of the quadratic approximation.
To remedy this, and this is our main contribution, we show that the condition number dependency can be circumvented by replacing the Laplace operator~$\Delta$ with a \emph{Laplace--Beltrami operator}~$L$ that takes into account second-order information such as the local curvature of the potential.
Such an analysis requires bounds on the stability of the local curvature, and this is where we make critical use of self-concordance.

While these structural results are of independent interest, we go on to demonstrate their algorithmic use.
We develop a novel \emph{quantum interior point method} (IPM) for convex optimization that combines the Schrödinger operator with a quantum annealing approach.
We prove efficiency and correctness of the quantum IPM using our non-asymptotic semiclassical analysis.
As a key distinction from other works, we develop the IPM directly at the level of the continuous Schrödinger operator, by introducing, among others, a ``continuous-variable'' quantum algorithm for quantum annealing that involves coupling to an external quantum harmonic oscillator. This approach can be seen as a quantum variant of classical Markov chain algorithms based on continuous Langevin dynamics and Hamiltonian dynamics (on manifolds) \cite{holley1989asymptotics,girolami2011riemann}; it also shares similarities with recent works on continuous-variable quantum algorithms \cite{chabaud25bosonic,brenner2024factoring}.

The remainder of this section is organized as follows. We first give a technical overview of our results: in \cref{sec:technical overview} we discuss how we establish a spectral gap for both operators and in \cref{sec:technical overview qIPM} how to apply these results in an optimization setting, resulting in a quantum interior point method. We state some open questions in \cref{sec:open} and we end with an overview of related work in \cref{sec:related}.

\subsection{Technical overview: Spectral gap bounds} 
\label{sec:technical overview}

We consider a convex function $f\colon \domain \to \R$ on a convex subset $\domain \subseteq \R^n$ of Euclidean space.
We associate to it a first family of Schrödinger operators
\begin{equation} \label{eq:schr-euclidean}
\HE
\coloneqq -\frac12 \Delta + \gamma^2 f
\end{equation}
for semiclassical parameter $\gamma > 0$, $\Delta$ the Laplace operator, and superscript~$E$ referring to the Euclidean setting.
We impose vanishing (Dirichlet) boundary conditions.
To provide some intuition, the energy of a state $\psi$ with respect to the Hamiltonian $\HE$ can be written as $\frac12 \int \|\nabla \psi(x)\|^2\d x + \semi^2 \int f(x)\psi(x)^2 \d x$. The ground state thus balances localization around the minimizer of $f$ with the size of its gradients (the larger $\gamma$, the more the localization).
The Laplacian depends on, and the localization in the above formula is quantified using the Euclidean geometry of~$\R^n$.
However, this geometry need not be well-adapted to the shape of the potential, leading to, as we will see, the occurrence of condition numbers.

For this reason we introduce a second family of Schr\"odinger operators.
Here, we consider $\manifold$ as a Riemannian manifold, with metric~$g$ defined by the Hessian of the convex potential~$f$, and define
\begin{equation} \label{eq:laplacebeltrami}
\HR
\coloneqq -\frac12 L + \gamma^2 f,
\end{equation}
where $L$ is now the Laplace--Beltrami operator on $M$, which takes into account the local geometry, and the superscript $R$ refers to the Riemannian setting.
In this case, the local geometry of $M$ reflects the curvature in the potential (e.g., the local volume element at $x$ is scaled by $\sqrt{\det D^2 f_x}$) -- as does the Laplace--Beltrami operator~$L$, which is a second-order differential operator (like the Laplacian) but with coefficients that likewise depend on the metric.
While seemingly more abstract, this construction is natural: it is ``basis-independent'' (that is, independent of the Euclidean structure) and using second-order information of $f$ has proved highly effective in a variety of settings including optimization (e.g.~Newton's method) and sampling (e.g.~Riemannian Langevin diffusion), as we discuss in \cref{sec:related}.

Under mild conditions, both of the operators $\HE$ and $\HR$ have a discrete spectrum (eigenvalues) and the ground state is unique.
Our goal then is to give a non-asymptotic lower bound on the spectral gap.
For this we rely on ideas from semiclassical analysis, which allows us to argue that for $\gamma$ sufficiently large, the low-energy properties of both operators are well approximated by their ``harmonic approximation''.
In the Euclidean setting, for $\HE$, this is a quantum harmonic oscillator with potential $\semi^2\frac12 x^T A x$, where $A$ is the Hessian of $f$ at its minimizer.\footnote{W.l.o.g., we can assume that the minimizer is at the origin and that $f(0) = 0$.}
Consequently, its spectral gap scales linearly in $\semi$, with coefficient $\lambda_{\min}(\sqrt{A})$, the square root of the smallest eigenvalue of $A$.
Critically, in the Riemannian setting, as a consequence of the aforementioned basis-independence of $\HR$, the harmonic approximation of $\HR$ is the \emph{standard} quantum harmonic oscillator with potential $\frac12 x^T x$. This operator has a constant spectral gap, which partly explains the condition-number independence of our spectral gap bound. The other part of the explanation is more subtle, and for this we require that $L$ is defined based on the metric associated to the Hessian of $f$.
In the following, we give a more detailed overview, focusing first on the parts that are common to both the Euclidean and Riemannian settings, and then on the differences.

\paragraph{IMS-localization.} In both settings, our goal is to compare a Schr\"odinger operator $\H$ to a quantum harmonic oscillator~$\H_0$. 
To do so, we will use a technique known as \emph{IMS-localization}, as in \cite{simon1983semiclassical}.
This technique is based on a smooth partition of unity given by nonnegative functions $J, \bar J$ that satisfy $J^2+\bar J^2 = 1$. The IMS-localization formula then implies the identity
\begin{equation} \label{eq:intro-IMS}
\H
= J \H_0 J + J (\H-\H_0) J + \bar J \H \bar J -  \|\nabla J\|^2- \|\nabla \bar J\|^2.
\end{equation}
For both $\H =\HE$ and $\H = \HR$, the term $J(\H-\H_0)J$ involves the multiplication operator $\semi^2 J(f-q)J$, where $q$ is the second-order Taylor expansion of $f$ around its minimizer. To ensure localization, we choose $J$ such that $J(x)=1$ for $x$ within distance $r:=\semi^{-2/5}$ of the minimizer, and $J(x)=0$ for $x$ distance at least $2r$. In both cases, distance will be measured in the \textit{local norm} associated to the Hessian of $f$ at the minimizer. To bound the multiplication operator $\semi^2 J(f-q)J$, we thus need to bound the third derivatives of $f$ for all $x$ that lie within distance $r$ of the minimizer, in the norm specified by the second derivative of $f$. This is a first reason to use \textit{self-concordant} functions.

\paragraph{Self-concordance.}
As mentioned, this is a key concept in modern convex optimization~\cite{nesterov1994interior,Renegar01} and we comment more on its relevance for optimization below.
A (thrice-differentiable) convex function $f\colon D\to\R$ is \emph{self-concordant} \cite[Section 2.5]{nesterov1994interior} if for all $u \in \R^n$,
\[
\abs{D^3 f(x) [u,u,u]}
\leq 2 \, D^2 f(x) [u,u]^{3/2}.
\]
This property for example ensures that $f-q$ can be bounded by an absolute constant $C$ times $r^3$ on the support of $J$, and thus $\semi^2 |J(f-q)J| \leq C \semi^{4/5}$. Another consequence of self-concordance is that the Hessian of $f$ changes slowly, locally. We will make extensive use of this in the Riemannian setting. We now discuss our results in the two settings separately.

\paragraph{Euclidean setting.}
Here we largely make more precise and quantitative the celebrated asymptotic analysis of Simon \cite{simon1983semiclassical}, see also e.g.~the textbooks~\cite{CyFKS1987,hislop2012introduction}.
At its core is the fact that increasing $\semi$ localizes the low energy states of $\HE$ around the minimizer of~$f$, which corresponds to the region where $\HE$ is well approximated by the harmonic approximation $\HE_0 = -\frac12 \Delta + \frac12 \gamma^2 x^T A x$. Here the local norm of $x$ will be $\sqrt{x^T Ax}$.
An upper bound on $\lambda_0(\HE)$ then follows by considering the ansatz ground state $J \ket{\psi_0}$, with $\ket{\psi_0}$ the ground state of $\H^E_0$.
For the lower bound on $\lambda_1(\HE)$, starting from the IMS-localization formula~\eqref{eq:intro-IMS}, we roughly show that
\[
\HE
\succeq \lambda_1(\HE_0) + F + O(\|A\| \gamma^{4/5}),
\]
with $F$ a rank-1 operator.
By the variational or Rayleigh--Ritz principle, this implies that $\lambda_1(\HE) \geq \lambda_1(\HE_0) + O(\|A\| \gamma^{4/5})$.
The error term $O(\|A\|\semi^{4/5})$ is less than $\lambda_1(\HE_0)-\lambda_0(\HE_0) = \semi \lambda_{\min}(\sqrt{A})$ when $\gamma \in \Omega((\|A\|/\lambda_{\min}(\sqrt{A}))^5)$, showing the appearance of a condition number. 
After also taking into account dependencies on the dimension $n$, we arrive at the following theorem.

\begin{theorem}[Spectral gap of $\HE$, informal]
Let $\HE$ and $\H^E_0$ be as above.
If $\gamma \gg  (n\|A\|/\lambda_{\min}(\sqrt{A}))^5$  then
\[
\lambda_1(\HE)-\lambda_0(\HE)
\geq \semi \lambda_{\min}(\sqrt{A})/2.
\]
\end{theorem}
We point out that our result illustrates the dependence on the conditioning of the potential, but it is independent of the geometry of the domain. The ``fundamental gap'' result of~\cite{andrews11fundamentalgap} that we mentioned earlier is complementary: for any convex potential $V:\domain \to \R$, it establishes a spectral gap for the operator $-\Delta+V$ that depends only on the diameter of $\domain$ (but not on $V$). 

\paragraph{Riemannian setting.}
Turning to the Riemannian setting, we pursue a similar strategy.
The key challenge is that now
\[
\HR - \HR_0
= -\frac12(L-\Delta) + \gamma^2 (f - q),
\]
and so we have to additionally bound the approximation of the Laplace--Beltrami operator $L$ by the Euclidean Laplacian $\Delta$.
To overcome this challenge, we can use the basis-independence of $\HR$: this allows us to choose our coordinates in such a way that the Hessian of $f$ at the minimizer is equal to the identity, i.e., we can compare $\HR$ with a standard quantum harmonic oscillator.
Since $L$ is based on the metric induced by the Hessian of $f$, and $\Delta$ essentially corresponds to a ``flat metric'', we get that $J L J$ approximates $J \Delta J$ provided that the Hessian is sufficiently stable close to the minimizer of $f$.
This is the second reason that self-concordance is natural within such a semiclassical analysis.
Using the stability of Hessians provided by self-concordance, we prove the following theorem.

\begin{theorem}[Spectral gap of $\HR$, informal]
Let $\HR$ and $\HR_0$ be as above.
If the potential $f$ is a self-concordant function and $\gamma \gg (n \log(n))^5$ then
\[
\lambda_1(\HR)-\lambda_0(\HR)
\geq \semi/2.
\]
\end{theorem}
Here the $n$-dependence is due to the ground state of the harmonic approximation: it corresponds to an $n$-dimensional Gaussian measure $\mathcal N(0,(2\semi I)^{-1})$ and thus it concentrates on a ball of radius roughly $\sqrt{n/\semi}$. To ensure that this region coincides with the support of $J$, whose radius is $\semi^{-2/5}$, we require $\sqrt{n/\semi} \leq \semi^{-2/5}$ and thus roughly $\semi \geq n^5$. 

\subsection{Technical overview: Quantum interior point method} \label{sec:technical overview qIPM}

We now describe how our bounds on spectral gaps can be utilized to develop a quantum interior point method for convex optimization with non-asymptotic performance guarantees.

\paragraph{Classical path-following.}
Self-concordance was originally introduced in the context of \emph{interior point methods} (IPMs) \cite{nesterov1994interior}.
To solve a constrained convex optimization problem $\min_{x \in D} c^T x$, for~$c \in \R^n$ and a convex domain $D \subseteq \R^n$, IPMs introduce a barrier function $f$ such that $f(x) \to \infty$ as $x \to \partial D$, and then solve the unconstrained problem
\[
\min_x \; \eta c^T x + f(x)
\]
for $\eta > 0$.
Now if $f$ is self-concordant function with ``barrier parameter'' $\vartheta$, then the central path $\{x_\eta = \argmin_x \left( \eta c^T x + f(x) \right)\}_{\eta \geq 0}$ approaches a minimizer at a known rate: $c^T x_{\eta} \leq \min_{x\in \overline D} c^T x + \vartheta /\eta$.
Moreover, self-concordance also gives a bound on the stability of the central path: namely, $\eta \leq \eta' \leq (1+{\delta}/{\sqrt{\vartheta}})\eta$ implies that $\|x_\eta - x_{\eta'}\|_{x_\eta} \leq O(\delta)$, where $\|y\|_x = \sqrt{D^2 f_x[y,y]}$.
Such a bound is critical for classical path-following methods that use Newton steps to trace out the central path and find an approximate minimizer.

\paragraph{Quantum path-following.}
In almost direct analogy, we can describe a quantum path-following method by considering a sequence of Hamiltonians
\[
\HR(\eta)
= -L + \semi^2(\eta c^T x +f)
\]
with corresponding ground states $\ket{\psi_\eta}$.
From our semiclassical analysis we then learned two things.
First, for a large $\gamma \in \widetilde{\Omega}(n^5)$, the spectral gap of $\HR(\eta)$ is proportional to $\gamma$. Importantly, $\semi$ is \textit{independent of $\eta$}. 
Second, for such a choice of~$\gamma$, each ground state $\ket{\psi_\eta}$ is close to a Gaussian wave function concentrated around~$x_\eta$.
Combining this with the stability of the central path~$\{x_\eta\}$ allows us to prove that $\braket{\psi_\eta|\psi_{\eta'}} \geq 1/2$ for $\eta' = (1+O(1/\sqrt{\gamma \vartheta})) \eta$.
This implies that there exists a sequence $1=\eta_0,\eta_1,\dots,\eta_T=\vartheta/\varepsilon$ with $T \in O(\sqrt{\gamma \vartheta} \log(\vartheta/\varepsilon))$ such that the ``quantum central path'' $\{\ket\psi_\eta\}$ is similarly stable: it satisfies $\braket{\psi_{\eta_\ell}|\psi_{\eta_{\ell+1}}} \geq 1/2$ for all $\ell$.

Combining these two observations allows us to use a technique called \emph{quantum annealing} to algorithmically follow this ``quantum central path'' from the initial state $\ket{\psi_{\eta_0}}$ (which we assume can be prepared efficiently\footnote{At least knowing $x_{\eta_0}$ is a typical assumption in IPMs and relatively mild: any point in the interior of $D$ lies on \textit{a} central path (for a suitably chosen objective), and one can follow this path backwards to the point where all central paths meet (corresponding to $c=0$). See for example \cite[Sec.~2.4.2]{Renegar01}.}) to the final state $\ket{\psi_{\eta_T}}$ (which is localized around the minimizer of our optimization problem).
Following e.g.~the approach by Wocjan and Abeyesinghe~\cite{wocjan2008speedup}, we can obtain an $\varepsilon$-approximation of $\ket{\psi_{\eta_T}}$ by making a total of $O(T \log(T/\varepsilon))$ calls to the $\pi/3$-rotations
\[
R_\ell
= I + (e^{i\pi/3}-I) \ket{\psi_{\eta_\ell}}\bra{\psi_{\eta_\ell}}
\]
and their inverse, for varying $\ell$.

\paragraph{Unbounded Hamiltonian ground state projector.}
Having reduced our path-following method to rotating around the ground state of a Hamiltonian $\HR(\eta)$, we are faced with a new challenge.
In finite dimensions, it is known that given black-box access to the (controlled) time-evolution operator $e^{i H t}$ one can use quantum phase estimation to implement a projection on (and thus rotation around) its ground state. 
The operator $\HR(\eta)$, however, is infinite-dimensional and even \emph{unbounded}, and this leads to unwanted aliasing effects when trying to use quantum phase estimation.
We instead propose a different approach, which is arguably more native to the type of continuous-variable operator that we are considering.
We show that by controlling the evolution $e^{it\HR(\eta)}$ on an extra continuous-variable register, initialized as a Gaussian state, we can effectively implement the \emph{imaginary-time evolution operator} $e^{-\tau (\HR(\eta))^2}$. If we assume the ground state energy is $0$, then for sufficiently large $\tau$ the operator $e^{-\tau (\HR(\eta))^2}$ approximates a projector on the ground state of $\HR(\eta)$.
At the core of the idea is the Hubbard-Stratonovich transform, which for $x \in \R$ states that
\[
e^{-x^2/2}
= \frac{1}{\sqrt{2\pi}} \int e^{-z^2/2} e^{-ixz} \, \mathrm{d} z,
\]
and thus shows how a (continuous) linear combination of Hamiltonian evolutions corresponds to an imaginary-time evolution.
While similar ideas appeared in previous works \cite{chowdhury17quantum,apers2022quadratic}, these were restricted to finite-dimensional operators.

We emphasize that our contribution here is mostly conceptual and not (yet) competitive with the current state of the art: we use a different computational model, and require assumptions that need justification (e.g.~how does one estimate the ground state energy of $\HR(\eta)$).
However, we view our algorithm as a first concrete step towards a new type of quantum algorithm for optimization that holds the potential to improve over the state of the art (e.g.~by avoiding condition-number dependencies of prior work, see \cref{sec:related}). At the same time, our work raises several new questions, which we discuss below.

\subsection{Open questions}
\label{sec:open}

Our work leaves open many directions for future work, here we list a few.

\paragraph{Spectral gap bounds for Schr\"odinger operators.} We establish our spectral gap bounds via semiclassical arguments. More specifically, we compare our Schr\"odinger operators to quantum harmonic oscillators. As discussed in \cref{sec:technical overview}, the concentration of the ground state of the harmonic oscillator forces us to choose a semiclassical parameter that scales polynomially with the dimension (and other factors in the Euclidean setting). A first natural question is whether one can improve, or even remove, this dependence by comparing to a Schr\"odinger operator whose ground state exhibits stronger concentration (i.e., whose potential grows faster than quadratically).

Using a semiclassical analysis allows us to use \textit{local} properties of our Schr\"odinger operator. What about using more \textit{global} properties? In the related setting of Markov diffusion processes (see \cref{sec:related}), more global properties such as Poincaré- or (log)-Sobolev-inequalities have been successfully used to quantify the rate of convergence towards the stationary distribution. A striking result here is that Newton-Langevin diffusion \eqref{eq:NLD} converges at a rate $e^{-t}$ for any strictly log-concave target distribution, due to the Brascamp-Lieb inequality. 

In the Riemannian setting, we use the Hessian of the potential $f\colon\domain \to \R$ to define a metric on~$\domain$. This allows us to view $\domain$ as a Riemannian manifold and we study the operator $-L +f$ where $L$ is the Laplace--Beltrami operator. In particular, the potential and the metric are directly related to each other. Can these be decoupled to some extent? Such results are known in the Riemannian Langevin diffusion literature. One concrete example is the mirror-Langevin diffusion process~\cite{ahn2021discrete}, where the self-concordant barrier $\phi$ and the convex potential $V$ need only be related to each other via $\alpha$-relative strong convexity and $\beta$-relative smoothness which requires $\alpha \nabla^2 \phi(x) \preceq \nabla^2 V(x) \preceq \beta \nabla^2 \phi(x)$ for all~$x \in \domain$. Even for time-discretized dynamics, this led to bounds that only depend on $\alpha,\beta$, and the dimension.

\paragraph{Algorithms (for optimization) based on unbounded operators.}

Turning our attention to algorithms, we proposed a quantum path-following method (an interior point method) for convex optimization based on the technique of quantum annealing. Our contribution here is conceptual and not (yet) competitive with the state of the art. We design an algorithm that performs controlled-Hamiltonian simulation with the continuous-variable Schr\"odinger operator $\H$ directly. Using a continuous-variable control register initialized in a Gaussian state, we can implement for example a projector on the ground state of $\H$. We did so under assumptions, see \cref{sec:Qannealing}, one of which is the ability to estimate the ground state energy of $\H$ (which essentially amounts to estimating the minimum value of~$f + \eta c^T x$). We leave removing such assumptions for future work.

Another natural question concerns the computational model: our approach fits naturally with the continuous-variable operator that we consider, but it does not fit naturally in the usual qubit-based model of computation. A better understanding of the complexity of continuous-variable algorithms is needed. For background information, see \cite{GKPoscillator,LB99model,bartlett02efficientgauss,ferraro2005gaussianstatescontinuousvariable,weedbrook12gauss}; for recent progress, see e.g.~\cite{chabaud25bosonic}. In a similar spirit, we obtained improved spectral gap bounds by going from the Euclidean setting to the Riemannian, but what is the complexity of simulating the Riemannian Schrödinger operator? Can one effectively incorporate the geometry of the space into a Hamiltonian simulation algorithm?

In a complementary direction, our path-following method is based on quantum annealing. In the finite-dimensional setting, however, an alternative would be to use adiabatic quantum computation \cite{farhi2000quantumAdiabatic}. Here, a spectral gap estimate also plays a central role. However, there is typically also a dependence on the operator norm of the Hamiltonian (and its derivatives), which prohibits their direct application in the unbounded setting. Recent work has made progress towards an adiabatic theorem that is applicable to unbounded operators, though it requires introducing an ``energy cutoff'' \cite{mozgunov23adiabatic}.
As a clear direction for future work, we leave open the question of developing a suitable (non-asymptotic) adiabatic theorem applicable to our setting.

\subsection{Related work}
\label{sec:related}

Here we give an overview of related work, emphasizing connections between conditioning and second-order information. 

\paragraph{Quantum algorithms for convex optimization.}
 We focus here on methods for constrained convex optimization, as this is the setting relevant to our work.  
 For a more complete treatment of quantum algorithms for optimization, we refer the reader to the recent survey~\cite{Abbas2024Challenges}.

The earliest proposals were quantum algorithms for semidefinite programming (SDP). Two classes of quantum SDP solvers emerged, each seeking to reduce the per-iteration complexity of a classical framework by incorporating quantum subroutines. The works in \cite{vAGGdW:quantumSDP, vanapeldoornImprovedSDP,  brandaoSDP17, brandão2019quantum} accelerated the multiplicative weights method~\cite{arora2016combinatorial} using the observation that the candidate solutions generated at each iterate are \textit{Gibbs states}, which can be efficiently prepared on quantum computers. The early proposals of \textit{quantum interior point methods} (QIPMs)~\cite{kerenidis2020quantum,augustino2021quantum}   sought to leverage the fact that one can perform certain linear algebraic operations (such as solve linear systems) on quantum states and unitaries faster than can be carried out for classical vectors and matrices \cite{gilyen2019qsvt,chakraborty19power} within the IPM framework. The quantum multiplicative weights and QIPM frameworks for SDP were specialized to linear programming  in \cite{vanapeldoorn2019quantum, bouland2023, gao2023logarithmic} and \cite{mohammadisiahroudi2024efficient, mohammadisiahroudi2023inexact}, respectively. Both types of algorithms typically incur a dependence on some form of conditioning: for the multiplicative weights based methods this appears in the form of a dependence on the \textit{width} of the oracle (which can be related to the size of optimal solutions); for the quantum interior point methods there is typically (see~\cite{apers2023quantum} for an exception) a polynomial dependence on the condition number of the Newton system that needs to be solved in each iteration (which in turn depends on conditioning of the instance, as well as the desired accuracy). 

A second line of research that is most relevant to our work are protocols that approach convex optimization by simulating the time dynamics defined by  a Schr\"odinger operator. While quantum algorithms based on dynamical simulation have been studied since the early days of the field~\cite{farhi2000quantumAdiabatic}, they traditionally focused on discrete optimization problems. This changed with the work of Leng et al.~\cite{leng2023quantumhamiltoniandescent}, who introduced a quantum analogue of gradient descent called \textit{Quantum Hamiltonian Descent} (QHD). The QHD dynamics arise from applying canonical quantization to a continuous-time dynamical system developed by Wibisono, Wilson, and Jordan~\cite{wibisono2016variational}. 
Convergence to the global optimum (in the convex setting) can be established using a Lyapunov argument, and a refined analysis is provided in~\cite{chakrabarti2025speedups}. A number of works have extended the ideas underlying QHD to develop new simulation-based quantum algorithms. One can obtain a quantum analogue of stochastic gradient descent upon considering dynamics of a system coupled with an infinite heat bath~\cite{Chen2025QuantumLangevinDynamics}. Augustino et al.~\cite{augustino2024quantumcentralpathalgorithm} generalized these techniques to the constrained setting by deriving a Schr\"odinger operator that captures the time dynamics of the central path in linear programming. 

We remark that the results in \cite{leng2023quantumhamiltoniandescent, augustino2024quantumcentralpathalgorithm, leng2025operatorlevelquantumaccelerationnonlogconcave, Chen2025QuantumLangevinDynamics} typically rely on \textit{(i)} a quantum simulation algorithm for simulating quantum dynamics in real-space from \cite{childs2022quantum} and/or \textit{(ii)} an adiabatic theorem for unbounded Hamiltonians, to prove convergence. Regarding \textit{(i)}, the algorithm in \cite{childs2022quantum} seems to not properly control all forms of error that can impact the quality of the solution and the resources required to carry out the simulation, see \cite{chakrabarti2025speedups} for a detailed discussion of these issues. As for \textit{(ii)}, we are not aware of a suitable adiabatic theorem that directly applies to the infinite-dimensional setting. One approach is to use a form of discretization, after which one can apply a finite-dimensional adiabatic theorem. Rigorously analyzing the error incurred by discretization, however, remains a challenging open problem. Using a grid to discretize space, for instance, inevitably leads to dependencies on the size of the domain, as well as conditioning parameters of e.g.~the ground state (the size of the first derivatives determines the spacing needed in the grid). An alternative would be to use an adiabatic theorem that directly applies to the infinite-dimensional setting. Such results are however not readily available in the literature, see our discussion in \cref{sec:open}.

\paragraph{Gradient flow and Newton flow.}
Differential equations naturally arise if one takes the continuous-time limit of iterative methods such as gradient descent or Newton's method (alternatively, one may also discretize a differential equation's evolution to obtain an iterative method). We refer the interested reader to a recent survey \cite{GarciaTrillos2023OptimizationSampling}, and to \cite{bubeck2015convex} for an overview of the discrete-time setting. 
One of the earliest examples is the \emph{gradient flow} associated to a function $f\colon\R^n \to \R$ defined as the curve $(x_t)_{t \geq 0}$ solving the differential equation $\frac{\d}{\d t} x_t = -\nabla f(x_t)$. One immediately obtains
\[
\frac{\d}{\d t} f(x_t) = - \|\nabla f(x_t)\|^2.
\]
If we assume some additional structure on $f$ that upper bounds the right hand side by $-\alpha(f(x_t) - f(x^*))$ where $x^*$ is a minimizer of $f$ and $\alpha>0$, then this identity implies a convergence rate of gradient flow towards a minimizer: $f(x_t)-f(x^*) \leq \exp(-\alpha t) (f(x_0)-f(x^*))$. E.g., in the optimization literature, $f$ is said to satisfy the $\alpha$-Polyak-{\L}ojasiewicz (PL) inequality if
\begin{equation}
\tag{PL} \label{eq:PL}
\alpha (f(x)-f(x^*)) \leq \frac12 \|\nabla f(x)\|^2, \qquad \forall x,
\end{equation}
where again $x^*$ is a minimizer of $f$. The $\alpha$-PL inequality holds for example for $\alpha$-strongly convex functions, see e.g.~\cite{karimi2016linear}. The above discussion relies on properties of $f$ with respect to the $2$-norm, in a fixed coordinate system. To avoid such dependencies, Nemirovskii and Yudin \cite{NemirovskyYudin1983} introduced mirror flow; a special case is Newton's flow which is based on the differential equation
\[
\frac{\d}{\d t} x_t = -(\nabla^2 f(x_t))^{-1}\nabla f(x_t).
\]
Exponential convergence rates hold under a Newton-like analogue of the PL inequality:
\begin{equation}
\tag{Newton-PL} \label{eq:NPL}
\alpha (f(x)-f(x^*)) \leq \frac12 \nabla f(x)^T (\nabla^2f(x))^{-1} \nabla f(x), \qquad \forall x.
\end{equation}
Convex quadratics satisfy \eqref{eq:NPL} with $\alpha=1$, illustrating the scale-invariance of Newton's method. The quantity $\lambda(x) = \sqrt{\nabla f(x)^T (\nabla^2f(x))^{-1} \nabla f(x)}$ is typically referred to as the Newton-decrement and plays a central role in the theory of IPMs. For instance, self-concordant functions satisfy \eqref{eq:NPL} \textit{locally}\footnote{E.g., \cite[Theorem 5.1.2]{nesterov2018lectures} shows that if $\lambda(x) \coloneqq \sqrt{\nabla f(x)^T (\nabla^2f(x))^{-1} \nabla f(x)} <1$, then $f(x)-f(x^*) \leq \omega^*(\lambda(x))$ where $\omega^*(t) = -t-\ln(1-t)$. It can be verified that $\omega^*(t) \leq t^2$ for $0 \leq t \leq 0.683$.}: if 
$\lambda(x) \leq 0.683$, then  $f(x)-f(x^*) \leq \lambda(x)^2$.

\paragraph{Markov chains.}

Here we consider a canonical sampling problem, which can be thought of as the sampling analogue of convex optimization; in the next paragraph we discuss its connection to a particular class of Schrödinger operators. Given a convex potential $V:\R^n \to \R$, the task is to sample from the distribution $\pi$ whose density is proportional to $e^{-V(x)}$. Langevin diffusion is a well-studied method to solve this problem. It is based on the stochastic differential equation
\begin{equation} \tag{LD}
    \label{eq:Langevin diffusion}
\d X_t = -\nabla V(X_t) \d t + \sqrt{2} B_t,
\end{equation}
where $(B_t)_{t \geq 0}$ is a Brownian motion on $\R^n$. We refer the reader to, e.g.~\cite{bakry2014analysis,chewi2025book} for formal definitions and more information. Here we highlight some connections to Schrödinger operators and convex optimization. Langevin diffusion can be thought of as the sampling analogue of the gradient flow for optimization mentioned above, in several ways. Clearly, omitting the Brownian motion term recovers the dynamics used in optimization. A much more fruitful connection is based on the distribution of $X_t$: \cref{eq:Langevin diffusion} then corresponds to the gradient flow associated to the Kullback-Leibler divergence $D_\mathrm{KL}(\cdot,\pi)$~\cite{jordan1998variational}. This connection inspired new sampling algorithms based on classical optimization algorithms. In particular, several sampling analogues of Newton's method have been studied \cite{ZhangPeyreFadiliPereyra2020,chewi2020exponential}. Of particular importance for us is the version called Newton Langevin diffusion in \cite{chewi2020exponential}; it is based on the stochastic differential equation
\begin{equation} \tag{NLD} \label{eq:NLD}
    X_t = \nabla V^*(Y_t), \quad \d Y_t = - \nabla V(X_t) \d t + \sqrt{2} [\nabla^2 V(X_t)]^{1/2} \d B_t,
\end{equation}
where $V^*$ denotes the convex conjugate of $V$ (cf.~\cite{boyd2004convex}) and $(B_t)_{t \geq 0}$ is again a Brownian motion on $\R^n$.

One can obtain convergence rates, in the same way as above for gradient flow and Newton flow, for various distance measures by establishing a suitable functional inequality. For the Langevin diffusion process, one has for example linear convergence with rate $C_P>0$ in the chi-squared divergence\footnote{For convergence in the above-mentioned KL-divergence, one would need to establish a log-Sobolev inequality.} (defined as $\chi^2(\mu, \pi):=\var_\pi (\d\mu/\d\pi) = \int (\d\mu/\d\pi)^2 \d \pi-1$ if $\mu \ll \pi$) if the Poincaré inequality
\begin{equation}
    \label{eq:poincare}
\var_\pi g := \int \big(g-\mathbb{E}[g]\big)^2 d\pi \leq C_P \, \mathbb E[\|\nabla g\|^2], \qquad \forall \text{ locally Lipschitz } g \in L^2(\pi).
\end{equation} For $\alpha$-strongly convex potentials the Poincaré inequality holds with $C_P = 1/\alpha$.

A striking difference with the optimization literature is the following. If we consider the convergence rate of \eqref{eq:NLD} in the chi-squared distance, then the corresponding functional inequality is the (mirror) Poincaré inequality
\begin{equation}
    \label{eq:Mpoincare}
\var_\pi g \leq C_{MP} \, \mathbb E[ \langle \nabla g, [\nabla^2 V]^{-1} \nabla g \rangle], \qquad \forall \text{ locally Lipschitz } g \in L^2(\pi).
\end{equation}
A celebrated result of Brascamp and Lieb shows that this inequality holds with constant $C_{MP}=1$ whenever $V$ is strictly convex \cite{brascamp1976inequality}. In particular, this establishes a linear convergence rate without any dependence on the geometry of $V$ beyond strict convexity. For the Newton flow differential equation in optimization, similar rates hold under much stricter assumptions: they hold for example for convex quadratic functions, or for selfconcordant functions (locally!).

We remark that the above is about continuous-time diffusion processes. Analyzing time-discretized versions is typically harder, for recent progress see e.g.~\cite{ZhangPeyreFadiliPereyra2020,Chewi2025AnalysisLangevinMonteCarlo}. Finally, we mention that a \textit{damped} version of the Langevin diffusion, where the Brownian motion is multiplied by a small factor $h>0$, has additional connection to optimization: its convergence rate is related to that of stochastic gradient descent~\cite{su2016accelerated}, and its log-Sobolev constant, as $h \to 0$, surprisingly equals the optimal constant in the PL-inequality~\cref{eq:PL}~\cite{ChewiStromme2024Ballistic}. 

\paragraph{Witten Laplacian.}  We mention here a particularly well-studied Schrödinger operator: the Witten Laplacian, used by by Witten in his proof of the Morse inequalities \cite{witten1982super}. For ease of notation, we restrict here to the Euclidean setting, but we note that the objects below have analogues in the Riemannian setting (see~e.g.~\cite{Cheng2022riemannian,Hsu2002stochmanifold,girolami2011riemann,bakry2014analysis,gatmiry2022convergenceriemannianlangevinalgorithm,ahn2021discrete}). Given a potential function $V$ 
it takes the form
\begin{equation} \label{eq:witten}
    -\Delta + \|\nabla V\|^2 -  \Delta V.
\end{equation}
The Witten Laplacian is strongly connected to the Langevin diffusion process discussed above. 
In \cref{eq:Langevin diffusion} we describe the Langevin diffusion process via a stochastic differential equation. Equivalently, one can study the Markov semigroup associated to \cref{eq:Langevin diffusion}, we refer to \cite{bakry2014analysis,chewi2025book} for more information. The generator of the Markov semigroup associated to \eqref{eq:Langevin diffusion}, with potential $2V$, is (cf.~\cite[Example 1.2.4]{chewi2025book})
\[
 \mathcal L = \Delta - 2\nabla V \cdot \nabla.
\]
This generator is related to the Witten Laplacian via the relation
\[
-\Delta + \|\nabla V\|^2 -  \Delta V = - e^{-V} \circ \mathcal L \circ e^V,
\]
where we view $e^V$ as a multiplication operator and $\circ$ denotes composition.
In particular, this relates the spectral gap of the Witten Laplacian to that of $\mathcal L$ and it shows that the ground state of the Witten Laplacian is proportional to $e^{-V}$, the pointwise square-root of the stationary distribution of the Langevin diffusion process.
This is a first crucial difference between the Schrödinger operators we study in our work and the Witten Laplacian: in our setting the ground state is not known. The second crucial difference is that for the Witten Laplacian one can leverage the connection to Langevin diffusion to analyze the spectral gap, whereas no such connection is known for $-\Delta + f$. The connection between Langevin diffusion and the Witten Laplacian has recently been used to derive improved quantum algorithms for sampling~\cite{leng2025operatorlevelquantumaccelerationnonlogconcave}.

\paragraph{Acknowledgments.}

We are particularly grateful to Brandon Augustino for many useful discussions. HN thanks Martin~Dam~Larsen for helpful discussions and pointers regarding Schrödinger operators.

SA was supported in part by the European QuantERA project QOPT (ERA-NET Cofund 2022-25), the French PEPR integrated projects EPiQ (ANR-22-PETQ0007) and HQI (ANR-22-PNCQ-0002), and the French ANR project QUOPS (ANR-22-CE47-0003- 01).

HN acknowledges financial support from the European Research Council (ERC Grant Agreement
No. 818761 and QInteract, Grant No. 101078107) and VILLUM FONDEN via the QMATH Centre of Excellence (Grant No. 10059).

MW acknowledges support by the European Research Council (ERC Grant SYMOPTIC, 101040907), the German Federal Ministry of Research, Technology and Space (QuBRA, 13N16135; QuSol, 13N17173), and the German Research Foundation under Germany's Excellence Strategy (EXC 2092 CASA, 390781972).

\section{Preliminaries}
\label{sec:prelim}

Let $\domain$ be an open subset of the Euclidean space~$E=\R^n$.
When we write~$E$, we will do so to emphasize that only the vector space structure and topology, and sometimes also the inner product~$\braket{u,v}_E \coloneqq u^Tv$ are important, but not the choice of coordinates, as well as to avoid notational ambiguities in the Riemannian setting.
We denote by $C^\infty(\domain)$ the space of real-valued%
\footnote{For the purpose of analyzing the spectral gap of our real-coefficient operators later, it suffices to work with real-valued functions. Indeed, if one has a complex-valued eigenfunction corresponding to a (real) eigenvalue, then both its real and complex parts are also eigenfunctions for the same eigenvalue, so every eigenvalue over the complex valued functions also has an associated eigenfunction over the real numbers.} 
smooth functions on~$\domain$ and by $C^\infty_c(\domain) \subseteq C^\infty(\domain)$ the subspace of compactly supported such functions.
For $f \in C^\infty(\domain)$, the \emph{$k$-th derivative}~$D^k f_x$ at $x\in\domain$ is the $k$-multilinear map
\[
    D^k f_x\colon E \times \dotsb \times E \to \R, \quad
    D^k f_x[u_1, \dotsc, u_k] = \partial_{t_1=0} \dotsb \partial_{t_k=0} f(x + t_1 u_1 + \dotsb + t_k u_k).
\]
When $k=1$ we also use the notation $f'(x) = \partial_{t=0} f(x+t)$.

\subsection{Self-concordance}
Let $\domain$ be an open bounded convex subset of the Euclidean space~$E=\R^n$.
Let $f\colon \domain \to \R$ be a twice continuously-differentiable function.
That is, the gradient and Hessian of~$f$ are well-defined and continuously depend on~$x \in \domain$.
We denote the latter by~$H(x)$.
When $H(x)$ is positive definite for all $x \in \domain$, it 
defines an inner product on $E$ at each point $x \in \domain$: 
\begin{equation}\label{eq:local inner product}
\braket{u, v}_x
\coloneqq \braket{u, H(x) v}_E
= D^2 f_x[u,v]
\quad \text{for}~u, v \in E,
\end{equation}
This is called the \emph{local inner product} at $x \in \domain$.
It gives rise to the \emph{local norm}:
\[
\left\| u  \right\|_{x} \coloneqq \sqrt{\braket{ u, H(x) u }_E} \quad \text{for}~u \in E.
\]
For any $x \in \domain$, the \emph{Dikin ellipsoid} is the open ball of radius~$1$ centered at $x \in \domain$, measured in the local norm at this point: $\{y \in E : \norm{y-x}_x < 1 \}$.
With these definitions in place we recall the central concept of self-concordance~\cite{Renegar01}.

\begin{definition}[Self-concordance]\label{d:p1}
A function $f\colon \domain \to \R$ is said to be \emph{(strongly non-degenerate) self-concordant} if
\begin{itemize}
    \item $f$ is convex with~$H(x)$ positive definite for every~$x \in \domain$,
    \item for all~$x \in \domain$ and~$y \in E$, if~$\norm{y-x}_x < 1$, then~$y \in \domain$, and it holds that
\begin{equation} \label{eq:self-concord-1}
1 - \| y - x \|_{x} \leq \frac{\| v \|_{y}}{\| v \|_{x}} \leq \frac{1}{1 - \| y - x \|_{x}}
\quad\text{for}~0 \neq v \in E.
\end{equation}
\end{itemize}
\end{definition}
This condition can be roughly interpreted as asserting that the Hessian $H(x)$ does not change too quickly, provided one takes small steps (with respect to the local norm).

For later use we also state the following, original definition of self-concordance due to Nesterov and Nemirovskii~\cite[Def.~2.1.1]{nesterov1994interior}:
a thrice continuously-differentiable convex function $f\colon \domain \to \R$ is self-concordant if for all $x \in \domain$ and $u \in E$, the function $\phi(t) = f(x + tu)$ satisfies
$$ | \phi^{\prime \prime \prime}(0) | \leq 2 \left( \phi^{\prime \prime}(0) \right)^{3/2}.$$
We may equivalently write
\begin{equation} \label{eq:self-concord-3}
| D^3 f_x [u,u,u] | \leq 2 (D^2 f_x [u,u])^{3/2},
\end{equation}
see, \cite[Sec.~2.5]{nesterov1994interior}.
This is equivalent to the definition given above when~$f$ is at least~$C^3$-smooth~\cite[Sec.~2.5]{Renegar01}.

Requiring that the Dikin ellipsoid of radius~$1$ around any point is contained in~$\domain$ in \cref{d:p1} implies the following:
\begin{lemma}[{\cite[Thm.~2.2.9]{Renegar01}}]
    \label{lem:strong sc blowup}
    For a self-concordant function~$f\colon \domain \to \R$,
    $f(x) \to \infty$ as $x \to \partial \domain$.
\end{lemma}

The central objects in the theory of interior-point methods are self-concordant barrier functions:
\begin{definition}[Barrier]
\label{def:barrier}
A self-concordant function~$f\colon \domain \to \R$ is called a \emph{self-concordant barrier} for~$\domain$
if the following quantity is finite:
\[ \vartheta \coloneqq \sup \left\{ \| H(x)^{-1} g(x) \|^2_{x} : x \in \domain \right\}, \]
where $g(x)$ is the gradient and $H(x)$ the Hessian of~$f$ at $x \in \domain$.
We refer to $\vartheta$ as the \emph{barrier parameter} of~$f$.
\end{definition}

The condition states that the gradient is uniformly bounded with respect to the dual of the local norm, by the barrier parameter.
Accordingly, one can interpret the barrier parameter as a proxy for Lipschitzness.

Self-concordant barriers can be used for the purpose of optimization in the following way.
The goal is to minimize a linear function~$c^T x$ over the convex set~$\domain$.
To this end, one introduces a parameter~$\eta \geq 0$ and considers
\begin{equation} \label{eq:central}
x_\eta \coloneqq \mathrm{argmin}_{x \in \domain} \left( \eta c^T x + f(x) \right).
\end{equation}
Then $\{x_\eta\}_{\eta \geq 0}$ is known as the \emph{central path} corresponding to $f$ and $c$.
We mention some important properties.
As $\eta \to \infty$, $x_\eta$ converges to a minimizer of $c^Tx$ over $\overline \domain$.
The rate of convergence of the objective is known: $\mathrm{val} \leq c^T x_{\eta} \leq \mathrm{val} + \vartheta /\eta$, where $\mathrm{val} = \min_{x\in \overline \domain} c^T x$ (cf.~\cite[Eq.~(2.12)]{Renegar01}).
Finally, using, e.g., \cite[eq.~(2.15) and Thm.~2.2.5]{Renegar01}, we have the following distance bound on consecutive points on the central path:
\begin{equation} \label{eq:safe eta}
\eta \leq \eta' \leq \left(1+\frac{\delta}{ \sqrt{\vartheta}}\right)\eta \quad \Longrightarrow \quad \|x_\eta - x_{\eta'}\|_{x_\eta} \leq O(\delta).
\end{equation}

\subsection{Quantum harmonic oscillator}

A central object in our work is the $n$-dimensional quantum harmonic oscillator, for ease of reference we record its key properties here. Consider the $n$-dimensional quantum harmonic oscillator
\[
\H_0
= -\frac{1}{2}\Delta + \frac{1}{2} \gamma^2 x^T A x
\]
with $A$ a positive definite $n\times n$ matrix.
Then $\H_0$ has the unique ground state
$\psi_0(x) = 
C_0 \exp(-\frac{\gamma}{2} x^T \sqrt{A} x)$,
where $C_0 = (\det(\semi \sqrt{A})/\pi^n)^{1/4}$ is the normalization constant such that $\norm{\psi_0} = 1$.
The corresponding eigenvalue (ground state energy) is given by
\[ \lambda_0(\H_0) = \frac{\gamma}{2} \Tr[\sqrt{A}], \]
and the second eigenvalue is given by
\[
\lambda_1(\H_1)
= \frac{\gamma}{2} \left( \Tr[\sqrt{A}] + 2 \lambda_{\min}(\sqrt{A}) \right).
\]

\subsection{Bump function and helper inequality}

We will use bump functions with certain properties to localize the analysis of the Schr\"odinger operators of interest.
The scaling with the parameter~$\semi>0$ is motivated by our later application.
See \cref{fig:bump} for an illustration.

\begin{lemma} \label{lemma:j and barj properties}
    Let $\semi > 0$. Then there exist real-valued functions~$j,\bar j \in C^\infty(\R)$ such that
    $j^2 + \bar j^2=1$,
    $j(t) = 1$ for $\abs{t} \leq \semi^{-2/5}$,
    $j(t) = 0$ for $\abs{t} \geq 2 \semi^{-2/5}$, and such that $\abs{j'(t)}, \abs{\bar{j}'(t)} \leq 5 \semi^{2/5}$ for all $t\in\R$.
\end{lemma}
\begin{proof}
  It suffices to construct such functions for $\semi=1$, since the general case can then be obtained through the substitution $t = \semi^{2/5}x$. Define
  \begin{align*}
    a(x) & =
    \begin{cases}
      \exp\parens{-\frac1x} & x > 0 \\
      0 & \text{otherwise,}
    \end{cases}
    \quad
    b(x) = \frac{a(x)}{a(x) + a(1-x)},  
    \quad
    c(x)  = b(2+x) b(2-x).
  \end{align*}
  One can verify that $c \in C^\infty(\R)$, $c(x) = 1$ for $|x| \leq 1$,
  $c(x) = 0$ for $|x| \geq 2$, and $|c'(x)| \leq 2$ for all $x \in \R$.
  We now define our partition of unity as $j(x) = \sin(\frac\pi2 c(x))$, $\bar j(x) = \cos(\frac\pi2 c(x))$. By the chain rule we have $|j'(x)|, |\bar j'(x)| \leq \pi \leq 5$.
\end{proof}

\begin{figure}[ht]
  \centering
  \pgfmathdeclarefunction{a}{1}{%
  \pgfmathparse{(#1>0) ? exp(-1/#1) : 0}%
}
\pgfmathdeclarefunction{b}{1}{%
  \pgfmathparse{a(1-#1)/(a(#1)+a(1-#1))}%
}
\pgfmathdeclarefunction{j}{1}{%
  \pgfmathparse{sin(180/2*(b(#1-1)+b(-#1-1)-1))}%
}

\begin{tikzpicture}
  \begin{axis}[
    name=myaxis,
    width=13cm, height=5cm,
    xmin=-3, xmax=3, ymin=-0.15, ymax=1.1,
    ticklabel style={
      fill=white, 
      inner sep=1pt, 
      font=\small 
    },
    xlabel={$x$},
    xtick={0,1,2},
    xticklabels={$0$, $\gamma^{-2/5}$, $2\gamma^{-2/5}$},
    ytick={0},
    yticklabels={$0$},
    axis lines=middle,
    domain=-2.5:2.5, samples=501, clip=false
  ]
    \addplot[thick] {j(x)};
    \addplot[dashed] coordinates {(-1,0) (-1,1.1)};
    \addplot[dashed] coordinates {( 1,0) ( 1,1.1)};
    \addplot[dotted] coordinates {(-2,0) (-2,1.1)};
    \addplot[dotted] coordinates {( 2,0) ( 2,1.1)};
    
    \node[font=\small, fill=white, inner sep=1pt] 
          at (axis cs:-.1,-.1) {$0$};
    \node[font=\small, fill=white, inner sep=1pt] 
          at (axis cs:-.1,.9) {$1$};
    \node[font=\small, fill=white, inner sep=1pt] 
          at (axis cs:.5,1) {$j(x)$};
  \end{axis}
    
\end{tikzpicture}
  \caption{The bump function $j$ from \cref{lemma:j and barj properties}.}
  \label{fig:bump}
\end{figure}
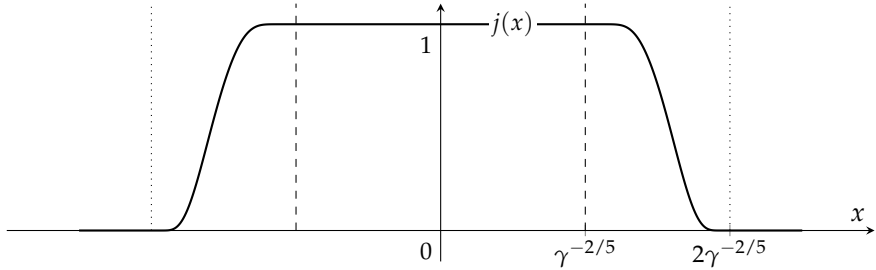

We will also be using the following helper inequality.
\begin{lemma} \label{lem:log-helper}
Let $\alpha \geq 1$ and $y \geq e^e$.
There exists a constant $c$ only dependent on $\alpha$ such that
\[
x \geq c y \log^\alpha(y)
\quad \Rightarrow \quad
x \geq y \log^\alpha(x).
\]
\end{lemma}
\begin{proof}
First, note that $x/\log^\alpha(x)$ is monotonically increasing for $x \geq e^\alpha$, and that $\log(y) \geq \log\log(y) \geq 1$ for $y \geq e^e$.
Second, let the constant $c > 1$ be such that
\[
c
\geq (\log(c) + 1 + \alpha)^\alpha.
\]
Then, note that for any $y \geq e^e$ and $x \geq c y \log^\alpha(y)$ (and so $x \geq e^\alpha$) we have that
\begin{align*}
\frac{x}{\log^\alpha(x)}
&\geq \frac{c y \log^\alpha(y)}{\left(\log(c) + \log(y) + \alpha \log\log(y) \right)^\alpha} \\
&\geq \frac{c y \log^\alpha(y)}{\left(\log(c) \log(y) + \log(y) + \alpha \log(y) \right)^\alpha}
= \frac{cy}{(\log(c) + 1 + \alpha)^\alpha}
\geq y. \qedhere
\end{align*}
\end{proof}

\section{Semiclassical analysis: Euclidean setting}
\label{sec:Euclidean}

Let $f\colon \domain \to \R$ be a self-concordant function on a bounded convex open subset $\domain$ of $\R^n$. We study the Schr\"odinger operator
\[
\H = -\frac12 \Delta + \semi^2 f,
\]
 where we impose vanishing (Dirichlet) boundary conditions. It is a well-studied operator, see, e.g.,~\cite{andrews11fundamentalgap,Davies1989}. It is an unbounded operator on $L^2(\domain)$ whose domain contains $C_c^\infty(M)$. Its spectrum is \textit{discrete} and its eigenfunctions are smooth (cf.~\cite[Sec.~6.5, 6.3, Ex.~6.2]{EvansPDE}). That is, the spectrum of $\H$ consists of a sequence of eigenvalues
\[
\lambda_0(\H) < \lambda_1(\H) < \lambda_2(\H) < \ldots,
\]
such that~$\lambda_k(\H) \to \infty$ as~$k \to \infty$,
each eigenvalue has finite multiplicity, and each eigenfunction is smooth. 
 
Formally defining the operator is standard, but subtle. For completeness, we give a short overview here. We define $\H$ as the operator associated to the form sum of $-\Delta$ and $\semi^2 f$. Here $\semi^2 f$ is viewed as the form $(u,v) \mapsto \langle u,\semi^2f v\rangle$, whose form domain is the set of $u \in L^2(M)$ for which $\sqrt{|f|}u \in L^2(M)$. Similarly, $-\Delta$ is the Dirichlet Laplacian: the operator associated to the symmetric form $(u,v) \mapsto \langle \nabla u, \nabla v\rangle$ whose form domain is $W_0^{1,2}(M)$, the closure of $C_c^\infty(M)$ with respect to the $W^{1,2}$-norm $\|u\|_{W^{1,2}}^2 = \langle u,u\rangle + \langle \nabla u, \nabla u\rangle$.\footnote{When $M$ has a sufficiently smooth boundary (e.g.~bounded and $C^1$), the set $W_0^{1,2}(M)$ coincides with the $u \in W^{1,2}(M)$ that vanish on the boundary of $M$ in the sense that their trace equals zero, cf.~\cite[Sec.~5.5]{EvansPDE}.} Using e.g.~\cite[Theorem~1.8.2]{Davies1989}, one obtains the form domain of $\H=-\Delta+\semi^2 f$, from which one can derive the operator domain of $\H$, which we will denote with $\Dom(\H)$:
\[
\Dom(\H) = \{u \in W_0^{1,2}(M) : (-\Delta+ \semi^2 f)u \in L^2(M)\}.
\]

For the purpose of determining the spectrum of $\H$, after translation, we may assume $0 \in \R^n$ is the minimizer of $f$ and $f(0)=0$. 
Let $\nabla^2 f(0) = A$ and note that $A$ is positive definite (cf.~\cref{d:p1}). 
Write
\begin{equation*}
  q(x) = \frac12 x^T A x.
\end{equation*}
We use the harmonic approximation
\begin{equation*}
  \H_0 = - \frac{1}{2}\Delta + \semi^2 q,
\end{equation*}
whose spectrum is well known. In particular, its spectrum is (purely) discrete and the smallest two points are $\lambda_0(\H_0) = \frac{\semi}{2} \Tr[\sqrt{A}]$ and $\lambda_1(\H_0) = \frac{\semi}{2}(\Tr[\sqrt{A}] + 2 \sqrt{\lambda_{\min}(A)})$. Throughout, we let $\ket{\psi_0}$ (or $\psi_0$) denote the eigenstate of $\H_0$ corresponding to $\lambda_0(\H_0)$.

The goal of this section is to prove the following theorem.

\begin{restatable}[Spectral gap of $-\frac{1}{2}\Delta + \semi^2 f$]{theorem}{euclideangap}\label{thm:mainEuclidean}
    Let $\H = -\frac12 \Delta + \semi^2 f$ for a self-concordant function~$f$.
    Then there exists a universal constant~$c>0$ such that if $\semi$ satisfies 
    \[
    \semi \geq c \, \max\{(n\sqrt{1+\|A\|})^5 \log^5(n(1+\|A\|)), (\|A\|/\lambda_{\min}(\sqrt{A}))^5\},\] 
    then we have
    \[
    \gap(\H)  \geq \frac{\semi}{2} \lambda_{\min}(\sqrt{A}).
    \]
\end{restatable}

Notably, the semiclassical parameter~$\semi$ that is needed to establish a quantitative spectral gap depends polynomially on the conditioning of~$A$, and thus~$\semi$ cannot in general be treated as a constant.

\subsection{Proof strategy}
We now explain the proof strategy.
We will analyze $\H$ by ``localizing'' it on the region where $q$ is a good approximation of $f$. Specifically, we will use the so-called IMS-localization formula, named after Ismagilov, Morgan and Simon, and popularized by I.M.~Sigal. We refer to, e.g., \cite[Sec.~3]{CyFKS1987} for the version below.
In \cref{sec:Riemannian} we will state (and prove) a version on Riemannian manifolds.

\begin{lemma}[IMS localization formula] \label{lem:IMS}
Let $J, \bar J \in C^\infty(\R^n)$ be such that $J^2 + \bar J^2 = I$.
Then
\[
\H
= J \H J + \bar J \H \bar J -  \|\nabla J\|^2 -  \|\nabla\bar{J}\|^2.
\]
\end{lemma}

Using the IMS localization formula we can write
\begin{align} \label{eq:decomp}
\H 
&= J \H_0 J + J (\H-\H_0) J + \bar J \H \bar J -  \|\nabla J\|^2- \|\nabla \bar J\|^2.
\end{align}

Our general proof strategy is now simple: we construct a function $J$ for which we have good control over the terms on the right hand side of \cref{eq:decomp}. In particular, in the next section we will construct a function $J$ that has the following properties:
\begin{itemize}
    \item inner region: $\|J(\H-\H_0)J\| \leq O(\semi^{4/5})$, see \cref{lem:HH0onJ}.
    \item boundary: $\|\nabla J\|^2, \|\nabla \bar J\|^2 \leq O(\|A\| \semi^{4/5})$, see \cref{lem:gradJ}.
    \item outer region: $\bar J \H \bar J \succeq \frac14\semi^{6/5}(\bar J)^2$ for $\semi=\Omega(1)$, see \cref{lem:JHJ}. 
    \item ground state: $\|\bar J  \psi_0\|^2 \leq \exp(-ctn)$ for some universal constant $c>0$ if $\semi \geq \big(t n \sqrt{\|A\|}\big)^5$, see \cref{lem:GaussianJbar}.
\end{itemize}
The first two items show that the~$J(\H - \H_0)J$, $\norm{\nabla J}^2$, $\norm{\nabla \bar J}^2$ terms in~\cref{eq:decomp} are of order~$\gamma^{4/5}$.
By contrast, the spectral gap of~$\H_0$ scales like~$\gamma$, hence we expect these terms to be irrelevant.
The other terms must be dealt with separately for the lower bound on~$\lambda_1(\H)$ and the upper bound on~$\lambda_0(\H)$, which we do in~\cref{subsec:euclidean lambda1,subsec:euclidean lambda0}, respectively.

\subsection{Construction of the localization}
\label{subsec:euclidean construction localization}
We will construct a localization function~$J$ in the following way.
Given $\semi>0$, let $j,\bar j \in C^\infty(\R)$ be as in \cref{lemma:j and barj properties}. That is, $j^2+\bar j^2=1$ and $j(t) = 1$ for $\abs{t} \leq \semi^{-2/5}$, $j(t) = 0$ for $\abs{t} \geq 2 \semi^{-2/5}$, and such that $\abs{j'(t)}, \abs{\bar{j}'(t)} \leq 5\semi^{2/5}$.
We then consider $J \in C_c^\infty(\R^n)$ defined by
\begin{equation} \label{eq:JdefEuclidean}
J(x)
= j(\sqrt{x^T A x})
= j(\norm{x}_A).
\end{equation}

To show that $\H_0$ is a good approximation of $\H$ when we restrict to the support of $J$, we observe that $J (\H - \H_0) J$ is the multiplication operator~
$\semi^2 J (f - q) J$.
The pointwise estimate from the corollary below is thus a bound on its operator norm $\norm{\semi^2 J (f - q) J}$. The difference between $f$ and $q$ is precisely the difference between $f$ and its second-order Taylor expansion around the minimizer. This difference can be expressed in terms of the third-derivative of $f$. Due to the localization, we need to bound this third-derivative for all points that are close to the minimizer with respect to the norm induced by the second-derivative. Self-concordance allows one to get a clean bound.

\begin{lemma}[{Third-order bound~\cite[Theorem 2.2.2]{Renegar01}}] \label{lem:third}
Let $f$ be self-concordant, $x \in \domain$, and let $q_x$ be the quadratic Taylor approximation of $f$ at $x$. Then, for $y$ such that $\|y-x\|_x <1$, we have
\[
\abs{f(y)-q_x(y)} \leq \frac{\norm{y-x}_x^3}{3(1-\norm{y-x}_x)}
\]
\end{lemma}
\begin{corollary}[Inner region]\label{lem:HH0onJ}
  Let $f$ be self-concordant and assume $0<2 \semi^{-2/5} \leq \frac12$. Then we have, for all $x \in \R^n$,
  \[
  \semi^2 \abs{ J(x) (f(x) - q(x)) J(x)} \leq \frac{8}{3} \semi^{4/5}.
  \]
\end{corollary}
\begin{proof}
    For~$y \in \supp J$, we have~$\norm{y-x}_x \leq 2 \gamma^{-2/5} \leq \frac12$, and so~$\abs{f(y) - q_x(y)} \leq 8 \gamma^{-6/5}/3$.
\end{proof}

\begin{lemma}[Boundary] \label{lem:gradJ}
    For any $x \in \R^n$, we have
    \[
    \norm{\nabla J(x)}^2, \norm{\nabla \bar J(x)}^2 \leq \frac{\norm{A}}{4} \cdot C \semi^{4/5},
    \]
    where $C>0$ is a universal constant. 
\end{lemma}
\begin{proof}
    We give the proof for~$J$; the proof for~$\bar J$ is analogous, as we assume the same estimates on the derivatives of~$\bar j$ as we do for~$j$.
    We have $\nabla J(x) = D j(\norm{x}_A) \cdot A x / (2 \sqrt{x^T A x})$, and so
    \begin{equation*}
      \norm{\nabla J(x)}^2 = \abs{D j({\norm{x}_A})}^2 \frac{\norm{Ax}^2}{4 \norm{x}_A^2} = \abs{D j({\norm{x}_A})}^2 \frac{x^T A^{1/2} A A^{1/2} x}{4 x^T A x} \leq \frac{\norm{A}}{4} \cdot O(\semi^{4/5}).
    \end{equation*}
    The last inequality follows from the upper bound on the sup-norm of $D j$ from \cref{lemma:j and barj properties}.
\end{proof}

\begin{lemma}[Outer region] \label{lem:JHJ}
    For a self-concordant $f$ with $f(0)=0$ as minimum, $\H = -\frac12\Delta + \semi^2 f$, $J$ and $\bar J$ as above, and $\semi \geq \Omega(1)$, we have
    \[
    \bar J \H \bar J \succeq \frac14 \semi^{6/5} (\bar J)^2.
    \]
\end{lemma}
\begin{proof}
    We first note that $-\Delta \succeq 0$ and therefore $\bar J \H \bar J \succeq \semi^2 f (\bar J)^2$. It thus suffices to give a lower bound on $f$ on the support of $\bar J$. The support of $\bar J$ is contained in the set $\{x : \|x\|_A \geq \semi^{-2/5}\}$ where $A = \nabla^2 f(0)$.
    Consider the function $\rho(r) = - r - \ln(1-r)$.
    Then $\rho(r) = \frac12 r^2 + O(r^3)$ for small $r$, and $\rho(r) \geq \frac14 r^2$ when $\abs{r} \leq 1$.
    When $1 > \norm{x}_A$, self-concordance of $f$ yields
    \begin{equation*}
      f(x) \geq f(0) + \rho(-\norm{x}_A) \geq f(0)+ \frac14 \norm{x}_A^2,
    \end{equation*}
    see e.g.~\cite[Theorem~5.1.8]{nesterov2018lectures}.
    Due to convexity of $f$, this implies  $f(x)-f(0) \geq \frac14$ whenever $\|x\|_A \geq 1$ as well.
    In particular this shows  $\semi^2 f(x) \geq \frac14 \semi^{6/5}$ whenever $\norm{x}_A \geq \semi^{-2/5}$.
\end{proof}

\subsection{Lower bound on \texorpdfstring{$\lambda_1(-\frac12 \Delta + \gamma^2 f)$}{lambda1(H)}}
\label{subsec:euclidean lambda1}
We now use the properties of the localization function to prove a lower bound on~$\lambda_1(\H)$.
Namely, \cref{lem:HH0onJ,lem:gradJ,lem:JHJ} imply that
\begin{align*}
  \H & = J \H_0 J &&+ J (\H - \H_0) J   &&+ \bar J \H \bar J &&-  \norm{\nabla J}^2 - \|\nabla \bar J\|^2 \\
  & \succeq J \H_0 J  &&- \frac{8}{3}\semi^{4/5}  &&+  \frac14\semi^{6/5} (\bar J)^2 &&- O(\|A\| \semi^{4/5}).
  \end{align*}
To lower bound the first term, we can use the inequality
\begin{align*}
  \H_0 & \succeq \lambda_1(\H_0) I - (\lambda_1(\H_0)-\lambda_0(\H_0)) \ket{\psi_0}\bra{\psi_0} \\
       & = \frac{\semi}{2} \left((\Tr[\sqrt{A}] + 2 \lambda_{\min}(\sqrt A)) I - 2\lambda_{\min}(\sqrt A)\ket{\psi_0}\bra{\psi_0}\right).
\end{align*}
This gives us
\[
\H \succeq \frac{\semi}{2} (\Tr[\sqrt A] + 2 \lambda_{\min}(\sqrt A)) J^2 - \semi \lambda_{\min}(\sqrt A) J \ket{\psi_0}\bra{\psi_0} J 
+  \semi^{6/5} (\bar J)^2 - O((\|A\|+1) \semi^{4/5}).
\]
Whenever 
\begin{equation} \label{eq:LBlambda}
\semi \geq \left[ \frac{1}{2} (\Tr[\sqrt A]+2\lambda_{\min}(\sqrt A)) \right]^5,
\end{equation}
we can further combine the terms containing $J^2$ and $\bar J^2$,
using $J^2 + \bar J^2 = 1$,  to obtain 
\begin{equation} \label{eq:H-LB}
\H  \succeq \frac{\semi}{2} (\Tr[\sqrt A] + 2 \lambda_{\min}(\sqrt A)) I - \semi \lambda_{\min}(\sqrt A) J \ket{\psi_0}\bra{\psi_0} J - O((\|A\|+1) \semi^{4/5}).
\end{equation}
Using that $J \ket{\psi_0}\bra{\psi_0} J \preceq I$, this implies that
\[
\H
\succeq \frac{\semi}{2} \Tr[\sqrt A] - O((\|A\|+1) \semi^{4/5}).
\]
and so
\begin{equation} \label{eq:lambda0-lowerb}
\lambda_0(\H)
\geq \frac{\semi}{2} \Tr[\sqrt A] - O((\|A\|+1) \semi^{4/5}).
\end{equation}
Towards the upper bound, we notice that \eqref{eq:H-LB} is a lower bound of the form
\begin{equation} \label{eq:op-lb}
\H
\succeq \frac{\gamma}{2} (\Tr[\sqrt A] + 2 \lambda_{\min}(\sqrt A)) I + F - O((\|A\|+1) \semi^{4/5}),
\end{equation}
where $F$ has rank 1.
This proves that $\H$ can only have a single eigenvalue below $\frac{\gamma}{2} (\Tr[\sqrt A] + 2 \lambda_{\min}(\sqrt A)) - O((\|A\|+1) \semi^{4/5})$. In other words, it gives a lower bound on $\lambda_1(\H)$:
\begin{theorem}
\label{thm:euclidean lambda1 lower bound}
    For a self-concordant function $f$ and $\semi \geq \left[ \frac{1}{2} (\Tr[\sqrt A]+2\lambda_{\min}(\sqrt A)) \right]^5$ 
    we have
    \[
    \lambda_1(\H) \geq \lambda_1(\H_0)  -  O(\|A\| \semi^{4/5}) = \frac{\gamma}{2} (\Tr[\sqrt A] + 2 \lambda_{\min}(\sqrt A)) - O(\|A\| \semi^{4/5}).
    \]
\end{theorem}

\subsection{Upper bound on \texorpdfstring{$\lambda_0(-\frac12 \Delta + \gamma^2 f)$}{lambda0(H)}}
\label{subsec:euclidean lambda0}
To prove an upper bound on $\lambda_0(\H)$, it suffices to construct a test function $\psi$ which has small energy. In particular, we want it to satisfy the following inequality:
\[
\frac{\langle \psi, \H \psi \rangle}{\langle \psi ,\psi\rangle}
\leq \lambda_0(\H_0) + O(\|A\| \semi^{4/5})
= \frac{\gamma}{2} \Tr[\sqrt A] + O(\|A\| \semi^{4/5}).
\]
To do so, we will use $\psi = J \psi_0$ where $\psi_0$ is the (known) ground state of $\H_0$.
For convenience, we recall that $\H_0 = - \frac{1}{2}\Delta +  \frac{1}{2} \semi^2 x^T A x$ corresponds to a quantum harmonic oscillator with ground state
\[
\psi_0(x) = 
C_0 \exp\left(-\frac{\gamma}{2} x^T \sqrt{A} x\right),
\]
where $C_0$ is the normalization constant for which $\langle \psi_0,\psi_0 \rangle =1$.

Concretely, we establish the following theorem.
\begin{theorem}
    \label{thm:euclidean lambda0 upper bound}
    For a self-concordant function $f$ and $\semi \geq \Omega(( n\sqrt{\|A\|})^5 \cdot \polylog(n,\|A\|))$, we have
    \[
    \lambda_0(\H) \leq \lambda_0(\H_0) +  O(\|A\| \semi^{4/5})
    = \frac{\gamma}{2} \Tr[\sqrt A] + O(\|A\| \semi^{4/5}).
    \]
\end{theorem}
Together with~\cref{thm:euclidean lambda1 lower bound} and \cref{eq:lambda0-lowerb}, this will imply \cref{thm:mainEuclidean}.
Moreover, if additionally $\semi \in \Omega((\|A\|/\lambda_{\min}(\sqrt{A}))^5)$ then
\[
\lambda_1(\H) - \lambda_0(\H)
\geq \frac{\gamma}{2} \lambda_{\min}(\sqrt A). 
\]
We start by proving a lemma on ground state concentration.
\begin{lemma}[Ground state of~$\H_0$: concentration] \label{lem:GaussianJbar}
  There exists some universal constant $c_0>0$ such that for $t \geq 1$ and for $\semi \geq \big(t n \sqrt{\|A\|}\big)^5$, we have $\langle J \psi_0, J\psi_0 \rangle \geq 1 - 2 \exp(-c_0tn)$.
\end{lemma}
\begin{proof}
    Using the identity $J^2  +\bar J^2 = 1$, we have
    \[
    \langle J \psi_0, J\psi_0 \rangle = \langle \psi_0, \psi_0\rangle - \langle \bar J \psi_0, \bar J \psi_0 \rangle
    \]
    and so we can equivalently prove an upper bound on $\langle \bar J \psi_0, \bar J \psi_0 \rangle$.
    We will use that
    \[
    \langle \bar J \psi_0, \bar J \psi_0 \rangle
    \leq \psi_0^2(\sup \bar J),
    \]
    so that it suffices to bound the probability of $\sup \bar J$ under the measure $\psi_0^2$.
    Since the measure $\psi_0^2$ corresponds to the high-dimensional Gaussian measure $\mathcal{N}(0,(2\gamma\sqrt{A})^{-1})$, we can use standard concentration results.
    E.g., from the Hanson-Wright inequality (\cite[Thm.~6.2.1]{vershynin2018high}) we can derive that if $X \sim \mathcal{N}(0,\Sigma)$ then there exists a universal constant $c_0 > 0$ such that
    \begin{equation} \label{eq:gaussian-conc}
    \mathbb{P}\left( \|x\|_{\Sigma^{-1}}^2 \geq (t+1) n \right) \leq 2 \exp(-c_0 t n), \quad \forall t \geq 1.
    \end{equation}

    Recalling that on the support of $\bar J$ we have $\|x\|_A \geq \semi^{-2/5}$, we get that
    \[
    \psi_0^2(\sup \bar J)
    \leq \mathbb{P}(\|x\|_A \geq \gamma^{-2/5})
    \]
    for $x \sim \mathcal{N}(0,(2\gamma\sqrt{A})^{-1})$.
    We can bound this by noting that
    \[
    \|x\|_{2 \semi \sqrt{A}}^2
    = 2\semi x^T \sqrt{A} x
    \geq 2 \semi x^T A x / \sqrt{\|A\|}
    = 2\semi \|x\|_A^2 / \sqrt{\|A\|},
    \]
    and so
    \[
    \mathbb{P}(\|x\|_A \geq \gamma^{-2/5})
    \leq \mathbb{P}(\|x\|^2_{2\gamma\sqrt{A}} \geq 2 \gamma^{1/5}/\sqrt{\|A\|})
    \]
    For $x \sim \mathcal{N}(0,(2\gamma\sqrt{A})^{-1})$ and $2 \gamma^{1/5}/\sqrt{\|A\|} \geq (t+1) n$, and using \eqref{eq:gaussian-conc}, this is bounded by $2 \exp(-c_0 t n)$.
\end{proof}

We now state two technical lemmas that will allow us to estimate the ground energy of $\H$ in \cref{lem:UBlaplace}.
\begin{lemma} \label{lem:Delta psi0}
  There exist some universal constants $c_0,c_1>0$ such that for $t \geq 1$ and for $\semi \geq \big(t n \sqrt{\|A\|}\big)^5$, we have
    \[
    \int_{\R^n} \|\nabla \psi_0(x)\|^2 \bar J(x)^2 dx
    = \frac{\gamma^2}{4} \int_{\R^n} (x^T A x) \psi_0(x)^2 \bar J(x)^2 dx
    \leq c_1 \semi \sqrt{\|A\|} n \exp(-c_0 t n). \label{eq:int nabla psi0}
    \]
\end{lemma}
\begin{proof}
We first note that $\|\nabla \psi_0(x)\|^2 = \frac{\semi^2}{4} x^T A x \cdot \psi_0(x)^2 \leq \frac{\semi^2}{4} \|A\|^{1/2} x^T \sqrt{A} x \cdot \psi_0(x)^2$. Then, using the Cauchy-Schwarz inequality, we have
\begin{align*}
    \int_{\R^n} x^T \sqrt{A} x \cdot \psi_0(x)^2 \bar J(x)^2 dx &\leq \left(\int_{\R^n} (x^T \sqrt{A} x)^2 \cdot \psi_0(x)^2 dx\right)^{1/2} \left(\int_{\R^n} \bar J(x)^4  \psi_0(x)^2 dx\right)^{1/2}.
\end{align*}
We bound the two factors separately. First, substituting $z = (2\gamma \sqrt{A})^{1/2} x$, we can write the first factor as $\mathbb E[(2\gamma)^4 (z^T z)^2]$ where $z \sim \mathcal N(0,I)$. Using $\mathbb E [z_i^2] = 1$ and $\mathbb E [z_i^4] = 3$, this gives
\[
\left(\int_{\R^n} (x^T \sqrt{A} x)^2 \cdot \psi_0(x)^2 dx\right)^{1/2}
= \mathbb{E}_x[\|x\|_{\sqrt{A}}^4]^{1/2}
= (2\gamma)^{-1} \mathbb{E}[\|z\|^4]^{1/2}
\leq (2\gamma)^{-1} \sqrt{n^2 + 2n}.
\]
For the remaining factor we use that $\bar J(x)^2 \leq 1$ to obtain
\[
\left(\int_{\R^n} \bar J(x)^4  \psi_0(x)^2 dx\right)^{1/2} \leq \sqrt{\langle \bar J \psi_0, \bar J \psi_0 \rangle} \leq \sqrt{2} \exp(-c_0 t n/2),
\]
where the last inequality uses \cref{lem:GaussianJbar}. 
\end{proof}

\begin{lemma} \label{lem:H0estimate}
  There exists some universal constant $c_0>0$ such that for $t \geq 1$ and for $\semi \geq \big(tn\sqrt{\|A\|}\big)^5$, we have
    \[
    \langle J^2 \psi_0, \H_0 J^2 \psi_0\rangle = \lambda_0(\H_0) + O((\semi^{4/5} \|A\| + \semi \sqrt{\|A\|} n) \exp(-c_0 tn)).
    \]
    In particular, we have
    \[
    \langle \bar J^2 \psi_0, -\Delta \bar J^2 \psi_0 \rangle = O((\semi^{4/5} \|A\| + \semi \|A\|^{1/2} n) \exp(-c_0 tn)).
    \]
\end{lemma}
\begin{proof}
    Using the identity $1 = J^2 + \bar J^2$, we have
    \begin{align*}
        \langle J^2 \psi_0, \H_0 J^2 \psi_0\rangle &=  \langle (1-\bar J^2) \psi_0, \H_0 (1-\bar J^2) \psi_0\rangle \\
        &= \langle \psi_0, \H_0 \psi_0 \rangle -2 \langle \bar J^2 \psi_0, \H_0  \psi_0 \rangle  + \langle \bar J^2 \psi_0, \H_0 \bar J^2  \psi_0\rangle \\
        &= \lambda_0(\H_0) -2 \lambda_0(\H_0) \langle \bar J^2 \psi_0,   \psi_0 \rangle  + \langle \bar J^2 \psi_0, \H_0 \bar J^2  \psi_0\rangle.
    \end{align*}
    We now note that $\langle \bar J^2 \psi_0,   \psi_0 \rangle \leq 2\exp(-c_0tn)$ by \cref{lem:GaussianJbar}. It thus remains to analyze $\langle \bar J^2 \psi_0, \H_0 \bar J^2  \psi_0\rangle$. For this we recall that $\H_0 = - \frac{1}{2} \Delta +  \frac{1}{2} \semi^2 x^T A x$ and hence
    \[
    \langle \bar J^2 \psi_0, \H_0 \bar J^2  \psi_0\rangle = \langle \bar J^2 \psi_0, -\frac{1}{2} \Delta \bar J^2  \psi_0\rangle + \semi^2 \langle \bar J^2 \psi_0,  (\frac{1}{2} x^T A x) \bar J^2  \psi_0\rangle.
    \]
    We now compute the two terms separately. First, we have
    \[
        \semi^2 \langle \bar J^2 \psi_0,  (\frac{1}{2} x^T A x) \bar J^2  \psi_0\rangle \leq \semi^2 \langle \bar J \psi_0,  (\frac{1}{2} x^T A x) \bar J  \psi_0\rangle
        \leq 2 c_1 \gamma \sqrt{\|A\|} n \exp(-c_0 t n),
    \]
    where the last inequality is due to \cref{lem:Delta psi0}. For the remaining term, we recall that for~$\phi \in C_c^\infty(\domain)$, we have
    \begin{equation}
        \label{eq:Delta to grad}
    \langle \phi, -\Delta \phi \rangle = \int_{\R^n} \|\nabla \phi(x)\|^2 dx.
    \end{equation}
    Using this we can rewrite
    \begin{align*}
        \langle \bar J^2 \psi_0, -\Delta \bar J^2  \psi_0\rangle
        &= \int_{\R^n} \| \nabla(\bar J^2 \psi_0)(x) \|^2 dx \\
        &= \int_{\R^n} \| 2\bar J(x)\psi_0(x) \nabla \bar J(x) + \bar J^2 \nabla \psi_0(x) \|^2  dx \\
        &\leq \int_{\R^n} 2\left(\| 2\bar J(x)\psi_0(x) \nabla \bar J(x)\|^2 + \|\bar J(x)^2 \nabla \psi_0(x) \|^2\right)  dx \\
        &=\int_{\R^n} 8\| \nabla \bar J(x)\|^2 \bar J(x)^2\psi_0(x)^2 dx  + \int_{\R^n} \|\nabla \psi_0(x) \|^2 \bar J(x)^2 dx
        \end{align*}
    Using \cref{lem:gradJ,lem:GaussianJbar} we can upper bound the first term by $O(\semi^{4/5} \|A\| \exp(-c_0 tn))$. \cref{lem:Delta psi0} upper bounds the second term by $O(\semi \sqrt{\|A\|} n \exp(-c_0 tn))$.
\end{proof}

We can now show the main claim of this section, proving that $J\psi_0 /\|J\psi_0\|$ has low energy with respect to $\H$.

\begin{lemma}[Ground energy comparison] \label{lem:UBlaplace}
Let $\psi_0$ be the ground state of $\H_0$ and $J$ as above.
There exist universal constants $c,C>0$ such that if $\semi \geq c (n \sqrt{1+\|A\|})^5 \log^5(n (1+\|A\|))$, then 
\[
\frac{\langle J\psi_0, \H J\psi_0 \rangle}{\langle J\psi_0 ,J\psi_0\rangle}
\leq \lambda_0(\H_0) + C(\|A\| + 1)\semi^{4/5}.
\]
\end{lemma}
\begin{proof}
For a fixed $t \geq 1$ (which we will specify later), assume that $\gamma \geq (t n \sqrt{\|A\|})^5$.
Then by \cref{lem:GaussianJbar} we have $\langle J\psi_0,J\psi_0 \rangle \geq 1 - 2 \exp(-c_0tn)$.  We thus have $\frac{\langle \psi, \H \psi \rangle}{\langle \psi,\psi\rangle} \leq  (1+4\exp(-c_0 t n)) \langle J \psi_0, \H J \psi_0 \rangle$ (assuming $2 \exp(-c_0tn) \leq 1/2$).
    We now use IMS-localization (in particular, \cref{eq:decomp}) again to obtain
    \begin{align*}
  \langle J \psi_0, \H J \psi_0 \rangle
&= \langle J \psi_0, \left(J (\H-\H_0) J + J \H_0 J + \bar J \H \bar J -  \|\nabla J\|^2- \|\nabla \bar J\|^2 \right) J\psi_0\rangle \\
&= c_0 \semi^{4/5} + \langle J \psi_0, J \H_0  JJ \psi_0 \rangle + \langle J\psi_0, \bar J \H \bar J J \psi_0 \rangle - c_1 \|A\| \semi^{4/5}.
    \end{align*}
    where $c_0$ and $c_1$ are the constants from \cref{lem:HH0onJ} and \cref{lem:gradJ} respectively. It remains to analyze $\langle J \psi_0, J \H_0  JJ \psi_0 \rangle$ and $\langle J\psi_0, \bar J \H \bar J J \psi_0 \rangle$.
    By \cref{lem:H0estimate} we have $\langle J \psi_0, J \H_0  JJ \psi_0 \rangle \leq \lambda_0(\H_0) + O((\semi^{4/5} \|A\| + \semi \sqrt{\|A\|} n) \exp(-c_0 t n))$.
    We finally upper bound the remaining term $\langle J\psi_0, \bar J \H \bar J J \psi_0 \rangle$. Using the definition of $\H$, we have
    \[
    \langle J\psi_0, \bar J \H \bar J J \psi_0 \rangle
    = \langle J\psi_0, \bar J (-\frac{1}{2} \Delta) \bar J J \psi_0 \rangle
    + \gamma^2 \langle J\psi_0, \bar J f \bar J J \psi_0 \rangle.
    \]
    For the first term we use \cref{eq:Delta to grad} and the identity \[
    \nabla(\bar J(x) J(x) \psi_0(x)) = J(x) \psi_0(x) \nabla \bar J(x) + \bar J(x) \psi_0(x) \nabla J(x) + J(x) \bar J(x) \nabla \psi_0(x)
    \]
    to obtain
    \begin{align*}
        \langle J\psi_0, \bar J (-\Delta) \bar J J \psi_0 \rangle
        &= \int_{\R^n} \| \nabla (\bar J(x) J(x) \psi_0(x)) \|^2 dx \\
        &\leq 3 \int_{\R^n} J(x)^2 \psi_0(x)^2 \|\nabla \bar J(x)\|^2 + \bar J(x)^2 \psi_0(x)^2 \|\nabla J(x)\|^2 + J(x)^2 \bar J(x)^2 \|\nabla \psi_0(x)\|^2 dx.
    \end{align*}
    By \cref{lem:gradJ} the first two terms are $O(\gamma^{4/5} \|A\|)$, while the last term is $O(\semi \sqrt{\|A\|} n \exp(-c_0 t n))$ by \cref{lem:Delta psi0}.
    This shows that
    \[
    \langle J\psi_0, \bar J (-\frac{1}{2}\Delta) \bar J J \psi_0 \rangle
    = O\left( \semi^{4/5}\|A\| + \gamma \sqrt{\|A\|} n \exp(-c_0 t n) \right).
    \]
    It remains to upper bound $\gamma^2 \langle J\psi_0, \bar J f \bar J J \psi_0\rangle$. \cref{lem:HH0onJ} shows
    \[
    \gamma^2 \langle J\psi_0, \bar J f \bar J J \psi_0\rangle
    \leq \gamma^2 \langle J\psi_0, \bar J (\frac{1}{2} x^T A x) \bar J J \psi_0\rangle + O(\semi^{4/5}).
    \]
    We next upper bound $\langle J\psi_0, \bar J (\frac{1}{2} x^T A x) \bar J J \psi_0\rangle$  by $O(\semi \sqrt{\|A\|} n \exp(-c_0 t n))$ using \cref{lem:Delta psi0}.

    Combining these estimates yields the final bound
    \[
    \frac{\langle J\psi_0, \H J\psi_0 \rangle}{\langle J\psi_0 ,J\psi_0\rangle}
    = O\left( \semi^{4/5}\|A\| + \gamma^2 \sqrt{\|A\|} n \exp(-c_0 t n) \right).
    \]
    For $\gamma \geq e$ we can set $t = C \, \log(\semi (1+\|A\|) n) \geq 1$ for a large enough constant~$C$, and obtain the claimed bound.
    A sufficient bound on $\semi$ is then
    \[
    \semi \geq \left( C n\sqrt{1+\|A\|} \log(\semi (1+\|A\|) n) \right)^5.
    \]
    Applying \cref{lem:log-helper} with $x = \gamma (1+\|A\|) n$, $\alpha = 5$ and $y = C^5 (1+\|A\|)^{7/2} n^6$ (which is $> e^e$ for sufficiently large $C$) shows that this is satisfied in particular when
    \[
    \semi
    \geq c (n \sqrt{1+\|A\|})^5 \log^5(n (1+\|A\|))
    \]
    for some universal constant $c$.
\end{proof}

As a corollary, we can prove a convenient bound on the overlap between the ground state of our Schrödinger operator and that of its harmonic approximation.
This will be useful for our later quantum path-following algorithm.
\begin{corollary}[Ground state overlap] \label{cor:eucl-gs-overlap}
Let $\psi$ and $\psi_0$ denote the ground states of $\H$ and $\H_0$, respectively.
For any $\delta > 0$,
if $\gamma \in \widetilde\Omega\left( \max\{ \log(1/\delta) n (1+\|A\|)^{5/2}, (1+\|A\|)^5/(\delta \lambda_{\min}(\sqrt{A}))^5 \} \right)$
then
\[
\abs{\braket{\psi,\psi_0}}
\geq 1 - \delta.
\]
\end{corollary}
\begin{proof}
By \cref{lem:GaussianJbar} we know that $\abs{\braket{\psi,\psi_0}} \geq \abs{\braket{\psi,J\psi_0}} - \delta/2$ if $\gamma \in \Omega((\log(1/\delta) tn \sqrt{\norm{A}})^5)$ for some $t \geq c \log(\semi \|A\| n)$ for a large enough universal constant~$c$.\footnote{Note that~$\braket{\psi_0, J \psi_0} \geq \braket{J \psi_0, J \psi_0}$, and so~$\abs{\braket{\psi,\psi_0 - J\psi_0}} \leq \norm{\psi_0 - J \psi_0} \norm{\psi} = \sqrt{1 + \norm{J \psi_0}^2 - 2 \braket{\psi_0, J \psi_0}} \leq \sqrt{1 - \norm{J \psi_0}^2}$. The latter quantity is upper bounded by~\cref{lem:GaussianJbar}. We now finish off with~$\abs{\braket{\psi, J \psi_0}} \leq\abs{\braket{\psi, \psi_0 - J \psi_0}} + \abs{\braket{\psi, \psi_0}}$.}
Again invoking \cref{lem:log-helper} as in the proof of \cref{lem:UBlaplace}, it suffices in particular that
\[
\gamma \in \Omega\left( \log(1/\delta) n (1+\|A\|)^{5/2} \log(n (1+\|A\|) \log(1/\delta)) \right).
\]
Now it suffices to prove that $\abs{\braket{\psi,J\psi_0}} \geq 1 - \delta/2$.
For this we use on the one hand \cref{lem:UBlaplace}, which implies that
\[
\braket{J\psi_0,\H J\psi_0}
\leq \frac{\braket{J\psi_0,\H J\psi_0}}{\braket{J\psi_0,J\psi_0}}
\leq \lambda_0(\H_0) + O((\norm{A}+1)\gamma^{4/5}).
\]
On the other hand we use that $\H \succeq \lambda_0(\H) \ket{\psi}\bra{\psi} + \lambda_1(\H) (I - \ket{\psi}\bra{\psi})$ and so
\[
\braket{J\psi_0,\H J\psi_0}
\geq \lambda_0(\H) \abs{\braket{J\psi_0,\psi}}^2 + \lambda_1(\H) (1 - \abs{\braket{J\psi_0,\psi}}^2).
\]
Combining these inequalities, and using that $\lambda_0(\H_0) = \lambda_0(\H) + O((\norm{A}+1) \gamma^{4/5})$ (\cref{eq:lambda0-lowerb}), we get that
\[
\abs{\braket{J\psi_0,\psi}}^2
\geq \frac{\lambda_1(\H) - \lambda_0(\H_0) - O((\norm{A}+1) \gamma^{4/5})}{\lambda_1(\H) - \lambda_0(\H)}
= 1 - O\left(\frac{(\norm{A}+1) \gamma^{4/5}}{\lambda_1(\H) - \lambda_0(\H)}\right).
\]
Finally, we use that for sufficiently large $\gamma$ we have $\lambda_1(\H)-\lambda_0(\H) \in \Omega(\gamma \lambda_{\min}(\sqrt{A}))$.
Taking $\gamma \in \Omega\left(\left(\frac{\|A\|+1}{\delta \lambda_{\min}(\sqrt{A})}\right)^5\right)$ we get the claimed conclusion.
\end{proof}

\section{Semiclassical analysis: Riemannian setting}
\label{sec:Riemannian}
Let $f\colon \domain \to \R$ again be a self-concordant function on an open bounded convex subset~$\domain$ of~$E = \R^n$. 
In this section we study the Schr\"odinger operator
\[ \H = -\frac{1}{2}\LB + \semi^2 f, \]
where instead of the ordinary Laplace operator we now consider the Laplace--Beltrami operator~$\LB$ with respect to the Riemannian metric on~$\domain$ induced by the Hessian of~$f$.
Our goal is to prove that this operator has a 
spectral gap when $\semi \geq C n^c$ for some universal constants $C,c>0$.
More precisely, we will show the following result:

\begin{restatable}[Spectral gap of $-\frac{1}{2}\LB + \semi^2 f$]{theorem}{riemanniangap}\label{thrm:gap}
    Let $\H = -\frac12 \LB + \semi^2 f$ for a self-concordant function~$f$.
    Then there exists a universal constant~$c>0$ such that if $\semi$ satisfies $\semi \geq c(n\log(n))^5$, then we have
    \[
    \gap(\H) \geq \frac\gamma2.
    \]
\end{restatable}

Unlike in the Euclidean case (\cref{sec:Euclidean}), there is no dependence on any condition number associated with~$f$.
Along the way, we also construct a state $\psi$ that has large overlap with the ground state of~$\H$.

This section is organized as follows.
We first introduce the necessary notation and definitions in \cref{sec:LBnotation}.
We review basic properties of the Laplace--Beltrami and Schrödinger operators in \cref{sec:LBS}.
In \cref{sec:IMSLB}, we state the version of the IMS localization formula that we use in the Laplace--Beltrami setting.
In \cref{sec:proof strategy Laplace-Beltrami} we outline the strategy that we will use to prove a lower bound on the spectral gap of~$\H$.
We establish a suitable local approximation of~$\H$ by a Laplace operator in \cref{sec:HvsH0}.
We then lower bound $\lambda_1(H)$ in \cref{sec:LBlower} and we upper bound $\lambda_0(\H)$ in \cref{sec:LBupper}.
In \cref{sec:LBgap} these are combined to establish \cref{thrm:gap}.

\subsection{Riemannian geometry}\label{sec:LBnotation}

Because $\domain$ is an open subset of the Euclidean space $E = \R^n$, it naturally has the structure of an $n$-dimensional oriented manifold.
The tangent space at any point $x\in\domain$ can be identified as~$T_x \domain = E$.
We briefly recall basic notions of differential and Riemannian geometry, and refer the reader to~\cite{leeIntroductionSmoothManifolds2012,leeIntroductionRiemannianManifolds2018} for more information. 

\paragraph{Metric induced by a convex function.}
Given a convex function~$f\colon \domain \to \R$ with positive definite Hessian~$H(x)$ (with respect to the coordinates of~$E = \R^n$), we can use the local inner product~\eqref{eq:local inner product} to define a Riemannian metric on~$\domain$:
\begin{align*}
  g_x(u, v) = \braket{u,v}_x = \braket{u, H(x) v}_E = D^2 f_x[u,v]
\quad \text{for}~u, v \in T_x \domain = E,
\end{align*}
If~$f$ is a self-concordant function and~$\domain$ is relatively compact in~$E$, then the resulting manifold is complete~\cite[Thm.~2.1]{nesterov2002geometry}.\footnote{The result is stated for self-concordant barriers, but actually only depends on ``strongly non-degenerate'' self-concordance as in~\cref{d:p1}.}

\paragraph{Tangent and cotangent vectors.}
The metric $g$ allows us to identify tangent vector and cotangent vectors.
Any tangent vector $u  \in T_x \manifold$ determines a cotangent vector $\ell = g_x(u,\cdot) \in T_x^* \manifold$, and vice versa.
In this way, the metric $g$ induces a dual inner product on the cotangent spaces, which we will denote by $g_x^*\colon T_x^* \manifold \times T_x^* \manifold \to \R$.
It is defined such that if $\ell = g_x(u,\cdot)$ and $\ell' = g_x(v,\cdot)$, then $g_x^*(\ell, \ell') = g_x(u,v)$.
There is also a variational characterization: for $\ell \in T^*_x \manifold$, we have
\begin{equation}
    \label{eq:dual norm variational}
    g_x^*(\ell, \ell) = \sup_{0 \neq u \in T_x \manifold} \frac{\abs{\ell(u)}^2}{g_x(u,u)}.
\end{equation}
This implies that the dualization is order-reversing in the following sense:
if one has two inner products~$g$ and~$\tilde g$ on~$T_xM$ (or any other vector space) such that~$g \preceq \tilde g$, i.e., $g(u,u) \leq \tilde g(u,u)$ for all~$u \in T_x\manifold$, then the dual inner products satisfy~$g^* \succeq {\tilde g}^*$, that is, for every~$\ell \in T_x^*\manifold$
\begin{equation}
  \label{eq:dual metric comparison}
  g^*(\ell, \ell) = \sup_{0 \neq u \in T_x \manifold} \frac{\abs{\ell(u)}^2}{g(u,u)} \geq \sup_{0 \neq u \in T_x \manifold} \frac{\abs{\ell(u)}^2}{\tilde g(u,u)} = \tilde g^*(\ell, \ell).
\end{equation}
Both~$g_x$ and~$g_x^*$ induce norms on~$T_xM$ and~$T_x^* \manifold$, respectively; we shall denote both by~$\normriem{\cdot}$.

We will write~$d\psi$ for the \emph{differential} of a function~$\psi \in C^\infty(\manifold)$.
It is formally a smooth map~$\manifold \to T^*\manifold$ such that~$x \mapsto d\psi_x \in T_x^* \manifold$.
In our setting~$\manifold \subseteq E$, and the differential can be concretely given by~$d\psi_x(u) = D\psi_x[u] = \partial_{t=0} \psi(x + t u)$; for the general definition see~\cite[Ch.~3]{leeIntroductionSmoothManifolds2012}.
The object~$d\psi$ is an example of a covector field.
We denote its norm by
\begin{equation}\label{eq:normriemdiff}
  \normriem{d\psi} \colon M \to \R, \quad
  \normriem{d\psi}(x) \coloneqq \sqrt{g_x^*(d\psi_x, d\psi_x)}
  = \sup_{0 \neq u \in T_x \manifold} \frac{\abs{d\psi_x(u)}}{\sqrt{g_x(u,u)}}.
\end{equation}
the smooth function that assigns to any $x \in \manifold$ the norm of the differential of~$\psi$ at this point.

The \emph{gradient} $\grad^g \psi$ is the vector field defined by~$g_x(\grad_x^g,\cdot) = d\psi_x$.
Unlike the differential, the gradient depends on the choice of the Riemannian metric.

\paragraph{Integration.}
Our choice of Riemannian metric also induces a measure on the manifold~$\domain$, which we denote by $d\vol_g$.
If one chooses linear coordinates~$x$ with respect to a basis of~$E$ and considers the metric~$g_x$ as a positive definite matrix, the measure $d\vol_g$ has density $\sqrt{\det g}$ with respect to the corresponding Lebesgue measure~$dx$.%
\footnote{It is a standard fact that~$d\vol_g$ does not depend on the local choice of coordinates (up to a sign, which is fixed by the orientation), see~\cite[Prop.~2.41]{leeIntroductionRiemannianManifolds2018}.}
We denote by~$L^2(\domain; \vol_g)$ the Hilbert space of square-integrable real-valued functions with respect to the measure~$\vol_g$, and denote by~$\ipriem{\cdot}{\cdot}$ the corresponding $L^2$-inner product.
Explicitly, we have
\[
    \ipriem{\phi}{\psi}
  = \int_\domain \phi(x) \psi(x) \, d\mathrm{vol}_g(x)
  = \int_\domain \phi(x) \psi(x) \sqrt{\det g} \, dx.
\]

We will also use the $L^2$-inner product~$\ipeuclz{\cdot}{\cdot}$ of functions induced by the inner product $\ipz{\cdot}{\cdot}$ on~$E$ for some fixed~$z \in M$.
If we choose linear coordinates~$\xi$ with respect to an orthonormal basis for the latter, this is given simply by
\[
  \ipeuclz{\phi}{\psi}
= \int_\domain \phi(\xi) \psi(\xi) \, d\xi,
\]
where $d\xi$ denotes the standard Lebesgue measure.

\subsection{Laplace--Beltrami and Schrödinger operators}
\label{sec:LBS}

Next, we define the Laplace--Beltrami operator, which is a natural generalization of the Laplace operator to the manifold setting, and we recall some fundamental facts about the corresponding Schr\"odinger operators.

\subsubsection{The Laplace--Beltrami operator}\label{sec:LB def}

There are various ways to understand the Laplace--Beltrami operator, which we (for now) define as a map
\[ \LB\colon C_c^\infty(\manifold) \to C_c^\infty(\manifold). \]
First off, it is the unique operator with the property that for any two functions $\phi, \psi \in C_c^\infty(\manifold)$ we have the following integration by parts type formula:
\begin{equation*}
\ipriem{\phi}{L \psi} = \int_\manifold \phi (\LB \psi) \, d\vol_g = - \int_\manifold g_x^*(d \phi_x, d \psi_x) \, d\vol_g = \int_\manifold (\LB \phi) \psi \, d\vol_g.
\end{equation*}
Note that $d\phi, d\psi$ are \emph{covector} fields here, so we use the inner product~$g_x^*$ on $T_x^* \manifold$ as opposed to $T_x \manifold$.

The Laplace--Beltrami operator can also succinctly be defined by the formula $\LB \psi = \Tr [\nabla (\grad \psi)]$, see e.g.~\cite[Exercise~5.14]{leeIntroductionRiemannianManifolds2018}.
Here, $\grad^g \psi$ is the gradient vector field (as defined in \cref{sec:LBnotation}) and~$\nabla$ is the so-called ``Levi-Civita connection'' associated with~$g$
Thus, the object $\nabla (\grad^g \psi)$ is an ``operator field'', i.e. for every~$x$ it gives a linear map~$T_x \manifold \to T_x \manifold$.\footnote{In fact it yields the Riemannian definition of the Hessian, not to be confused with the Euclidean Hessian~$H(x)$ defined earlier.}
Therefore we can take its trace.

In smooth local coordinates, if one expresses the metric~$g$ as a positive-definite matrix valued function, it is given by
\[
\LB \psi = \frac{1}{\sqrt{\det g}} \sum_{i,j=1}^n \frac{\partial}{\partial x_i} \left( g^{ij} \sqrt{\det g} \frac{\partial \psi}{\partial x_j} \right).
\]
Here $g^{ij} \coloneqq (g^{-1})_{ij}$ is the $(i,j)$-th entry of the inverse matrix~$g^{-1}$.
If we use the standard coordinates of~$E=\R^n$, then $g$ is the Euclidean Hessian~$H(x)$ and~$g^{ij}$ are the entries of its inverse.
In particular:
\begin{enumerate}
    \item If~$M = \R^n$ and $g_x(u,v) = u^T v$, then the Laplace--Beltrami operator equals the usual Laplacian.
    \item\label{ex:two} If~$M = \R^n$ and~$g_x(u,v) = u^T H v$ where~$H$ is a fixed positive definite matrix, then
    \[
        L \psi = \sum_{i,j=1}^n H^{ij} \frac{\partial}{\partial x_i} \frac{\partial \psi}{\partial x_j} = \Delta (\psi \circ H^{-1/2}).
    \]
    This is because~$H^{ij}$ and~$\sqrt{\det g}$ do not depend on~$x$.
    After a change to orthonormal coordinates for the inner product (e.g., $\xi = H^{-1/2} x$), one recovers the usual Laplacian.
    Note that this change of coordinates is implemented by a \emph{unitary} $L^2(\R^n; \sqrt{\det(H)} \, dx) \to L^2(\R^n; d\xi)$.
\end{enumerate}

\subsubsection{Schrödinger operators}\label{ssub:schroedinger_operators}

We recall some basic facts about Schrödinger operators
\[ \H = -\frac12 \LB+V \]
in the setting of complete Riemannian manifolds. Throughout, we will consider the setting where $V \in C^\infty(\manifold)$ is a smooth function on $\manifold$ whose sublevel sets are compact (i.e., $V \in C^\infty(\manifold)$ and for every $C \in \R$, we have $\{x \in \manifold: V(x) \leq C\}$ is compact).
This encompasses the setting where the potential is self-concordant: the sublevel sets of self-concordant functions are compact due to \cref{lem:strong sc blowup}, and the associated manifold is complete due to~\cite[Thm.~2.1]{nesterov2002geometry}.

\paragraph{Domain.} We initially consider $\H$ as an operator from $C_c^\infty(M)$ to $C_c^\infty(M)$. It is known that this is an essentially self-adjoint operator, meaning it has a unique self-adjoint extension. This follows for example from~\cite[Thm.~1]{Oleinik1994}, where we use that, since $V \in C^\infty(M)$ has compact sub-level sets, the potential is bounded from below by a constant function and belongs in particular to $L_{\loc}^\infty(M)$. We refer the interested reader to~\cite{braverman2002esa} for a discussion of essential self-adjointness of Schr\"odinger operators under different assumptions. We will from now on consider the unique self-adjoint extension of this operator, the Friedrichs extension, as an operator in $L^2(\manifold; \vol_g)$, here we follow e.g.~\cite{kondratevDiscretenessSpectrumSchrodinger1999}. We will denote the operator domain of $\H$ by $\Dom(\H)$:
\[
\Dom(\H) = \{u \in H_1 : \H u \in L^2(M)\},
\]
where $H_1$ is the completion of $C_c^\infty(M)$ with respect to the norm $\|u\|_\H^2 = \langle u,(V-V_{\min}+1) u\rangle_R + \langle \nabla u, \nabla u\rangle_R$, where the potential is shifted by $V_{\min} = 
\min_{x \in M} V(x)$ plus one to ensure $\|u\|_\H$ is a norm.

\paragraph{Spectrum.}
The completeness of~$M$ and the assumption of compactness of the sublevel sets of~$V$ implies that the spectrum of $\H$ is (purely) discrete and bounded from below.
In other words, the spectrum consists of a sequence of eigenvalues
\[
\lambda_0(\H) < \lambda_1(\H) < \lambda_2(\H) < \ldots,
\]
such that~$\lambda_k(\H) \to \infty$ as~$k \to \infty$,
and each eigenvalue has finite multiplicity.
This is well known (cf.~\cite{urakawa1991laplace,braverman2025semiclassicalweyllawcomplete}), but the proof for our setting of non-compact manifolds is not easily found in the literature, so we provide it in~\cref{sec:discrete spec}.

\paragraph{Regularity.}
We observe that $\H$ is an elliptic operator.
This holds because~$V$ is a smooth function and the Laplace--Beltrami operator is a second-order elliptic operator.
To see the latter, note that its principal symbol is $P_2(x,\xi) = \sum_{i,j} g^{ij}(x) \xi_i \xi_j$, and hence $P_2(x,\xi) \neq 0$ for any $x \in \manifold$ and $\xi \neq 0$. From elliptic regularity, it follows that eigenfunctions of $\H$ belong~to~$C^\infty(\manifold)$.%
\footnote{To see this, for an eigenfunction $\phi \in \Dom(\H)$ of $\H$ corresponding to eigenvalue $\lambda$, one can for instance apply \cite[Cor.~8.3.2]{hormanderVol1} to $P = \H - \lambda I$ and $u=\phi$. The corollary shows that the wavefrontset of $\phi$ is empty, which implies that there exists a smooth representation of $\phi$.}

\subsection{IMS localization on manifolds} \label{sec:IMSLB}

We briefly recall IMS localization on manifolds, see for example \cite{simon1983semiclassical} and \cite{shubin1996} for the setting of differential operators on vector bundles.
First we need the following lemma.
\begin{proposition}[{cf.~\cite[Lem.~3.1]{shubin1996}, \cite[Eq.~(11.37)]{CyFKS1987}}]\label{prop:comm2LB}
  Let $h \in C_c^\infty(\manifold)$ and $L$ the Laplace--Beltrami operator.
  Then
  \begin{equation*}
    [h, [h, L]] = - 2 \normriem{dh}^2, 
  \end{equation*}
  where $\normriem{\cdot}$ is defined in \cref{eq:normriemdiff}, and the functions $h$ and $\normriem{dh}^2$ are interpreted as multiplication operators on~$L^2(\manifold;\vol_g)$.
\end{proposition}
We provide a proof in~\cref{sec:comm2LB proof} that only relies on the formal properties of~$L$.
As in \cref{sec:Euclidean}, the above proposition can be used to derive a localization formula.

\begin{lemma}[IMS localization formula] \label{lem:IMS LB}
Let $J, \bar J \in C^\infty(\manifold)$ be such that $J^2 + \bar J^2 = I$.
Then,
  \begin{align*}
    \H = J \H J + \bar{J} \H \bar{J} -  \normriem{dJ}^2 -  \normriem{d\bar{J}}^2,
  \end{align*}
    where $\normriem{\cdot}$ is defined in \cref{eq:normriemdiff}.
\end{lemma}
\begin{proof}
By \cref{prop:comm2LB} we have that
\[
J^2 \H + \H J^2 - 2 J \H J
= [J,[J,\H]]
= -2 \normriem{d J}^2.
\]
We obtain a second equation by replacing $J$ by $\bar J$ in this equation.
Summing the resulting equations (and using that $J^2 + \bar J^2 = I$) proves the lemma.
\end{proof}

\subsection{Proof strategy} 
\label{sec:proof strategy Laplace-Beltrami}

Here we give an outline of the proof strategy, deferring formal statements to the later sections.
For the given self-concordant function~$f$ and parameter~$\semi > 0$, our goal is to study the Schrödinger operator defined as in \cref{ssub:schroedinger_operators},
\[
\H
= -\frac{1}{2} \LB + \semi^2 f.
\]
The idea is (again) to relate $\H$ to another operator, $\H_0$, which is defined as follows.
Let $z$ be the unique minimizer of~$f$.
We assume that~$f(z) = 0$; if not, we may replace~$f$ by~$f - f(z)$ without changing the spectral gap.
Consider the vector space~$E$ but equipped with the local inner product
\[
    \ipz u v = D^2 f_z[u,v].
\]
We denote the associated Laplace operator by~$\Lz$ (defined as in Example~\ref{ex:two} in \cref{sec:LB def}), and also define the function
\begin{equation}
    \label{eq:rz definition}
    r_z(x) = \normz{x-z}
\end{equation}
The operator~$\H_0$ is then given by the Euclidean Schr\"odinger operator
\[
\H_0 \coloneqq -\frac12\Lz + \semi^2 q, \qquad \text{where } q(x) = \frac12 r_z(x)^2.
\]
which is an essentially self-adjoint unbounded operator on $L^2(E)$ with the inner product~$\ipeuclz{\cdot}{\cdot}$ defined earlier.
That is, if one chooses affine coordinates~$\xi_1,\dots,\xi_n$ with respect to a $\braket{\cdot,\cdot}_z$-orthonormal basis of~$E$ such that~$z$ has coordinate~$\xi=0$, then
\begin{align*}
  \H_0 = -\sum_{i=1}^n \frac{\partial^2}{\partial \xi_i^2} + \frac{\gamma^2} 2 \xi^T \xi
\end{align*}
is a standard harmonic oscillator.
To show that $\H_0$ is a good approximation of $\H$, we will show that, after suitable localization near~$z$,
$-\Lz$ is a good approximation of $-L$ and $f$ is well approximated by $q$; this will also imply that their ground states have high overlap.

The main difference with the previous section is that now $\H-\H_0$ captures not just an approximation in potential, but also an approximation in the differential operator.
To see this, we write
\begin{equation} \label{eq:components H-H0}
    \H-\H_0 = \frac12(-L + \Lz) + \semi^2 (f-q).
\end{equation}
Given the IMS-localization formula from \cref{lem:IMS LB}, we thus need bounds on $J(-L+\Lz)J$ and $\semi^2 J(f-q)J$ for suitable cutoff functions~$J$.
Informally, the cutoff function allows us to establish such bounds \emph{locally}, i.e., we only need to consider the operators applied to functions whose support is contained in a small region around the minimizer of~$f$.
The main difference with the Euclidean setting is that we now have to control the operator $J(-\LB + \Lz)J$.
To this end, we take as the cutoff function
\begin{equation} \label{eq:Jdef}
    J(x)
    \coloneqq j(r_z(x))
    = j\mleft(\normz{x-z}\mright),
\end{equation}
where $j \in C^\infty(\R)$ is the univariate function constructed in \cref{lemma:j and barj properties}.
Then the support of~$J$ is contained in a Dikin ellipsoid of small radius~$2 \semi^{-2/5}$.
In particular, having~$2 \semi^{-2/5} < 1$ suffices to ensure that the support of~$J$ is contained in~$\domain$.
Then we can use properties of self-concordant functions to compare~$\LB$ and~$\Lz$, see \cref{thm:laplacebeltrami-vs-laplace}. The multiplication operator~$\semi^2 J(f-q)J$ can be dealt with in a similar fashion as in the previous section, using \cref{lem:HH0onJ}.

To argue about the ground state energy and the energy of the first excited state, we will use the Rayleigh--Ritz quotient:
\[
\mathcal R(\H,\psi) \coloneqq \frac{\ipriem{\psi}{\H \psi}}{\ipriem{\psi}{\psi}}.
\]
Crucially, we know the spectrum of $\H_0$ because it is a standard quantum harmonic oscillator.
Using IMS-localization we can bound the difference between $\ipriem{\psi}{\H \psi}$ and $\ipeuclz{\psi}{\H_0 \psi}$, see \cref{thm:H-vs-H0} for a formal statement.
In particular, we arrive at an estimate of the following form: for all $\psi \in \Dom(\H)$, we have
\begin{equation}\label{eq:goal}
\mathcal R(\H,\psi) \geq  (1-\poly(n, \semi^{-1})) \left(\lambda_1(\H_0) - \gap(\H_0) \mathcal R(F,\psi) \right) -o(\semi),
\end{equation}
where $F$ is a rank-$1$ operator. This gives the desired lower bound on $\lambda_1(\H)$.
To upper bound $\lambda_0(\H)$, we use $\psi=J \psi_0$, where $\psi_0$ is the ground state of $\H_0$, as a test function.
We show that
\begin{equation*}
    \mathcal R(\H,J \psi_0) \leq 
    \lambda_0(\H_0) + o(\semi).
\end{equation*}
For $\semi \in \Omega((n\log n)^5)$, combining these estimates gives a lower bound on the spectral gap of~$\H$, see \cref{thrm:gap}:
\[
\lambda_1(\H)-\lambda_0(\H) \geq \frac\semi2.
\]

\subsection{Locally approximating \texorpdfstring{$\H$ by $\H_0$}{H by H0}}
\label{sec:HvsH0}

Here we show that, locally, $\H$ is well-approximated by $\H_0$.
The main result of this section is \cref{thm:H-vs-H0}.
We first prove some bounds on the cut-off function~$J$ from \cref{eq:Jdef}. 
Let $\bar J$ be such that $J^2 + \bar J^2 = 1$.

\begin{lemma}\label{lem:gradbounds}
Let $J, \bar J \in C^\infty(\domain)$ as above.
If~$2 \semi^{-2/5} \leq 1 - 1 / \sqrt{2}$, then we have, for all~$x \in \domain$,
    \[
        \normriem{dJ}^2(x) \leq C \semi^{4/5}
        \quad\text{and}\quad
        \normriem{d\bar{J}}^2(x) \leq C \semi^{4/5}
    \]
for a universal constant~$C>0$.
\end{lemma}
\begin{proof}
Recall from \cref{eq:normriemdiff} that $\normriem{d\psi}(x) = \sup_{0 \neq u \in T_x \domain} {\abs{d\psi_x[u]}}/{\sqrt{g_x(u,u)}}$ for any $\psi \in C^\infty(\domain)$.
By the chain rule, we have $dJ_x(u) = j'(r_z(x)) \, d(r_z)_x(u)$ and hence
\begin{align*}
  \normriem{dJ}^2(x) = \abs{j'(r_z(x))}^2 \cdot \normriem{d(r_z)}^2(x).
\end{align*}
Since~$j'(t) \neq 0$ only when~$\abs{t} \leq 2 \semi^{-2/5}$, we may restrict to~$x \in \domain$ such that~$r_z(x) \leq 2 \semi^{-2/5}$, hence in particular $1 - \norm{x - z}_z \geq 1/\sqrt2$ by the assumption of the lemma.

Using that~$r_z(x) = \sqrt{\ipz{x-z}{x-z}}$, we get~$d(r_z)_x[u] = \ipz{x-z}{u} / r_z(x)$.
Let~$\ell_x\colon T_x\domain \to \R$ be given by~$u \mapsto \ipz{x-z}{u}$.
Then, $\normriem{d(r_z)} = \normriem{\ell_x} / r_z(x)$, where
\begin{align*}
  \normriem{\ell_x}^2
= \sup_{0 \neq u \in E} \frac{\abs{\ell_x(u)}^2}{g_x(u,u)}
\leq 2 \, \sup_{0 \neq u \in E} \frac{\abs{\ell_x(u)}^2}{\ipz u u}
= 2 \, \sup_{0 \neq u \in E} \frac{\abs{\ipz{x-z}{u}}^2}{\ipz u u}
= 2 \, \ipz{x-z}{x-z}
= 2 r_z(x)^2;
\end{align*}
for the inequality we used the self-concordance of~$f$ (\cref{eq:self-concord-1}) and that $1 - \norm{x - z}_z \geq 1/\sqrt2$.
It follows that $\normriem{d(r_z)}^2 \leq 2$, and hence
$\normriem{dJ}^2(x) \leq 2 \, \abs{j'(r_z(x))}^2 \leq 50 \, \semi^{4/5}$,
because $\abs{j'(t)} \leq 5 \semi^{2/5}$ for all $t\in\R$.
The same argument applies to~$\bar{J}$.
\end{proof}

\begin{lemma} \label{lem:JHJ LB}
    For any self-concordant $f$ with $f(z)=0$ as minimum, $\H = -\LB + \semi^2 f$, $J$ and $\bar J$ as above, and $\semi \geq \Omega(1)$, we have
    \[
    \bar J \H \bar J \succeq \frac14 \semi^{6/5} {\bar J}^2.
    \]
\end{lemma}
\begin{proof}
The proof is analogous to the one of \cref{lem:JHJ}, using instead that $-L \succeq 0$ and that the support of $\bar J$ is contained in the set $\{x : r_z(x) \geq \semi^{-2/5}\}$.
\end{proof}

We now restrict to compactly supported functions $\psi$ whose support lies within a small Dikin ellipsoid around $z$ and show that: (1) we can bound the difference between $\ipriem{\psi}{\psi}$ and $\ipeuclz{\psi}{\psi}$ (\cref{thm:g vs gz}) , and (2) we can bound the difference between $\ipriem{\psi}{\LB \psi}$ and $\ipeuclz{\psi}{\Lz \psi}$ (\cref{thm:laplacebeltrami-vs-laplace}). 
\begin{theorem}[Approximating inner products on Dikin ellipsoids]
    \label{thm:g vs gz}
    Let $r \in (0,1)$ and let $\psi \in C_c^\infty(\domain)$ be a compactly supported function such that $r_z(x) \leq r$ for~$x \in \supp \psi$.
    Then, we have
        \begin{align*}
        \left((1-r)^{n} -1\right) \ipeuclz{\psi}{\psi} \leq \ipriem{\psi}{\psi} - \ipeuclz{\psi}{\psi} \leq \left((1-r)^{-n} - 1\right) \ipeuclz{\psi}{\psi}.
    \end{align*}
\end{theorem}
\begin{proof}
Recall that if one chooses linear coordinates~$x$ with respect to a basis of~$E$ and considers the metric~$g_x$ as a positive-definite matrix with respect to this basis, then $d\vol_g = \sqrt{\det g_x} \, dx$, with $dx$ the Lebesgue measure in these coordinates.
Since self-concordance is invariant under linear transformations, when $r_z(x) = \norm{x-z}_z \leq r$, we have
$(1-r)^2 g_z \preceq g_x \preceq (1-r)^{-2} g_z$.
If we work with an orthonormal basis with respect to~$\ipz{\cdot}{\cdot}$, then $g_z = I$.
Then, $\ipeuclz\psi\psi = \int_M \abs{\psi(x)}^2 dx$ and
\begin{align*}
    \ipriem{\psi}{\psi}
    = \int_\domain \abs{\psi}^2 \, d\vol_g
    = \int_\domain \abs{\psi(x)}^2 \sqrt{\det g_x} \, dx
\end{align*}
is between $(1-r)^n \ipeuclz\psi\psi$ and $(1-r)^{-n} \ipeuclz\psi\psi$ because~$\psi$ is supported only on~$x$ such that $r_z(x) \leq r$.
\end{proof}
\begin{remark} We remark that in the above theorem we obtain a $(1-r)^n$-dependence by using that the operator norm of $g_z^{-1/2} g_x g_z^{-1/2}-I$ is upper bounded by roughly $r$. Certain self-concordant barriers satisfy a stronger notion of self-concordance, called \emph{strong self-concordance}, which allows one to bound the Frobenius norm by $r$, resulting in an improved $(1-r)^{\sqrt{n}}$-dependence. This notion was introduced in \cite{laddha2020strongsampling} in a sampling context; it is known to hold for instance for the standard logarithmic barrier for linear programming as well as the entropic barrier (for general convex sets).
\end{remark}

\begin{theorem}[Approximating $\LB$ by $\Lz$, on Dikin ellipsoids]
    \label{thm:laplacebeltrami-vs-laplace}
    Let $r \in (0,1)$ and let $\psi \in C_c^\infty(\domain)$ be a compactly supported function such that $r_z(x) \leq r$ for~$x \in \supp \psi$.
    Then, we have
    \begin{align*}
      \left((1-r)^{n+2} -1\right) \ipeuclz{\psi}{- \Lz \psi} \leq \ipriem{\psi}{-\LB \psi} - \ipeuclz{\psi}{-\Lz \psi} \leq \left((1-r)^{-(n+2)} -1\right) \ipeuclz{\psi}{-\Lz \psi}
    \end{align*}
\end{theorem}

\begin{proof}
As in the preceding proof we work in linear coordinates~$x$ determined by an orthonormal basis for~$E$ with respect to $\ipz{\cdot}{\cdot}$, identify $g_x$ for $x\in\domain$ with a positive definite matrix, and use that $d\vol_g = \sqrt{\det g_x} \, dx$.
Then, identifying $T^*_xM = E$, we have
\begin{align}\label{eq:integrals}
  \ipriem{\psi}{-\LB \psi} = \int_\manifold g_x^*(d \psi_x, d \psi_x) \sqrt{\det g_x} \, dx
  \quad\text{and}\quad
  \ipeuclz{\psi}{-\Lz \psi} = \int_\manifold g_z^*(d \psi_x, d\psi_x) \, dx.
\end{align}
Thus it suffices to estimate the bilinear form
\begin{align*}
  g_x^* \sqrt{\det g_x} - g_z^*
\end{align*}
As before, the self-concordance of~$f$ yields the estimate
$(1-r)^2 g_z \preceq g_x \preceq (1-r)^{-2} g_z$
when $r_z(x) \leq r$.
Then the same relation holds for~$g_z^*$ and~$g_x^*$, see \cref{eq:dual metric comparison}.
Moreover, $(1-r)^n \leq \sqrt{\det g_x} \leq (1-r)^{-n}$, using that $g_z = I$ in the chosen coordinates.
Combing the two estimates, we obtain that
\begin{align*}
  \left( (1-r)^{n+2} - 1 \right) g_z^* \preceq g_x^* \sqrt{\det g_x} - g_z^* \preceq \left( (1-r)^{-(n+2)} - 1 \right) g_z^*,
\end{align*}
which in view of \cref{eq:integrals} implies the statement of theorem.
\end{proof}

We finally show that, locally, $\H$ is well-approximated by $\H_0$.
\begin{theorem}[Approximating $\H$ by $\H_0$, on Dikin ellipsoids]
  \label{thm:H-vs-H0}
  Let $r \leq 2 \gamma^{-2/5}$ and let $\psi \in C_c^\infty(\domain)$ be a compactly supported function such that $r_z(x) \leq r$ for all~$x \in \supp \psi$.
  Then we can bound 
  \begin{align*}
    \left((1-r)^{n+2} -1\right) \ipeuclz{\psi}{\H_0 \psi} - \frac{8}{3} \semi^{4/5} \ipriem{\psi}{\psi} &\leq \ipriem{\psi}{\H \psi} - \ipeuclz{\psi}{\H_0 \psi} \\&\leq  \left((1-r)^{-(n+2)} - 1\right)
    \ipeuclz{\psi}{\H_0 \psi} + \frac{8}{3} \semi^{4/5} \ipriem{\psi}{\psi}.
  \end{align*}
\end{theorem}
\begin{proof}
  We have
  \[
    \ipriem{\psi}{\H \psi} - \ipeuclz{\psi}{\H_0 \psi}
    = \frac12\left(\ipriem{\psi}{-\LB \psi} - \ipeuclz{\psi}{-\Lz \psi} \right)
    + \semi^2 \left(\ipriem{\psi}{f \psi} - \ipeuclz{\psi}{q \psi}\right).
  \]
  We bound the two terms separately. For the first term we use \cref{thm:laplacebeltrami-vs-laplace} to obtain the estimates
  \begin{equation}\label{eq:est1}
    \left((1-r)^{n+2} -1\right) \ipeuclz{\psi}{-\Lz \psi}
    \leq \ipriem{\psi}{-\LB \psi} - \ipeuclz{\psi}{-\Lz \psi}
    \leq \left((1-r)^{-(n+2)} - 1\right) \ipeuclz{\psi}{-\Lz \psi}.
  \end{equation}
  For the second term we further decompose
  \begin{align*}
    \semi^2 \left(\ipriem{\psi}{f \psi} - \ipeuclz{\psi}{q \psi} \right) &= \ipriem{\psi}{\semi^2 (f-q) \psi} + \left(\ipriem{\psi}{\semi^2 q \psi} - \ipeuclz{\psi}{\semi^2 q \psi}\right).
  \end{align*}
  \Cref{lem:HH0onJ} then shows that
\begin{align}\label{eq:est3}
   \abs{\ipriem{\psi}{\semi^2 (f-q)\psi}} \leq \frac{8}{3} \semi^{4/5} \ipriem{\psi}{\psi}.
\end{align}
To conclude, observe that we can write the positive definite quadratic function $q(x) = \frac12 \ipz{x-z}{x-z}$ as a sum of squares, $q = \sum_j k_j^2$.
Applying \cref{thm:g vs gz} to $\semi k_j \psi \in C_c^\infty(\domain)$ and summing the resulting estimates gives
\begin{equation} \label{eq:est2}
  \left((1-r)^{n} -1\right) \ipeuclz{\psi}{\semi^2 q \psi} \leq \ipriem{\psi}{\semi^2 q \psi} - \ipeuclz{\psi}{\semi^2 q \psi} \leq \left((1-r)^{-n} - 1\right) \ipeuclz{\psi}{\semi^2 q \psi}.
\end{equation}
  Now the theorem follows by combining the estimates in \cref{eq:est1,eq:est2,eq:est3}.
\end{proof}

\subsection{Proof of the lower bound on \texorpdfstring{\smash{$\lambda_1(-\frac12 L+\semi^2f)$}}{lambda1}}
\label{sec:LBlower}

Here we prove a more explicit version of~\cref{eq:goal}, which we restate here for convenience:
\begin{align*}
\mathcal R(\H,\psi) \coloneqq \frac{\ipriem{\psi}{\H \psi}}{\ipriem{\psi}{\psi}} \geq (1-\poly(n, \semi^{-1})) \left(\lambda_1(\H_0) - \gap(\H_0) \mathcal R(F,\psi) \right) -o(\semi),
\end{align*}
where $F$ is a rank-$1$ operator.

\begin{theorem} \label{thrm:LBlambda1}
  There exist universal constants $C,c>0$ such that the following holds.
  Let $\H = -\frac12 \LB + \semi^2 f$ and let~$J$ as in \cref{eq:Jdef}.
  Let~$\psi_0$ be the ground state of~$\H_0$ (it is unique up to a scalar), and let~$F = \frac{\psi_0\ipeuclz{\psi_0}{\cdot}}{\ipeuclz{\psi_0}{\psi_0}}$.
  If~$\semi \geq c n^5$, then, for all $\psi \in C^\infty(\domain) \cap \Dom(\H)$, we have
  \[
    \mathcal R(\H,\psi) \geq \lambda_1(\H_0) (1-2\semi^{-2/5})^{2n+2} - \gap(\H_0)(1-2\semi^{-2/5})^2 \mathcal R(JFJ,\psi) - C \semi^{4/5}.
  \]
  In particular, it holds that
  \begin{align*}
    \lambda_1(\H) \geq \lambda_1(\H_0) (1-2\semi^{-2/5})^{2n+2} - C \semi^{4/5}.
  \end{align*}
\end{theorem}

\begin{proof}
  The IMS localization formula from \cref{lem:IMS LB} states that
  \[
    \H
    = J \H J + \bar{J} \H \bar{J} -  \normriem{dJ}^2 -  \normriem{d\bar{J}}^2.
  \]
  \Cref{lem:gradbounds} establishes a pointwise upper bound on the right-hand side functions:
  there exists a universal constant $\tilde{C}>0$ such that $\normriem{dJ}^2(x), \normriem{d\bar J}^2(x) \leq \tilde{C} \semi^{4/5}$ for all $x \in \domain$.
    We thus have
    \begin{align}
      \mathcal R(\H,\psi)
      & \geq
        \frac{\ipriem{\psi}{J\H J \psi} + \ipriem{\psi}{\bar J \H \bar J \psi}}{\ipriem{\psi}{\psi}}
      - 2\tilde{C} \semi^{4/5}.
      \label{eq:R1}
    \end{align}
    By \cref{lem:JHJ LB} we have $\bar J \H \bar J \succeq \frac14 \semi^{6/5} \bar J^2$ and thus
    \begin{align} \label{eq:R2}
      \ipriem{\psi}{\bar J \H \bar J \psi} \geq \frac14 \semi^{6/5}  \ipriem{\psi}{\bar J^2 \psi}.
    \end{align}
    We now lower bound $\ipriem{\psi}{J \H J \psi}$. Let $r = 2\semi^{-2/5}$.
    Then \cref{thm:H-vs-H0} (applied to~$J \psi$) implies that
    \begin{align}\notag
      \ipriem{\psi}{J\H J \psi}
    &\geq (1-r)^{n+2} \ipeuclz{\psi}{J \H_0 J \psi} - \frac{8}{3} \semi^{4/5} \ipriem{\psi}{J^2 \psi} \\
    \label{eq:R3}
    &\geq (1-r)^{n+2} \ipeuclz{\psi}{J \H_0 J \psi} - \frac{8}{3} \semi^{4/5},
    \end{align}
    where we also used that $J^2 \leq 1$ and hence $\ipriem{\psi}{J^2 \psi} \leq \ipriem{\psi}{\psi}$.
    Writing $\gap(\H_0) = (\lambda_1(\H_0)-\lambda_0(\H_0))$, and using that~$\psi_0$ is the unique ground state of~$\H_0$, hence $\H_0 \succeq \lambda_1(\H_0) I - \gap(\H_0) F$, we have
    \begin{align}
      \ipeuclz{\psi}{J \H_0 J \psi} &\geq \lambda_1(\H_0) \ipeuclz{\psi}{J^2 \psi} - \gap(\H_0) \ipeuclz{\psi}{J F J \psi} \notag \\
                              &\geq  \lambda_1(\H_0)(1-r)^n  \ipriem{\psi}{J^2 \psi} - \frac{\gap(\H_0)}{(1-r)^n} \ipriem{\psi}{J F J \psi}, \label{eq:R4}
    \end{align}
    where the second inequality uses \cref{thm:g vs gz} (applied to $J\psi$ and $FJ\psi$, respectively).
    Combining the estimates from \cref{eq:R1,eq:R2,eq:R3,eq:R4} gives
    \begin{align*}
        &\mathcal R(\H,\psi) \geq \frac{\ipriem{\psi}{J\H J \psi} + \ipriem{\psi}{\bar J \H \bar J \psi}}{\ipriem{\psi}{\psi}} - 2\tilde{C} \semi^{4/5} \\
        &\geq \frac{(1-r)^{n+2} \left( \lambda_1(\H_0)(1-r)^n \ipriem{\psi}{J^2 \psi} - \frac{\gap(\H_0)}{(1-r)^n} \ipriem{\psi}{J F J \psi}\right)  + \frac14 \semi^{6/5} \ipriem{\psi}{\bar J^2 \psi}}{\ipriem{\psi}{\psi}} - \left( 2\tilde{C} +\frac{8}{3} \right)\semi^{4/5}.
    \end{align*}
    By our choice $r = 2 \gamma^{-2/5}$, there exists a universal constant $c>0$ such that when $\semi \geq c n^5$, we have 
    \[
    \frac14 \semi^{6/5}
    \geq \frac\gamma2 (n+2)
    \geq \frac\gamma2 (n+2)(1-2\semi^{-2/5})^{2n+2}
    = \lambda_1(\H_0) (1-r)^{2n+2}.
    \]
    This allows us to combine the terms $\ipriem{\psi}{J^2 \psi}$ and $\ipriem{\psi}{\bar J^2 \psi}$ to obtain
    \begin{align*}
      \mathcal R(\H,\psi) \geq \lambda_1(\H_0) (1-2\semi^{-2/5})^{2n+2} - \gap(\H_0)(1-2\semi^{-2/5})^2 \frac{\ipriem{\psi}{J F J \psi}}{\ipriem{\psi}{\psi}} - \left( 2\tilde{C} +\frac{8}{3} \right)\semi^{4/5}.
    \end{align*}
    This proves the lower bound on the Rayleigh--Ritz quotient for $\psi \in C^\infty(M) \cap \Dom(\H)$. Since the eigenfunctions of $\H$ are smooth and the perturbation $JFJ$ has rank $1$, the lower bound on $\lambda_1(\H)$ follows. 
\end{proof}

\subsection{Proof of the upper bound on \texorpdfstring{\smash{$\lambda_0(-\frac12 L+\semi^2f)$}}{lambda0}} \label{sec:LBupper}

Here we show that a localization of the ground state of the harmonic approximation~$\H_0$ has energy with respect to~$\H$ near $\lambda_0(\H_0)$ for sufficiently large $\semi$, giving us a good upper bound on the ground state energy of~$\H$.
Because this energy is significantly below $\lambda_1(\H)$, this also implies that the ground state of~$\H_0$ has significant overlap with the ground state of~$\H$, see \cref{cor:riem-gs-overlap} below.

\begin{theorem} \label{thrm:UBlambda0}
  There exist universal constants $c,C>0$ such that the following holds.
  Let $\psi_0$ be the ground state of $\H_0$.
   If $\semi \geq c(n\log(n))^5$, then
    \[
    \lambda_0(\H) \leq
    \mathcal R(\H,J \psi_0) \leq  \lambda_0(\H_0) (1-2\semi^{-2/5})^{-2n-2}  + C (1-2\semi^{-2/5})^{-2n-2}\semi^{4/5}.
    \]
\end{theorem}
\begin{proof}
    The first inequality holds by the variational principle, so we only have to prove the upper bound on the Rayleigh--Ritz quotient.
    To this end, let $\phi = J \psi_0$ and note that $\phi \in C_c^\infty(\domain)$ is a compactly supported function such that $r_z(x) \leq r$ for all~$x \in \supp \phi$ for $r=2\semi^{-2/5}$.
    We then have, using \cref{thm:H-vs-H0} and \cref{thm:g vs gz},
    \begin{align*}
    \mathcal R(\H,\phi) &\leq (1-r)^{-(n+2)}\frac{\ipeuclz{\phi}{\H_0 \phi}}{\ipriem{\phi}{\phi}} + \frac{8}{3} \semi^{4/5} \\
                     &\leq (1-r)^{-2n-2}\frac{\ipeuclz{\phi}{\H_0 \phi}}{\ipeuclz{\phi}{\phi}} + \frac{8}{3} \semi^{4/5}.
    \end{align*}
    We finally use \cref{lem:UBlaplace} (with $\mathcal H = \mathcal H_0$ and $A = I$) to bound $\frac{\ipeuclz{\phi}{\H_0 \phi}}{\ipeuclz{\phi}{\phi}}$: there exist universal constants $c,\tilde C>0$ such that if $\semi$ satisfies $\semi \geq c(n \log(\semi n))^5$, then
    \begin{align*}
      \frac{\ipeuclz{\phi}{\H_0 \phi}}{\ipeuclz{\phi}{\phi}}
    \leq \lambda_0(\H_0) + 2 \tilde C \gamma^{4/5}.
    \end{align*}
    Combining the estimates gives the claimed bound.
    We finish by using \cref{lem:log-helper} to argue that $\semi \geq c(n \log(\semi n))^5$ is implied by $\semi \geq \tilde c (n \log(n))^5$ for sufficiently large $\tilde c$.
\end{proof}

\subsection{A lower bound on the spectral gap of \texorpdfstring{\smash{$-\frac12 \LB+\semi^2f$}}{H}}
\label{sec:LBgap}

Here we combine the lower bound on $\lambda_1(\H)$ and the upper bound on $\lambda_0(\H)$ to prove our main result: a concrete choice of $\gamma$ for which we can bound the spectral gap.

\riemanniangap*

\begin{proof}
    On the one hand, \cref{thrm:UBlambda0} shows that
    \[
        \lambda_0(\H) \leq \lambda_0(\H_0) (1-2\semi^{-2/5})^{-2n-2}  + C (1-2\semi^{-2/5})^{-2n-2}\semi^{4/5}.
    \]
    On the other hand, \cref{thrm:LBlambda1} shows that
    \[
    \lambda_1(\H) \geq \lambda_1(\H_0) (1-2\semi^{-2/5})^{2n+2} - C \semi^{4/5}.
    \]
    Combining the two estimates shows that, after possibly increasing the universal constant~$c$,
    \begin{align*}
    \gap(\H) &\geq \lambda_1(\H_0) (1-2\semi^{-2/5})^{2n+2} - \lambda_0(\H_0) (1-2\semi^{-2/5})^{-2n-2}  - C(1+ (1-2\semi^{-2/5})^{-2n-2})\semi^{4/5} \\
    &\geq \lambda_1(\H_0) (1-(cn)^{-2})^{2n+2} - \lambda_0(\H_0) (1-(cn)^{-2})^{-2n-2}  - C(1+ (1-n^{-2})^{-2n-2})\semi^{4/5} \\
    &\geq \gap(\H_0)/2,
    \end{align*}
    where for the final inequality we use that $\lambda_1(\H_0) = \frac{\gamma}{2} (n+2)$, $\lambda_0(\H_0) = \frac{\gamma}{2} n$, and thus $\gap(\H_0) = \gamma$. 
\end{proof}

As in the Euclidean setting (\cref{cor:eucl-gs-overlap}), we obtain as a corollary a bound on the overlap between the ground state of our Schrödinger operator $\H$ and that of its harmonic approximation $\H_0$. We use this bound later in our quantum path-following algorithm.

\begin{corollary}[Ground state overlap, Riemannian]\label{cor:riem-gs-overlap}
Let $\psi$ and $\psi_0$ denote the ground states of $\H$ and $\H_0$, respectively.
For any $\delta > 0$, if $\semi \in \Omega\left((\max\{n\log(n), 1/\delta)\})^5\right)$ then
\[
\abs{\ipriem{\psi}{J \psi_0}}
\geq 1 - \delta. 
\]
\end{corollary}
\begin{proof}
We use \cref{thrm:UBlambda0}, which implies that for~$\gamma \geq c (n \log(n))^5$, so that
\[
\ipriem{J\psi_0}{\H J\psi_0} \leq
\frac{\ipriem{J\psi_0}{\H J\psi_0}}{\ipriem{J\psi_0}{J\psi_0}}
\leq \lambda_0(\H_0)(1-2\semi^{-2/5})^{-2n-2}  + C (1-2\semi^{-2/5})^{-2n-2}\semi^{4/5}.
\]
On the other hand we use that $\H \succeq \lambda_0(\H) \ket{\psi}\bra{\psi} + \lambda_1(\H) (I - \ket{\psi}\bra{\psi})$ and so
\[
\ipriem{J\psi_0}{\H J\psi_0}
\geq \lambda_0(\H) \abs{\ipriem{J\psi_0}{\psi}}^2 + \lambda_1(\H) (1 - \abs{\ipriem{J\psi_0}{\psi}}^2).
\]
Combining these inequalities , and using that $\lambda_0(\H_0) = \lambda_0(\H) + O( \semi^{4/5})$ (\cref{thrm:UBlambda0}), we get that 
\[
\abs{\ipriem{J \psi_0}{\psi}}^2
\geq \frac{\lambda_1(\H) - \lambda_0(\H_0) - O(\semi^{4/5})}{\lambda_1(\H) - \lambda_0(\H)}
= 1 - O\left(\frac{\semi^{4/5}}{\lambda_1(\H) - \lambda_0(\H)}\right).
\]
Finally, we use that for sufficiently large $\semi$ we have $\lambda_1(\H)-\lambda_0(\H) \in \Omega(\semi)$.
Taking $\semi \in \Omega(\left(1/\delta\right)^5)$ we get the claimed conclusion.
\end{proof}

\section{A quantum path-following method} 
\label{sec:Qannealing}

We now show how our non-asymptotic semiclassical analysis can be used in an optimization context. Specifically, our goal is to solve the convex optimization problem $\min_{x \in \domain} c^T x$.
We let $f$ be a self-concordant barrier for $\domain$, and $\{x_\eta\}$ the central path associated to the barrier $f$ and objective $c$ (see \cref{eq:central} for details).

Our quantum path-following algorithm builds on two sequences of Hamiltonians, parametrized by $\semi>0$ and $\eta \geq 0$, defined as
\[
\HE(\eta) \coloneqq -\Delta + \semi^2(\eta c^T x +f),
\]
and
\[
\HR(\eta) \coloneqq -L + \semi^2(\eta c^T x +f),
\]
with superscripts $E$ and $R$ referring to the Euclidean and Riemannian manifold setting (following \cref{sec:Euclidean,sec:Riemannian}); for statements that apply to both, we leave out the superscript.
From our semiclassical analysis, we know that the ground states $\ket{\psi_\eta}$ of $\H(\eta)$ have large overlap with Gaussians centered at $x_\eta$, and we make this more precise in \cref{sec:overlap}.
The idea then is to use a quantum annealing algorithm to trace out a ``quantum central path'' $\{\ket{\psi_{\eta_0}},\ket{\psi_{\eta_1}},\dots,\ket{\psi_{\eta_T}}\}$, with $\eta_T = \vartheta/\varepsilon$ so that $\ket{\psi_{\eta_T}}$ concentrates at the minimizer of our optimization problem (see \cref{eq:central} for details).
Our algorithm builds on the following assumptions:
\begin{enumerate}
\item
\textbf{Initial state:} We can prepare a state $\ket{\widetilde\psi_{\eta_0}}$ that has constant overlap with $\ket{\psi_{\eta_0}}$.
\item
\textbf{Ground state energy estimates:}
For any $\eta$, we have a $\delta$-additive approximation of the ground state energy $\lambda_0(\H(\eta))$, with $\delta \in \widetilde{O}(\sqrt{\varepsilon} \Delta(\H(\eta))/T)$ for $\Delta(\H(\eta))$ the spectral gap of $\H(\eta)$ and $\varepsilon>0$ the goal precision.
\item
\textbf{(Controlled) Hamiltonian simulation:} We need to implement Hamiltonian simulation $e^{i \H(\eta)}$ controlled on a continuous-variable register that contains a standard Gaussian state: for any $\ket{\phi} \in \mathrm{dom}(\H(\eta))$ and $\ket{\psi_g} = (2\pi)^{-1/4} \int e^{-z^2/4} \ket{z} \, \mathrm{d} z$, this corresponds to the operation
\[
U \ket{\phi} \ket{\psi_g}
= \frac{1}{(2\pi)^{1/4}} \int e^{-z^2/4} e^{i \H(\eta) z} \ket{\phi} \ket{z} \, \mathrm{d} z.
\]
Given that the Gaussian concentrates around the origin, we assume that this is charged at the same cost as that of implementing $e^{i\H(\eta)}$.
\end{enumerate}
Under these assumptions, we prove the following theorem.
\begin{theorem}
Under assumptions 1.-3., we describe a quantum path-following method that returns with constant probability a state $\varepsilon$-close to $\ket{\psi_\eta}$, at the cost of simulating $\H(\eta)$ for total time $T^E$ and $T^R$ in the Euclidean and Riemannian settings respectively, where 
\[
T^E = \poly(n,\vartheta,h^-(\varepsilon),h^+(\eps),\log(1/\varepsilon)), 
\]
where $h^-(\varepsilon) = \max_{\eta_0 \leq \eta \leq \vartheta/\eps} \|g_{x_\eta}^{-1/2}\|$ and $h^+ (\varepsilon) = \max_{\eta_0 \leq \eta \leq \vartheta/\eps} \|g_{x_\eta}^{1/2}\|$, and 
\[
T^R = O\big( (n\log(n))^{5/2} \sqrt{\vartheta} \log(\vartheta/\varepsilon) \big).
\]
In particular, $T^R$ does \emph{not} depend on any condition numbers.
\end{theorem}

Let us briefly reflect on the assumptions.
The \emph{ground state energy estimates} are used for doing ground-state projection, and such estimates are also required in the usual approach of doing quantum phase estimation on bounded Hamiltonians (see e.g.~\cite{ge2019faster}).
Moreover, our semiclassical analysis tells us that the ground state energy of $\H(\eta)$ is well-approximated by that of its harmonic approximation. The ground state energy of the harmonic approximation is in turn well known: $\frac{\gamma}{2} \Tr[\sqrt{A}] + \gamma^2 \min_x f(x)$. In the Euclidean setting, $A$ is the Hessian of $f$ at the minimizer. In the Riemannian setting, $A=I$. Estimating the ground state energy can thus be done by estimating the minimum of $f$, and in the Euclidean setting also the Hessian at the minimizer. 
Regarding the \emph{Hamiltonian simulation}, our complexity measure is the time for which we need to do controlled-Hamiltonian simulation with respect to $\H$.
In light of the harmonic approximation $\H_{0,\semi} = -\Delta + \semi^2 \frac12 x^T A x$, a reasonable assumption would be that the cost of simulating $\H$ for one unit of time scales with $\semi$. Indeed, applying a substitution $x \leftarrow \semi x$ shows that $\H_{0,\semi}$ is equivalent to $\semi \H_{0,1}$.

\subsection{Quantum annealing algorithm}
\label{sec:idealized}

A natural approach to carry out our algorithm would be to use the quantum adiabatic theorem.
It roughly states that if we start in the ground state of $\H(\eta_0)$ and we evolve with $\H(\eta)$ for $\eta$ slowly going from $\eta_0$ to $\eta_T$, then the system should remain in the instantaneous ground state and hence end up in the target ground state of $\H(\eta)$.
Unfortunately we did not manage to find an appropriate, non-asymptotic version of the adiabatic theorem that applies to our setting of unbounded Hamiltonians (see also the open questions in \cref{sec:open}).

To get around this, we use a \emph{discrete-time} (but continous-space!) quantum annealing approach as described by Wocjan and Abeyesinghe~\cite{wocjan2008speedup}, and based on Grover's fixed-point $\pi/3$-algorithm.
For a sequence of states $\{\ket{\psi(\ell)}\}$, the algorithm is essentially a sequence of $\pi/3$-rotations $R(\ell)$ around each of these states, defined by
\[
R(\ell)
= I + (\omega-1) \ket{\psi(\ell)}\bra{\psi(\ell)},
\]
where $\omega = e^{i\pi/3}$.
\begin{corollary}[{\cite[Corollary 1]{wocjan2008speedup}}]
    Let $\ket{\psi(0)},\ldots,\ket{\psi(T)}$ be a sequence of states with $|\langle \psi(\ell)|\psi(\ell+1)\rangle| \geq w_*$ for $\ell=0,\ldots,T-1$. Then, given the state $\ket{\psi(0)}$, we can prepare a state $\ket{\phi}$ such that
    \[
    \|\ket{\phi}-\ket{\psi(T)}\| \leq \eps
    \]
    using the unitaries $R(\ell), R(\ell)^*$, $\ell=0,\ldots,T$, at most $O(T \log(T/\eps)/w_*^2)$ many times.\footnote{One might be able to improve the $w_*^2$ dependency to $w_*$ using "optimal fixed-point search" \cite{yoder2014fixed}, but in our application we will have $w_* = \Omega(1)$.}
\end{corollary}

\noindent
Our goal is thus twofold:
\begin{enumerate}
    \item Construct a sequence $\eta_0,\ldots,\eta_T = \vartheta/\eps$ such that the sequence of Hamiltonians and corresponding ground states $\{\H(\eta_\ell),\ket{\psi_{\eta_\ell}}\}$ satisfies $|\langle \psi_{\eta_\ell}|\psi_{\eta_{\ell+1}})\rangle| = \Omega(1)$ for each $\ell=0,\ldots,T-1$.
    \item Show how to implement the corresponding rotations $R(\eta_\ell)$ around $\ket{\psi_{\eta_\ell}}$ for each $\ell$.
\end{enumerate}
In \cref{sec:overlap} we use our semiclassical analysis to construct a sequence of $\eta$'s with the desired property.
In \cref{sec:gs-proj-unbounded} we show how to implement the rotations $R(\eta_\ell)$ for these states.

\subsection{Quantum central path}
\label{sec:overlap}

Here we define a sequence $\eta_0,\ldots,\eta_T = \vartheta/\eps$ that defines a suitable ``quantum central path'' of ground states $\{\ket{\psi_{\eta_0}},\dots,\ket{\psi_{\eta_T}}\}$.
In the Euclidean setting, we show in \cref{cor:eucl-gs-overlap2} below that picking
\[
\eta_{\ell+1}
= \Big(1+ O\big( (n+2\semi\|g_{x_{\eta_\ell}}^{-1/2}\|)\vartheta\big)^{-1/2})\Big) \eta_\ell
\]
for sufficiently large $\gamma$ ensures a constant overlap between the ground states of $\HE(\eta_\ell)$ and $\HE(\eta_{\ell+1})$.
If we let $h^-(\varepsilon) = \max_{\eta_0 \leq \eta \leq \vartheta/\eps} \|g_{x_\eta}^{-1/2}\|$ and $h^+ (\varepsilon) = \max_{\eta_0 \leq \eta \leq \vartheta/\eps} \|g_{x_\eta}^{1/2}\|$, and taking into account the scaling of $\gamma$ from \cref{cor:eucl-gs-overlap2}, then we can pick
\[
T^E
\in O\left( \sqrt{(n + 2 \gamma h^-(\varepsilon)) \vartheta} \log(\vartheta/\varepsilon) \right)
\in \poly(n,\vartheta,h^-(\varepsilon), h^+(\eps),\log(1/\varepsilon)).
\]
In the Riemannian setting, we show in \cref{cor:riem-gs-overlap2} below that picking
\[
\eta_{\ell+1}
= \Big(1+ O\big( (n+2\semi)\vartheta\big)^{-1/2})\Big) \eta_\ell
\]
ensures a constant overlap between the ground states of $\HE(\eta_\ell)$ and $\HE(\eta_{\ell+1})$, so we can pick
\[
T^R
\in O\left( \sqrt{(n + 2 \gamma) \vartheta} \log(\vartheta/\varepsilon) \right)
\in \poly(n,\vartheta,\log(1/\varepsilon)).
\]

\subsubsection{Ground state overlaps}

In \cref{cor:eucl-gs-overlap,cor:riem-gs-overlap} we saw that for sufficiently large $\gamma$ the ground states $\psi^E_\eta$ and $\psi^R_\eta$ of both $\HE(\eta)$ and $\HR(\eta)$ have large overlap with the ground states $\psi^E_{0,\eta}$ and $\psi^R_{0,\eta}$ of the harmonic approximations $\HE_0(\eta)$ and $\HR_0(\eta)$.
The latter satisfy
\[
\psi^E_{0,\eta}
\propto \exp\left( -\frac{\gamma}{2} (x-x_\eta)^T g_{x_\eta}^{1/2} (x-x_\eta)\right), \qquad
\psi^R_{0,\eta}
\propto \exp\left( -\frac{\gamma}{2} (x-x_\eta)^T g_{x_\eta} (x-x_\eta)\right)
\]
so that their induced measures correspond to the Gaussian states
\[
\psi^E_{0,\eta}(x)^2
\propto \mathcal{N}(x_\eta,(2\gamma g_{x_\eta}^{1/2})^{-1}), \qquad
\psi^R_{0,\eta}(x)^2
\propto \mathcal{N}(x_\eta,(2\gamma g_{x_\eta})^{-1}).
\]
This gives us a convenient means for bounding ground state overlaps, e.g., by using the following standard relation between the Hellinger distance and the total variation distance: 
let $\mu$ and $\pi$ be probability distributions and let $\ket{\mu} = \int \sqrt{\mu(x)} \ket{x} \, \mathrm{d}x$ and similar for $\ket{\pi}$, then
\[
1-\langle \mu|\pi\rangle \leq \|\mu - \pi \|_\TV \leq \sqrt{1-\langle \mu|\pi\rangle^2}.
\]
Combined with \cref{cor:eucl-gs-overlap,cor:riem-gs-overlap} this yields the following lemma.
\begin{lemma}
Let $\eta, \eta' > 0$.
For
$\semi \in \Omega\big(\max_{s \in \{\eta,\eta'\}}(\sqrt{\norm{g_{x_s}}} \, \max\{n \log(n \|g_{x_s}\|), \sqrt{\norm{g_{x_s}}/\lambda_{\min}(g_{x_\eta})}\})^5\big)$ 
we have that
\[
\abs{\braket{\psi^E_\eta|\psi^E_{\eta'}}}
\geq \frac{9}{10}  - \| \mathcal{N}(x_\eta,(2\gamma g_{x_\eta}^{1/2})^{-1}) - \mathcal{N}(x_{\eta'},(2\gamma g_{x_{\eta'}}^{1/2})^{-1}) \|_{\TV}.
\]
For $\gamma \geq \Omega\left((n \log(n))^5\right)$ we have that
\[
\abs{\ipriem{\psi^R_\eta}{\psi^R_{\eta'}}}
\geq \frac{9}{10} - \norm{ \mathcal{N}(x_\eta,(2\gamma g_{x_\eta})^{-1}) - \mathcal{N}(x_{\eta'},(2\gamma g_{x_{\eta'}})^{-1}) }_{\TV}.
\]
\end{lemma}
\begin{proof}
In the Euclidean setting, relating the inner product to total-variation distance is straightforward. In the Riemannian setting we give an explicit proof.
    Let~$J_\eta$ and~$J_{\eta'}$ be the cut-off functions for~$x_\eta$ and~$x_{\eta'}$, respectively.
    From~\cref{cor:riem-gs-overlap}, we get that~$\abs{\ipriem{\psi^R_\eta}{J_\eta \psi^R_{0,\eta}}} \geq 1-\delta$, using~$\semi \geq \Omega\left((\max\{n \log(n), 1/\delta)\})^5\right)$.
    Then, possibly after replacing~$\psi^R_\eta$ by~$-\psi^R_\eta$, one obtains~$\normriem{\psi^R_\eta - J_\eta \psi^R_{0,\eta}}^2 \leq 2 \delta$.
    We now get
    \begin{align*}
        \ipriem{\psi^R_\eta}{\psi^R_{\eta'}} & = 
        \ipriem{\psi^R_\eta - J_\eta \psi^R_{0,\eta}\,}{\psi^R_{\eta'} - J_{\eta'} \psi^R_{0,\eta'}} \\
        & + \ipriem{J_\eta \psi^R_{0,\eta}\,}{\psi^R_{\eta'} - J_{\eta'} \psi^R_{0,\eta'}} \\
        & + \ipriem{\psi^R_\eta - J_\eta \psi^R_{0,\eta}\,}{J_{\eta'} \psi^R_{0,\eta'}} \\
        & + \ipriem{J_\eta \psi^R_{0,\eta}\,}{J_{\eta'} \psi^R_{0,\eta'}} \\
        & \geq - 2 \delta - 2 \sqrt{2 \delta} + \ipriem{J_\eta \psi^R_{0,\eta}\,}{J_{\eta'} \psi^R_{0,\eta'}}.
    \end{align*}
    Now choose affine coordinates, so that we can evaluate
    \begin{align*}
        \ipriem{J_\eta \psi^R_{0,\eta}\,}{J_{\eta'} \psi^R_{0,\eta'}} & = \int_\domain J_\eta(x) J_{\eta'}(x) \psi^R_{0,\eta}(x) \psi^R_{0,\eta'}(x) \sqrt{\det g_x} \, dx \\
        & = \int_{\{ x : \norm{x-x_\eta}_{x_\eta},\norm{x-x_{\eta'}}_{x_{\eta'}} \leq 2 \semi^{-2/5} \}} J_\eta(x) J_{\eta'}(x) \psi^R_{0,\eta} \psi^R_{0,\eta'} \sqrt{\det g_x} \, dx.
    \end{align*}
    Now we must be slightly more explicit, and note that~$(\psi^R_{0,\eta})^2(x) = \exp(-\semi \norm{x - x_{\eta}}_{x_\eta}^2) / (\sqrt{(4\pi \semi)^n} \det(g_{x_\eta}))$. 
    On the region~$\norm{x - x_{\eta}}_{x_\eta} \leq r < 1$, we have the bounds (cf. the proof of~\cref{thm:g vs gz})
    \[
        \sqrt{\det(g_x)} \in [(1-r)^n \sqrt{\det(g_{x_\eta})}, (1+r)^n \sqrt{\det(g_{x_\eta})}].
    \]
    We now estimate
    \begin{align*}
        \int_{\{ x : \norm{x-x_\eta}_{x_\eta},\norm{x-x_{\eta'}}_{x_{\eta'}} \leq 2 \semi^{-2/5} \}} J_\eta(x) J_{\eta'}(x) \psi^R_{0,\eta} \psi^R_{0,\eta'} \sqrt{\det g_x} \, dx
        & \geq (1-2\gamma^{-2/5})^{n} \braket{J_\eta \phi_{0,\eta}^R, J_{\eta'} \phi^R_{0,\eta'}},
    \end{align*}
    where~$\phi_{0,\eta}^R = \psi_{0,\eta}^R \sqrt[4]{\det g_{x_\eta}}$ and we write~$\braket{\cdot,\cdot}$ for the~$L^2$-inner product with respect to the current choice of coordinates. Note that $(\phi_{0,\eta}^R)^2$ is the density function of $\mathcal{N}(x_\eta,(2\semi g_{x_\eta})^{-1})$.
    We now rewrite
    \begin{align*}
        \braket{J_\eta \phi_{0,\eta}^R, J_{\eta'} \phi^R_{0,\eta'}} & = 
        \braket{J_\eta \phi_{0,\eta}^R - \phi_{0,\eta}^R, J_{\eta'} \phi^R_{0,\eta'} - \phi_{0,\eta'}^R} \\ 
        & \ + \braket{\phi_{0,\eta}^R, J_{\eta'} \phi^R_{0,\eta'} - \phi_{0,\eta'}^R} \\ 
        & \ + \braket{J_\eta \phi_{0,\eta}^R - \phi_{0,\eta}^R, \phi_{0,\eta'}^R} \\ 
        & \ + \braket{\phi_{0,\eta}^R, \phi_{0,\eta'}^R} \\ 
        & \geq 0 - \norm{J_{\eta'} \phi^R_{0,\eta'} - \phi_{0,\eta'}^R} - \norm{J_{\eta} \phi^R_{0,\eta} - \phi_{0,\eta}^R} + 1 - \norm{ \mathcal{N}(x_\eta,(2\semi g_{x_\eta})^{-1}) - \mathcal{N}(x_{\eta'},(2\semi g_{x_{\eta'}})^{-1}) }_{\TV}
    \end{align*}
    since~$\phi_{0,\eta}^R \geq J_\eta \phi_{0,\eta}^R$ and~$\norm{\phi_{0,\eta}^R} = 1$.
    Finally, let~$t \in [1, 2 \semi^{1/5}/n - 1]$; then
    \[
    \norm{J_{\eta} \phi^R_{0,\eta} - \phi_{0,\eta}^R}^2 
    = \int_\domain (1-J_\eta(x))^2 \phi^R_{0,\eta}(x)^2 \, dx
    \leq \mathbb{P}(\|x\|^2_{2 \semi g_{x_\eta}} \geq 2 \semi^{1/5})
    \leq 2 \exp(- c_0 t n)
    \]
    where the probability is taken with respect to the Gaussian~$(\phi^R_{0,x_\eta})^2$ and we use~\cref{eq:gaussian-conc}.
    This uses that the support of $1 - J_\eta(x)$ is concentrated on the region~$\{x : \norm{x - x_\eta}_{x_\eta} \geq \semi^{-2/5}\}$.
    We conclude that $\ipriem{\psi^R_\eta}{\psi^R_{\eta'}}$ is at least 
    \begin{equation}
      \label{eq:almost-final-riemoverlap-tv}
         -2 \delta - 2 \sqrt{2 \delta} + (1-2\semi^{-2/5})^{n} \left(1 - \norm{ \mathcal{N}(x_\eta,(2\semi g_{x_\eta})^{-1}) - \mathcal{N}(x_{\eta'},(2\semi g_{x_{\eta'}})^{-1}) }_{\TV} - 2 \sqrt{2} \exp(-\frac12 c_0 t n)\right).
    \end{equation}
    When~$\delta \leq 1/10000$, $-2\delta - 2 \sqrt{2 \delta} \geq -0.03$.
    When~$n \leq 0.005 \gamma^{2/5}$, we also obtain~$(1 - 2 \gamma^{-2/5})^n \geq 1 - 2 n \gamma^{-2/5} \geq 0.99$.
    Lastly, $t = \max \{ 10 / (c_0 n), 1 \}$ suffices for~$-\exp(-c_0 t n) \leq 0.00001$, and satisfies~$t+1 \leq 2 \gamma^{1/5} / n$ whenever~$2\gamma^{1/5}/n \geq \max \{ 10 / (c_0 n), 1 \} + 1$; the latter right-hand side is a constant.
    Under all these conditions, \cref{eq:almost-final-riemoverlap-tv} is lower bounded by
    \begin{align*}
      &{-0.03} + 0.99 \, (1 - \norm{ \mathcal{N}(x_\eta,(2\semi g_{x_\eta})^{-1}) - \mathcal{N}(x_{\eta'},(2\semi g_{x_{\eta'}})^{-1}) }_{\TV} - 0.00001) \\
      \geq & \, 0.9 - \norm{ \mathcal{N}(x_\eta,(2\semi g_{x_\eta})^{-1}) - \mathcal{N}(x_{\eta'},(2\semi g_{x_{\eta'}})^{-1}) }_{\TV},
    \end{align*}
    where we have inflated the prefactor of~$0.99$ to~$1$ on the total variation distance.
\end{proof}

We can now invoke known estimates on the distance between Gaussian distributions.
The following estimate follows from an explicit expression for the KL-divergence between two normal distributions and Pinsker's inequality.
\begin{lemma}[{cf.~\cite[Fact~A.3]{Arbas2023TotalVariationDistance_ICML}}]
\label{lem:distanceGaussians}
Let $\Sigma_1,\Sigma_2 \succ 0$ be such that all eigenvalues of $\Sigma_2^{-1/2} \Sigma_1 \Sigma_2^{-1/2}$ are at least $1/2$, then
\begin{equation}
    \norm{\mathcal N(\mu_1, \Sigma_1) - \mathcal N(\mu_2, \Sigma_2)}_{\TV} \leq \frac12 \sqrt{\|\Sigma_2^{-1/2} \Sigma_1 \Sigma_2^{-1/2} - I\|_F^2 + (\mu_2 - \mu_1)^T \Sigma_2^{-1}(\mu_2-\mu_1)}.
\end{equation}
\end{lemma}

Using this lemma, we can bound the distance between the ground states of the harmonic approximations for different $\eta$ and $\eta'$.
\begin{lemma} \label{lem:distanceGround}
    Let $x, y \in D$ be such that $\|y-x\|_x \leq r$ for $r \in (0,1/4)$ and let $\semi>0$. Then
    \begin{align*}
        \norm{\mathcal N(x,(2\semi {g_x}^{1/2})^{-1}) - \mathcal N(y,(2\semi {g_y}^{1/2})^{-1})}_\TV &\leq \frac12 \sqrt{\|g_y^{1/4} g_x^{-1/2} g_y^{1/4} - I\|_F^2 + 2 \semi (y - x)^T g_y^{1/2}(y-x)} \\
        &\leq C \sqrt{n + 2 \semi \|g_x^{-1/2}\|}\, r,
    \end{align*}
    and similarly
    \begin{align*}
        \norm{\mathcal N(x,(2\semi {g_x})^{-1}) - \mathcal N(y,(2\semi {g_y})^{-1})}_\TV &\leq \frac12 \sqrt{\|g_y^{1/2} g_x^{-1} g_y^{1/2} - I\|_F^2 + 2\semi (y - x)^T g_y(y-x)} \\
        &\leq C \sqrt{n + 2\semi}\, r,
    \end{align*}
    where $C>0$ is a universal constant.
\end{lemma}
\begin{proof}
     We first establish the bound on $\norm{\mathcal N(x,(\semi g_x)^{-1}) - \mathcal N(y,(\semi g_y)^{-1})}_\TV$. Since $\|y-x\|_x \leq r$, we get from self-concordance (\cref{eq:self-concord-1}) that $(1-r)^2 g_x \preceq g_y \preceq (1-r)^{-2} g_x$, and thus $(1-r)^2 g_y^{-1} \preceq g_x^{-1} \preceq (1-r)^{-2} g_y^{-1}$.
    From this it follows that
    \begin{align} \label{eq:spectral est}
        (1-r)^2 I \preceq g_y^{1/2} g_x^{-1} g_y^{1/2} \preceq \frac{1}{(1-r)^2} I, 
    \end{align}
    and thus 
    $\|g_y^{1/2} g_x^{-1} g_y^{1/2} - I\|_F^2 = O( nr^2 )$ and all eigenvalues of $g_y^{1/2} g_x^{-1} g_y^{1/2}$ are at least $1/2$.
    We may therefore apply \cref{lem:distanceGaussians} to obtain the first inequality of the lemma. The second inequality follows from the above bound on $\|g_y^{1/2} g_x^{-1} g_y^{1/2} - I\|_F^2$ and \eqref{eq:self-concord-1}.

    We finally bound $\norm{\mathcal N(x,(2 \semi {g_x}^{1/2})^{-1}) - \mathcal N(y,(2 \semi {g_y}^{1/2})^{-1})}_\TV$. Analogously to \cref{eq:spectral est} one has
    \[
    (1-r) I \preceq g_y^{1/4} g_x^{-1/2} g_y^{1/4} \preceq \frac{1}{(1-r)} I,
    \]
    showing that $\|g_y^{1/4} g_x^{-1/2} g_y^{1/4} - I\|_F^2 = O( nr^2 )$ and that all eigenvalues of $g_y^{1/4} g_x^{-1/2} g_y^{1/4}$ are at least $1/2$. \Cref{lem:distanceGaussians} therefore shows
    \[
    \norm{\mathcal N(x,(\semi {g_x}^{1/2})^{-1}) - \mathcal N(y,(\semi {g_y}^{1/2})^{-1})}_\TV \leq \frac12 \sqrt{\|g_y^{1/4} g_x^{-1/2} g_y^{1/4} - I\|_F^2 + \semi (y - x)^T g_y^{1/2}(y-x)}.
    \]
    We further upper bound the right hand side using $\|g_y^{1/4} g_x^{-1/2} g_y^{1/4} - I\|_F^2 = O( nr^2 )$, \eqref{eq:self-concord-1} to show $g_y^{1/2} \preceq \frac{1}{1-r} g_x^{1/2}$, and the inequality $g_x^{1/2} \preceq g_x / \sqrt{\lambda_{\min}(g_x)} = \|g_x^{-1/2}\| g_x$.
\end{proof}

Combining \cref{lem:distanceGround} with \cref{eq:safe eta} immediately leads to the following two corollaries, establishing a safe choice for updating $\eta$.

\begin{corollary}[Euclidean ground state overlap] \label{cor:eucl-gs-overlap2}
If $\eta \leq \eta' \leq \Big(1+ O\big( (n+2\semi\|g_{x_\eta}^{-1/2}\|)\vartheta\big)^{-1/2})\Big)\eta$ and
\[
\semi \in \Omega\big(\max_{s \in \{\eta,\eta'\}}(\sqrt{\norm{g_{x_s}}} \, \max\{n \log(n \|g_{x_s}\|), \sqrt{\norm{g_{x_s}}/\lambda_{\min}(g_{x_\eta})}\})^5\big)
\]
then
\[
\abs{\braket{\psi_0^E(\eta),\psi_0^E(\eta')}}
\geq 9/10 - \norm{\mathcal N(x_\eta,(\semi {g_{x_\eta}^{1/2}})^{-1}) - \mathcal N(x_{\eta'},(\semi {g_{x_\eta'}^{1/2}})^{-1})}_\TV
\geq 4/5.
\]
\end{corollary}

\begin{corollary}[Riemannian ground state overlap] \label{cor:riem-gs-overlap2}
If
\[
\eta
\leq \eta'
\leq \left(1+ O( (n+2\semi)\vartheta)^{-1/2})\right)\eta
\quad \text{ and } \quad
\semi \in \Omega\left((n \log(n))^5\right)
\]
then
\[
\abs{\ipriem{\psi_0^R(\eta)}{\psi_0^R(\eta')}}
\geq 9/10 - \norm{\mathcal N(x_\eta,(\semi {g_{x_\eta}})^{-1}) - \mathcal N(x_{\eta'},(\semi {g_{x_\eta'}})^{-1})}_\TV
\geq 4/5.
\]
\end{corollary}

\subsection{Ground state projection for unbounded Hamiltonians} \label{sec:gs-proj-unbounded}
Here we describe how to implement a rotation around the ground state $\ket{\psi_0}$ of an unbounded Hamiltonian $H$. For a bounded Hamiltonian - a matrix - a natural approach is to use quantum phase estimation and the phase register to apply the $\pi/3$-rotation around the ground state. By normalizing the Hamiltonian appropriately, one can ensure for example that the phases are contained in the interval $[0,\pi)$, i.e., the spectrum of $e^{iH}$ does not wrap around the unit circle. The spectral gap of the Hamiltonian then determines the precision with which one has to run phase estimation. For unbounded operators $H$ the spectrum of $e^{iH}$ necessarily wraps around the unit circle.
This leads to unwanted aliasing effects, and we therefore cannot follow the same approach as for bounded operators.

\paragraph{Rotating around the ground state of an unbounded Hamiltonian.}

The approach we follow here is to approximate a ground state projector of an unbounded Hamiltonian through imaginary time evolution, which we implement via a continuous-variable version of the linear combination of unitaries technique.
This approach is inspired by ideas in \cite{chowdhury17quantum,apers2022quadratic}, though these works only consider finite-dimensional systems.

Let $V$ be the state space of $H$ and let $W = L^2(\R)$ be a one-dimensional continuous-variable state space. We then consider the Hamiltonian $H'$ acting on $V \otimes W$ defined as
\[
H'
= H \otimes \hat z,
\]
where $\hat z$ the position operator on $W$.
Note that
\[
e^{i H'} \ket{\phi} \ket{z} = \left( e^{i H} \right)^z \ket{\phi} \ket{z} = \left( e^{i Hz} \right) \ket{\phi} \ket{z},
\]
so that $H'$ is a continuous-variable version of a controlled-unitary. We next define the Gaussian state
\[
\ket{\psi_g}
= \frac{1}{(2\pi)^{1/4}} \int e^{-z^2/4} \ket{z} \, \mathrm{d}z \in W.
\]
We show that we can approximately rotate around the ground state of $H$ by using Hamiltonian simulation of $H'$ and a $\pi/3$-rotation around $\ket{\psi_g}$.

\medskip

For an arbitrary state $\ket{\phi}$ on $V$, we start the algorithm from $\ket{\phi} \ket{\psi_g}$ and apply Hamiltonian evolution with $H'$ for time $\sqrt{2t}$.
This yields the state
\[
e^{i H' \sqrt{2t}} \ket{\phi} \ket{\psi_g}
= \frac{1}{(2\pi)^{1/4}} \int e^{-z^2/4} e^{-i \sqrt{2t} H z} \ket{\phi} \ket{z} \, \mathrm{d} z.
\]
To interpret the latter, we use the Hubbard-Stratonovich transform, which for $x \in \R$ states that
\[
e^{-x^2/2}
= \frac{1}{\sqrt{2\pi}} \int e^{-z^2/2} e^{-ixz} \, \mathrm{d} z.
\]
Setting $x = \sqrt{2t} H$, this shows that
\[
(I \otimes \ket{\psi_g} \bra{\psi_g}) e^{i H' \sqrt{2t}} \ket{\phi} \ket{\psi_g}
= \left( e^{-t H^2} \ket{\phi} \right) \ket{\psi_g}.
\]
When $t \geq c\log(1/\delta)/\gap(H)^2$ for a suitable constant $c>0$, we moreover have
\begin{align} \label{eq:approx-proj}
\left( e^{-t H^2} \ket{\phi} \right) \ket{\psi_g}
\approx_\delta \big( \ket{\psi_0}\braket{\psi_0|\phi} \big) \ket{\psi_g},
\end{align}
where the approximation is in the $\ell_2$-distance. 
Here we crucially use the assumption that the ground state $\ket{\psi_0}$ has energy zero, i.e., that $H\ket{\psi_0}=0$.\footnote{If we know an $\eps_0$-additive estimate $\bar \lambda$ of $\lambda_0(H)$, then we can apply the same algorithm to $H-\bar \lambda$, incurring a sub-normalization of $e^{-t\eps_0^2}$. Assumption 2 ensures that this sub-normalization amounts to an overall error $\eps$ over the $T$ iterations (via a union bound).} 

We finally use a $\pi/3$-rotation around $\ket{\psi_g}$ to implement the $\pi/3$-rotation around $\ket{\psi_0}$:
\begin{align*}
&e^{-i H' \sqrt{2t}}
\left( I \otimes \left( I + (\omega-1) \ket{\psi_g}\bra{\psi_g} \right) \right) e^{i H' \sqrt{2t}} \ket{\phi} \ket{\psi_g} \\
&\hspace{4cm} = (\omega-1) e^{-i H' \sqrt{2t}} ( I \otimes \ket{\psi_g}\bra{\psi_g} ) e^{i H' \sqrt{2t}} \ket{\phi} \ket{\psi_g} + \ket{\phi} \ket{\psi_g} \\
&\hspace{4cm} = (\omega-1)  e^{-i H' \sqrt{2t}}\left( e^{-t H^2} \ket{\phi}\ket{\psi_g} \right)  + \ket{\phi} \ket{\psi_g} \\
&\hspace{4cm} \approx_\delta \left[ \left( (\omega-1) \ket{\psi_0}\bra{\psi_0} + I \right) \ket{\phi} \right] \ket{\psi_g},
\end{align*}
where in the last line we use \cref{eq:approx-proj} and the identity $e^{-iH'\sqrt{2t}} \ket{\psi_0}\ket{\psi_g} = \ket{\psi_0}\ket{\psi_g}$.

Summarizing, we can implement a $\delta$-close approximation of $R = I + (\omega-1) \ket{\psi_0}\bra{\psi_0}$ by simulating the Hamiltonian $H'$ for a time $O(\sqrt{\log(1/\delta)}/\Delta_\ell)$ and using a $\pi/3$-rotation around the Gaussian state $\ket{\psi_g}$. Similarly, replacing the $\pi/3$-rotation around the Gaussian state by a $-\pi/3$-rotation, one can approximately implement $R^*$.

\bibliographystyle{alpha}
\bibliography{biblio}

\newcommand{\etalchar}[1]{$^{#1}$}
\providecommand{\noopsort}[1]{}
\begin{thebibliography}{{\noopsort{Apeldoorn}v}AGGdW20}

\bibitem[AAA{\etalchar{+}}24]{Abbas2024Challenges}
Amira Abbas, Andris Ambainis, Brandon Augustino, Andreas Bärtschi, Harry
  Buhrman, Carleton Coffrin, Giorgio Cortiana, Vedran Dunjko, Daniel~J. Egger,
  Bruce~G. Elmegreen, Nicola Franco, Filippo Fratini, Bryce Fuller, Julien
  Gacon, Constantin Gonciulea, Sander Gribling, Swati Gupta, Stuart Hadfield,
  Raoul Heese, Gerhard Kircher, Thomas Kleinert, Thorsten Koch, Georgios
  Korpas, Steve Lenk, Jakub Marecek, Vanio Markov, Guglielmo Mazzola, Stefano
  Mensa, Naeimeh Mohseni, Giacomo Nannicini, Corey O’Meara, Elena
  Peña~Tapia, Sebastian Pokutta, Manuel Proissl, Patrick Rebentrost, Emre
  Sahin, Benjamin C.~B. Symons, Sabine Tornow, Victor Valls, Stefan Woerner,
  Mira~L. Wolf-Bauwens, Jon Yard, Sheir Yarkoni, Dirk Zechiel, Sergiy Zhuk, and
  Christa Zoufal.
\newblock Challenges and opportunities in quantum optimization.
\newblock {\em Nature Reviews Physics}, 6:718--735, 2024.

\bibitem[AAL23]{Arbas2023TotalVariationDistance_ICML}
Jamil Arbas, Hassan Ashtiani, and Christopher Liaw.
\newblock Polynomial time and private learning of unbounded gaussian mixture
  models.
\newblock In {\em Proceedings of the 40th International Conference on Machine
  Learning}, ICML'23, 2023.

\bibitem[AC11]{andrews11fundamentalgap}
Ben Andrews and Julie Clutterbuck.
\newblock Proof of the fundamental gap conjecture.
\newblock {\em Journal of the American Mathematical Society}, 24(3):899--916,
  2011.

\bibitem[AC21]{ahn2021discrete}
Kwangjun Ahn and Sinho Chewi.
\newblock Efficient constrained sampling via the mirror-{L}angevin algorithm.
\newblock In {\em Proceedings of the 35th International Conference on Neural
  Information Processing Systems}, NIPS '21, 2021.

\bibitem[ACNR22]{apers2022quadratic}
Simon Apers, Shantanav Chakraborty, Leonardo Novo, and J{\'e}r{\'e}mie Roland.
\newblock Quadratic speedup for spatial search by continuous-time quantum walk.
\newblock {\em Physical Review Letters}, 129(16):160502, 2022.

\bibitem[AG23]{apers2023quantum}
Simon Apers and Sander Gribling.
\newblock Quantum speedups for linear programming via interior point methods,
  2023.

\bibitem[AK16]{arora2016combinatorial}
Sanjeev Arora and Satyen Kale.
\newblock A combinatorial, primal-dual approach to semidefinite programs.
\newblock {\em Journal of the ACM (JACM)}, 63(2):1--35, 2016.

\bibitem[ALN{\etalchar{+}}24]{augustino2024quantumcentralpathalgorithm}
Brandon Augustino, Jiaqi Leng, Giacomo Nannicini, Tamás Terlaky, and Xiaodi
  Wu.
\newblock A quantum central path algorithm for linear optimization, 2024.
\newblock arXiv:2311.03977.

\bibitem[ANTZ23]{augustino2021quantum}
Brandon Augustino, Giacomo Nannicini, Tam{\'a}s Terlaky, and Luis~F Zuluaga.
\newblock Quantum interior point methods for semidefinite optimization.
\newblock {\em Quantum}, 7:1110, 2023.

\bibitem[{\noopsort{Apeldoorn}v}AG19a]{vanapeldoornImprovedSDP}
Joran {\noopsort{Apeldoorn}v}an~Apeldoorn and Andr{\'a}s Gily{\'e}n.
\newblock {Improvements in Quantum SDP-Solving with Applications}.
\newblock In {\em 46th International Colloquium on Automata, Languages, and
  Programming (ICALP 2019)}, volume 132 of {\em Leibniz International
  Proceedings in Informatics (LIPIcs)}, pages 99:1--99:15, 2019.

\bibitem[{\noopsort{Apeldoorn}v}AG19b]{vanapeldoorn2019quantum}
Joran {\noopsort{Apeldoorn}v}an~Apeldoorn and András Gilyén.
\newblock Quantum algorithms for zero-sum games, 2019.
\newblock arXiv:1904.03180.

\bibitem[{\noopsort{Apeldoorn}v}AGGdW20]{vAGGdW:quantumSDP}
Joran {\noopsort{Apeldoorn}v}an~Apeldoorn, Andr{\'a}s Gily{\'e}n, Sander
  Gribling, and Ronald de~Wolf.
\newblock Quantum {SDP}-solvers: Better upper and lower bounds.
\newblock {\em Quantum}, 4(230), 2020.
\newblock Earlier version in FOCS'17.

\bibitem[Aub82]{aubinNonlinearAnalysisManifolds1982}
Thierry Aubin.
\newblock {\em Nonlinear {{Analysis}} on {{Manifolds}}. {{Monge-Amp{\`e}re
  Equations}}}, volume 252 of {\em Grundlehren Der Mathematischen
  {{Wissenschaften}}}.
\newblock Springer, New York, NY, 1982.

\bibitem[BCCRK24]{brenner2024factoring}
Lukas Brenner, Libor Caha, Xavier Coiteux-Roy, and Robert Koenig.
\newblock Factoring an integer with three oscillators and a qubit.
\newblock {\em arXiv preprint arXiv:2412.13164}, 2024.

\bibitem[BGJ{\etalchar{+}}23]{bouland2023}
Adam Bouland, Yosheb Getachew, Yujia Jin, Aaron Sidford, and Kevin Tian.
\newblock Quantum speedups for zero-sum games via improved dynamic gibbs
  sampling.
\newblock In {\em Proceedings of the 40th International Conference on Machine
  Learning}, ICML'23, 2023.

\bibitem[BGL14]{bakry2014analysis}
Dominique Bakry, Ivan Gentil, and Michel Ledoux.
\newblock {\em {Analysis and Geometry of Markov Diffusion Operators}}, volume
  348 of {\em Grundlehren der mathematischen Wissenschaften}.
\newblock Springer, 2014.

\bibitem[BKL{\etalchar{+}}19]{brandão2019quantum}
Fernando G. S.~L. Brandão, Amir Kalev, Tongyang Li, Cedric Yen-Yu Lin,
  Krysta~M. Svore, and Xiaodi Wu.
\newblock Quantum {SDP} solvers: Large speed-ups, optimality, and applications
  to quantum learning.
\newblock In {\em Proceedings of the 46th International Colloquium on Automata,
  Languages and Programming (ICALP 2019)}, volume 132, pages 27:1--27:14, 2019.

\bibitem[BL76]{brascamp1976inequality}
Herm~Jan Brascamp and Elliott~H Lieb.
\newblock {On extensions of the Brunn-Minkowski and Prékopa-Leindler theorems,
  including inequalities for log concave functions, and with an application to
  the diffusion equation}.
\newblock {\em Journal of Functional Analysis}, 22(4):366--389, 1976.

\bibitem[BMS02]{braverman2002esa}
M~Braverman, O~Milatovic, and M~Shubin.
\newblock Essential self-adjointness of {Schrödinger}-type operators on
  manifolds.
\newblock {\em Russian Mathematical Surveys}, 57(4):641, aug 2002.

\bibitem[Bra25]{braverman2025semiclassicalweyllawcomplete}
Maxim Braverman.
\newblock {The semi-classical Weyl law on complete manifolds}, 2025.
\newblock arXiv:2505.12157.

\bibitem[Bre11]{brezisFunctionalAnalysisSobolev2011}
Haim Brezis.
\newblock {\em Functional {{Analysis}}, {{Sobolev Spaces}} and {{Partial
  Differential Equations}}}.
\newblock Springer, New York, NY, 2011.

\bibitem[BS17]{brandaoSDP17}
F.~L. Brandao and K.~M. Svore.
\newblock Quantum speed-ups for solving semidefinite programs.
\newblock In {\em IEEE 58th Annual Symposium on Foundations of Computer Science
  (FOCS)}, pages 415--426, 2017.

\bibitem[BSBN02]{bartlett02efficientgauss}
Stephen~D. Bartlett, Barry~C. Sanders, Samuel~L. Braunstein, and Kae Nemoto.
\newblock Efficient classical simulation of continuous variable quantum
  information processes.
\newblock {\em Physical Review Letters}, 88:097904, Feb 2002.

\bibitem[Bub15]{bubeck2015convex}
S\'ebastien Bubeck.
\newblock {Convex Optimization: Algorithms and Complexity}.
\newblock {\em Foundations and Trends® in Machine Learning}, 8(3-4):231--357,
  2015.

\bibitem[BV04]{boyd2004convex}
Stephen Boyd and Lieven Vandenberghe.
\newblock {\em Convex Optimization}.
\newblock Cambridge University Press, 2004.

\bibitem[CEL{\etalchar{+}}25]{Chewi2025AnalysisLangevinMonteCarlo}
Sinho Chewi, Murat~A. Erdogdu, Mufan Li, Ruoqi Shen, and Matthew~S. Zhang.
\newblock Analysis of {Langevin Monte Carlo from Poincaré to log-Sobolev}.
\newblock {\em Foundations of Computational Mathematics}, 25:1345--1395, 2025.

\bibitem[CFKS87]{CyFKS1987}
Hans~L. Cycon, Richard~G. Froese, Werner Kirsch, and Barry Simon.
\newblock {\em Schr{\"o}dinger Operators: With Application to Quantum Mechanics
  and Global Geometry}.
\newblock Theoretical and Mathematical Physics. Springer Berlin Heidelberg,
  1987.

\bibitem[CGJ19]{chakraborty19power}
Shantanav Chakraborty, Andr\'{a}s Gily\'{e}n, and Stacey Jeffery.
\newblock {The Power of Block-Encoded Matrix Powers: Improved Regression
  Techniques via Faster Hamiltonian Simulation}.
\newblock In Christel Baier, Ioannis Chatzigiannakis, Paola Flocchini, and
  Stefano Leonardi, editors, {\em 46th International Colloquium on Automata,
  Languages, and Programming (ICALP 2019)}, volume 132 of {\em Leibniz
  International Proceedings in Informatics (LIPIcs)}, pages 33:1--33:14, 2019.

\bibitem[Che25]{chewi2025book}
Sinho Chewi.
\newblock Log-concave sampling, 2025.
\newblock Book draft available at \url{https://chewisinho.github.io/}.

\bibitem[CHW{\etalchar{+}}25]{chakrabarti2025speedups}
Shouvanik Chakrabarti, Dylan Herman, Jacob Watkins, Enrico Fontana, Brandon
  Augustino, Junhyung~Lyle Kim, and Marco Pistoia.
\newblock On speedups for convex optimization via quantum dynamics.
\newblock {\em arXiv preprint arXiv:2503.24332}, 2025.

\bibitem[CJMM25]{chabaud25bosonic}
Ulysse Chabaud, Michael Joseph, Saeed Mehraban, and Arsalan Motamedi.
\newblock {Bosonic Quantum Computational Complexity}.
\newblock In Raghu Meka, editor, {\em 16th Innovations in Theoretical Computer
  Science Conference (ITCS 2025)}, volume 325 of {\em Leibniz International
  Proceedings in Informatics (LIPIcs)}, pages 33:1--33:19. Schloss Dagstuhl --
  Leibniz-Zentrum f{\"u}r Informatik, 2025.
\newblock Full version at arXiv:2410.04274.

\bibitem[CLGL{\etalchar{+}}20]{chewi2020exponential}
Sinho Chewi, Thibaut Le~Gouic, Chen Lu, Tyler Maunu, Philippe Rigollet, and
  Austin Stromme.
\newblock Exponential ergodicity of mirror-{L}angevin diffusions.
\newblock In {\em Proceedings of the 34th International Conference on Neural
  Information Processing Systems}, NIPS '20, 2020.

\bibitem[CLL{\etalchar{+}}22]{childs2022quantum}
Andrew~M Childs, Jiaqi Leng, Tongyang Li, Jin-Peng Liu, and Chenyi Zhang.
\newblock Quantum simulation of real-space dynamics.
\newblock {\em Quantum}, 6:860, 2022.

\bibitem[CLW{\etalchar{+}}25]{Chen2025QuantumLangevinDynamics}
Zherui Chen, Yuchen Lu, Hao Wang, Yizhou Liu, and Tongyang Li.
\newblock Quantum {Langevin} dynamics for optimization.
\newblock {\em Communications in Mathematical Physics}, 406:52, 2025.
\newblock Published online 17 February 2025.

\bibitem[CS17]{chowdhury17quantum}
Anirban~Narayan Chowdhury and Rolando~D. Somma.
\newblock Quantum algorithms for {Gibbs} sampling and hitting-time estimation.
\newblock {\em Quantum Information \& Computation}, 17(1–2):41–64, 2017.

\bibitem[CS24]{ChewiStromme2024Ballistic}
Sinho Chewi and Austin~J. Stromme.
\newblock {The ballistic limit of the log-Sobolev constant equals the
  Polyak–Łojasiewicz constant}.
\newblock {\em arXiv preprint}, page arXiv:2411.11415v1, 2024.

\bibitem[CZS22]{Cheng2022riemannian}
Xiang Cheng, Jingzhao Zhang, and Suvrit Sra.
\newblock {Efficient sampling on Riemannian manifolds via Langevin MCMC}.
\newblock In {\em Proceedings of the 36th International Conference on Neural
  Information Processing Systems}, NIPS '22, 2022.

\bibitem[Dav89]{Davies1989}
E.~B. Davies.
\newblock {\em Heat Kernels and Spectral Theory}.
\newblock Cambridge Tracts in Mathematics. Cambridge University Press, 1989.

\bibitem[Eva10]{EvansPDE}
Lawrence~C. Evans.
\newblock {\em Partial Differential Equations}, volume~19 of {\em Graduate
  Studies in Mathematics}.
\newblock American Mathematical Society, Providence, RI, 2 edition, 2010.

\bibitem[FGGS00]{farhi2000quantumAdiabatic}
Edward Farhi, Jeffrey Goldstone, Sam Gutmann, and Michael Sipser.
\newblock Quantum computation by adiabatic evolution, 2000.

\bibitem[FOP05]{ferraro2005gaussianstatescontinuousvariable}
Alessandro Ferraro, Stefano Olivares, and Matteo G.~A. Paris.
\newblock Gaussian states in continuous variable quantum information, 2005.

\bibitem[GC11]{girolami2011riemann}
Mark Girolami and Ben Calderhead.
\newblock {Riemann manifold Langevin and Hamiltonian Monte Carlo methods}.
\newblock {\em Journal of the Royal Statistical Society Series B: Statistical
  Methodology}, 73(2):123--214, 2011.

\bibitem[GJLW23]{gao2023logarithmic}
Minbo Gao, Zhengfeng Ji, Tongyang Li, and Qisheng Wang.
\newblock Logarithmic-regret quantum learning algorithms for zero-sum games.
\newblock {\em Advances in Neural Information Processing Systems},
  36:31177--31203, 2023.

\bibitem[GKP01]{GKPoscillator}
Daniel Gottesman, Alexei Kitaev, and John Preskill.
\newblock Encoding a qubit in an oscillator.
\newblock {\em Phys. Rev. A}, 64:012310, Jun 2001.

\bibitem[GSLW19]{gilyen2019qsvt}
Andr{\'a}s Gily{\'e}n, Yuan Su, Guang~Hao Low, and Nathan Wiebe.
\newblock Quantum singular value transformation and beyond: exponential
  improvements for quantum matrix arithmetics.
\newblock In {\em Proceedings of the 51st Annual ACM SIGACT Symposium on Theory
  of Computing}, pages 193--204, 2019.

\bibitem[GTC19]{ge2019faster}
Yimin Ge, Jordi Tura, and J~Ignacio Cirac.
\newblock Faster ground state preparation and high-precision ground energy
  estimation with fewer qubits.
\newblock {\em Journal of Mathematical Physics}, 60(2), 2019.

\bibitem[GV22]{gatmiry2022convergenceriemannianlangevinalgorithm}
Khashayar Gatmiry and Santosh~S. Vempala.
\newblock {Convergence of the Riemannian Langevin algorithm}, 2022.
\newblock arXiv:2204.10818.

\bibitem[HKS89]{holley1989asymptotics}
Richard~A Holley, Shigeo Kusuoka, and Daniel~W Stroock.
\newblock Asymptotics of the spectral gap with applications to the theory of
  simulated annealing.
\newblock {\em Journal of Functional Analysis}, 83(2):333--347, 1989.

\bibitem[HS12]{hislop2012introduction}
Peter~D Hislop and Israel~Michael Sigal.
\newblock {\em Introduction to spectral theory: With applications to
  Schr{\"o}dinger operators}, volume 113.
\newblock Springer Science \& Business Media, 2012.

\bibitem[Hsu02]{Hsu2002stochmanifold}
Elton~P. Hsu.
\newblock {\em Stochastic Analysis on Manifolds}.
\newblock Number~38 in Graduate Studies in Mathematics. American Mathematical
  Society, 2002.

\bibitem[Hö03]{hormanderVol1}
Lars Hörmander.
\newblock {\em {The Analysis of Linear Partial Differential Operators I}}.
\newblock Classics in Mathematics. Springer Berlin, Heidelberg, 2 edition,
  2003.

\bibitem[JKO98]{jordan1998variational}
Richard Jordan, David Kinderlehrer, and Felix Otto.
\newblock {The variational formulation of the Fokker--Planck equation}.
\newblock {\em SIAM Journal on Mathematical Analysis}, 29(1):1--17, 1998.

\bibitem[KNS16]{karimi2016linear}
Hamed Karimi, Julie Nutini, and Mark Schmidt.
\newblock Linear convergence of gradient and proximal-gradient methods under
  the {Polyak-{\L}ojasiewicz} condition.
\newblock In Paolo Frasconi, Niels Landwehr, Giuseppe Manco, and Jilles
  Vreeken, editors, {\em Machine Learning and Knowledge Discovery in
  Databases}, pages 795--811, Cham, 2016. Springer International Publishing.

\bibitem[KP20]{kerenidis2020quantum}
Iordanis Kerenidis and Anupam Prakash.
\newblock A quantum interior point method for {LPs} and {SDPs}.
\newblock {\em ACM Transactions on Quantum Computing}, 1(1):1--32, 2020.

\bibitem[KS99]{kondratevDiscretenessSpectrumSchrodinger1999}
Vladimir Kondrat'Ev and Mikhail Shubin.
\newblock Discreteness of spectrum for the {{Schr{\"o}dinger}} operators on
  manifolds of bounded geometry.
\newblock In J{\"u}rgen Rossmann, Peter Tak{\'a}{\v c}, and G{\"u}nther
  Wildenhain, editors, {\em The {{Maz}}'ya {{Anniversary Collection}}},
  Operator {{Theory}}: {{Advances}} and {{Applications}}, pages 185--226,
  Basel, 1999. Birkh{\"a}user.

\bibitem[LB99]{LB99model}
Seth Lloyd and Samuel~L. Braunstein.
\newblock Quantum computation over continuous variables.
\newblock {\em Phys. Rev. Lett.}, 82:1784--1787, Feb 1999.

\bibitem[LDCL25]{leng2025operatorlevelquantumaccelerationnonlogconcave}
Jiaqi Leng, Zhiyan Ding, Zherui Chen, and Lin Lin.
\newblock Operator-level quantum acceleration of non-logconcave sampling, 2025.
\newblock arXiv:2505.05301.

\bibitem[Lee12]{leeIntroductionSmoothManifolds2012}
John~M. Lee.
\newblock {\em Introduction to Smooth Manifolds}, volume 218 of {\em Graduate
  {{Texts}} in {{Mathematics}}}.
\newblock Springer, New York, NY, 2012.

\bibitem[Lee18]{leeIntroductionRiemannianManifolds2018}
John~M. Lee.
\newblock {\em Introduction to {{Riemannian Manifolds}}}.
\newblock Number 176 in Graduate {{Texts}} in {{Mathematics}}. Springer
  International Publishing : Imprint: Springer, Cham, 2nd ed. 2018 edition,
  2018.

\bibitem[LHLW23]{leng2023quantumhamiltoniandescent}
Jiaqi Leng, Ethan Hickman, Joseph Li, and Xiaodi Wu.
\newblock Quantum {Hamiltonian} descent, 2023.
\newblock arXiv:2303.01471.

\bibitem[LLV20]{laddha2020strongsampling}
Aditi Laddha, Yin~Tat Lee, and Santosh Vempala.
\newblock Strong self-concordance and sampling.
\newblock In {\em Proceedings of the 52nd Annual ACM SIGACT Symposium on Theory
  of Computing}, STOC 2020, page 1212–1222, New York, NY, USA, 2020.
  Association for Computing Machinery.

\bibitem[MFT24]{mohammadisiahroudi2024efficient}
Mohammadhossein Mohammadisiahroudi, Ramin Fakhimi, and Tam{\'a}s Terlaky.
\newblock Efficient use of quantum linear system algorithms in inexact
  infeasible {IPMs} for linear optimization.
\newblock {\em Journal of Optimization Theory and Applications},
  202(1):146--183, 2024.

\bibitem[MFWT25]{mohammadisiahroudi2023inexact}
Mohammadhossein Mohammadisiahroudi, Ramin Fakhimi, Zeguan Wu, and Tam{\'a}s
  Terlaky.
\newblock An inexact feasible interior point method for linear optimization
  with high adaptability to quantum computers.
\newblock {\em SIAM Journal On Optimization (to appear)}, 2025.

\bibitem[ML23]{mozgunov23adiabatic}
Evgeny Mozgunov and Daniel~A. Lidar.
\newblock Quantum adiabatic theorem for unbounded hamiltonians with a cutoff
  and its application to superconducting circuits.
\newblock {\em Philosophical Transactions of the Royal Society A: Mathematical,
  Physical and Engineering Sciences}, 381(2241):20210407, 2023.

\bibitem[Nes18]{nesterov2018lectures}
Yurii Nesterov.
\newblock {\em Lectures on Convex Optimization}.
\newblock Springer Optimization and Its Applications. Springer, 2nd edition,
  2018.

\bibitem[NN94]{nesterov1994interior}
Yurii~E. Nesterov and Arkadi Nemirovskii.
\newblock {\em Interior-{Point} {Polynomial} {Algorithms} in {Convex}
  {Programming}}, volume~13.
\newblock SIAM, 1994.

\bibitem[NT02]{nesterov2002geometry}
Yu.~E. Nesterov and M.~J. Todd.
\newblock {On the Riemannian Geometry Defined by Self-Concordant Barriers and
  Interior-Point Methods}.
\newblock {\em Foundations of Computational Mathematics}, 2:333--361, 2002.

\bibitem[NY83]{NemirovskyYudin1983}
A.~S. Nemirovsky and D.~B. Yudin.
\newblock {\em Problem Complexity and Method Efficiency in Optimization}.
\newblock Wiley-Interscience Series in Discrete Mathematics. John Wiley \&
  Sons, New York, 1983.

\bibitem[Ole94]{Oleinik1994}
I.~M. Oleinik.
\newblock On the connection of the classical and quantum mechanical
  completeness of a potential at infinity on complete riemannian manifolds.
\newblock {\em Mathematical Notes}, 55:380--386, 1994.

\bibitem[Ren01]{Renegar01}
James Renegar.
\newblock {\em A Mathematical View of Interior-Point Methods in Convex
  Optimization}.
\newblock MOS-SIAM Series on Optimization. SIAM, USA, 2001.

\bibitem[SBC16]{su2016accelerated}
Weijie Su, Stephen Boyd, and Emmanuel~J. Cand\`{e}s.
\newblock A differential equation for modeling {Nesterov}'s accelerated
  gradient method: theory and insights.
\newblock {\em Journal of Machine Learning Research}, 17(1):5312–5354,
  January 2016.

\bibitem[Sch12]{schmudgenUnbdd}
Konrad Schm\"udgen.
\newblock {\em {Unbounded Self-adjoint Operators on Hilbert Space}}.
\newblock Graduate Texts in Mathematics. Springer Dordrecht, 1 edition, 2012.

\bibitem[Shu96]{shubin1996}
M.~A. Shubin.
\newblock Semiclassical asymptotics on covering manifolds and morse
  inequalities.
\newblock {\em Geometric and Functional Analysis}, 6:370--409, 1996.

\bibitem[Sim83]{simon1983semiclassical}
Barry Simon.
\newblock Semiclassical analysis of low lying eigenvalues. i. non-degenerate
  minima : asymptotic expansions.
\newblock {\em Annales de l'I.H.P. Physique théorique}, 38(3):295--308, 1983.

\bibitem[Sim84]{simon1984tunneling}
Barry Simon.
\newblock {Semiclassical Analysis of Low Lying Eigenvalues, II. Tunneling}.
\newblock {\em Annals of Mathematics}, 120(1):89--118, 1984.

\bibitem[THSA23]{GarciaTrillos2023OptimizationSampling}
N.~Garc\'{\i}a Trillos, B.~Hosseini, and D.~Sanz-Alonso.
\newblock {From Optimization to Sampling Through Gradient Flows}.
\newblock {\em Notices of the American Mathematical Society}, 70(6):905--917,
  2023.

\bibitem[Ura93]{urakawa1991laplace}
Hajime Urakawa.
\newblock {Geometry of Laplace-Beltrami Operator on a Complete Riemannian
  Manifold}.
\newblock In {\em Progress in Differential Geometry}, volume~22 of {\em
  Advanced Studies in Pure Mathematics}, pages 347--406. Mathematical Society
  of Japan, 1993.

\bibitem[Ver18]{vershynin2018high}
Roman Vershynin.
\newblock {\em High-dimensional probability: An introduction with applications
  in data science}, volume~47.
\newblock Cambridge university press, 2018.

\bibitem[WA08]{wocjan2008speedup}
Pawel Wocjan and Anura Abeyesinghe.
\newblock Speedup via quantum sampling.
\newblock {\em Physical Review A}, 78(4):042336, 2008.

\bibitem[Wit82]{witten1982super}
Edward Witten.
\newblock {Supersymmetry and Morse theory}.
\newblock {\em Journal of Differential Geometry}, 17(4):661--692, 1982.

\bibitem[WPGP{\etalchar{+}}12]{weedbrook12gauss}
Christian Weedbrook, Stefano Pirandola, Ra\'ul Garc\'{\i}a-Patr\'on, Nicolas~J.
  Cerf, Timothy~C. Ralph, Jeffrey~H. Shapiro, and Seth Lloyd.
\newblock Gaussian quantum information.
\newblock {\em Rev. Mod. Phys.}, 84:621--669, May 2012.

\bibitem[WWJ16]{wibisono2016variational}
Andre Wibisono, Ashia~C. Wilson, and Michael~I. Jordan.
\newblock A variational perspective on accelerated methods in optimization.
\newblock {\em Proceedings of the National Academy of Sciences},
  113(47):E7351--E7358, 2016.

\bibitem[YLC14]{yoder2014fixed}
Theodore~J. Yoder, Guang~Hao Low, and Isaac~L. Chuang.
\newblock Fixed-point quantum search with an optimal number of queries.
\newblock {\em Phys. Rev. Lett.}, 113:210501, Nov 2014.

\bibitem[ZPFP20]{ZhangPeyreFadiliPereyra2020}
Kelvin~Shuangjian Zhang, Gabriel Peyr{\'e}, Jalal Fadili, and Marcelo Pereyra.
\newblock {Wasserstein control of mirror-Langevin Monte Carlo}.
\newblock In {\em Proceedings of the 33rd Annual Conference on Learning Theory
  (COLT 2020)}, volume 125 of {\em Proceedings of Machine Learning Research},
  pages 1--28, 2020.

\bibitem[Zwo12]{zworski2012}
Maciej Zworski.
\newblock {\em Semiclassical analysis}.
\newblock Graduate Studies in Mathematics. American Mathematical Society, 2012.

\end{thebibliography}

\appendix

\section{Spectra of Schr\"odinger operators}

\begin{definition}[Spectrum] \label{def:resolvent}
    Let $T$ be a closed operator on a Hilbert space $H$. Then a complex number $\lambda$ belongs to the resolvent set $\rho(T)$ of $T$ if the operator $T-\lambda I$ has a bounded, everywhere on $H$ defined, inverse $\mathcal R_\lambda(T) = (T-\lambda I)^{-1}$. The set $\sigma(T) = \C \setminus \rho(T)$ is called the \emph{spectrum} of $T$.
\end{definition}

\label{sec:discrete spec}

Let $\LB$ be the Laplace--Beltrami operator on a complete Riemannian manifold $(\manifold,g)$, as in \cref{sec:LB def}. Let $V \in C^\infty(\manifold)$ be a potential with $V(x) \to \infty$ as $x \to \infty$. Here we give a (short) proof that the spectrum of the Schr\"odinger operator $\H = -\LB +V$ is discrete, for this we follow the proof of \cite{urakawa1991laplace}. For comparison, on $\R^n$ a similar statement can be found, e.g., in \cite[Prop.~12.7]{schmudgenUnbdd}.

\begin{theorem}[{cf.~\cite[Thm.~1.8]{urakawa1991laplace}}] \label{thrm:discrete}
Let $\LB$ be the Laplace--Beltrami operator on a complete Riemannian manifold $(\manifold,g)$ and let $V \in C^\infty(\manifold)$ with compact sublevel sets. 
Then the Schr\"odinger operator $\H = -L + V$ has a spectrum consisting only of eigenvalues with finite multiplicities.
\end{theorem}

\begin{proof}
Let $V_\min \coloneqq \min_{x \in \manifold} V(x)$, then $-(|V_\min| + 1)$ belongs to the resolvent of $\H$, that is, the operator $\H + |V_\min| + 1 : \Dom(\H) \to L^2(\manifold)$ has a bounded inverse $R = (\H + |V_\min| + 1)^{-1}$. 
The spectrum of $\H$ consists only of eigenvalues with finite multiplicities if and only if $R$ is a compact operator (cf.~\cite[Prop.~2.11]{schmudgenUnbdd}). We thus show that $R$ is a compact operator. To that end, let $S$ be a bounded subset of $L^2(\manifold)$. Let $C>0$ be such that $\|f\|_{L^2(\manifold)} \leq C$ for all $f \in S$. Let $\{f_m\}_{m \in \N}$ be a sequence in $S$ and define the sequence $\{u_m\}_{m \in \N}$ via $u_m = R f_m$ for $m \in \N$. We show that $\{u_m\}$ has a convergent subsequence in~$L^2(\manifold)$.

First note that $u_m$ is a bounded sequence in $H^1(\manifold)$. Indeed, writing $u=u_m$, we have
\begin{align*}
    \|u\|_{H^1(\manifold)}^2 &= \|du\|_{L^2(\manifold)}^2 + \|u\|_{L^2(\manifold)}^2 \\
    &= -\langle u, L u\rangle + \|u\|_{L^2(\manifold)}^2 \\
    &\leq \langle u, (-L + V + \abs{V_\min}) u\rangle + \|u\|_{L^2(\manifold)}^2 \\
    &= \langle u,(\H+|V_\min| + 1)u\rangle \\
    &\leq \|u\|_{L^2(\manifold)} \|(\H+|V_\min| + 1)u\|_{L^2(\manifold)} \\
    &\leq \|u\|_{H^1(\manifold)} C,
\end{align*}
where in the last line we use that $(\H+|V_\min| + 1)u = f_m \in S$.

Now fix a reference point $x_0 \in \manifold$ and let $\{r_k\}_{k \in \N} \subseteq \N$ be a sequence with the property that if $d(x,x_0) \geq r_k$, then $V(x) \geq k$ (here we use that the sublevel sets of $V$ are compact). For each $k \in\N$, let $B_k = \{x \in \manifold : d(x,x_0) \leq r_k\}$. We may assume without loss of generality that the~$B_k$ are smooth manifolds with boundary; otherwise, we can replace~$r_k$ by a slightly larger~$r_k$, and avoid critical values of~$d(x,x_0)$ by appealing to Sard's theorem~\cite[Thm.~6.10]{leeIntroductionSmoothManifolds2012}, which is enough to guarantee smoothness of the sublevel set~\cite[Prop.~5.47]{leeIntroductionSmoothManifolds2012}.
Let~$B_k^\circ = \{ x \in \manifold : d(x,x_0) < r_k \} = B_k \setminus \partial B_k$.
Then the sequence $\{u_m\}$ is bounded in $H^1(B_k^\circ)$ for any $k \in \N$, as~$\norm{u_m}_{H^1(B_k^\circ)} \leq \norm{u_m}_{H^1(M)} \leq C$. Then, the Rellich--Kondrachov theorem (\cref{thrm:rellich}) shows the existence of a subsequence $\{u_{1,m}\}$ of $\{u_m\}$ that converges in $L^2(B_1^\circ)$, and a subsequence $\{u_{2,m}\}$ of $\{u_{1,m}\}$ that converges in $L^2(B_2^\circ)$, and so on. We let $v_k = u_{k,k}$ for $k \in \N$ and we show that $\{v_k\}_{k \in \N}$ is Cauchy (and thus convergent) in $L^2(\manifold)$.  Fix $\eps > 0$ and let $N \in \N$ be such that $C^2/N < \eps/6$. For any $m \in \N$, we have
\begin{align*}
    \int_{\manifold \setminus B_N^\circ} |u_m|^2 d\vol_g &= \int_{\manifold \setminus B_N^\circ} (V+|V_\min|)^{-1} \abs*{\sqrt{V+|V_\min|} \, u_m}^2 d\vol_g \\
    &\leq (N+|V_\min|)^{-1} \int_{\manifold \setminus B_N^\circ} \abs*{\sqrt{V+|V_\min|} \, u_m}^2 d\vol_g \\
    &\leq N^{-1} \langle u_m, (\H+|V_\min|+1) u_m\rangle \\
    &\leq C^2 N^{-1} < \eps/6,
\end{align*}
where in the second inequality we use that $-\LB \geq 0$, and the last inequality uses that~$\norm{u_m}_{H^1(M)} \leq C$ and~$(\H + \abs{V_\min} + 1) u_m = f_m$ has norm at most~$C$.
Finally, since $\{v_k\}_{k \in \N}$ is convergent in $L^2(B_N^\circ)$, there exists a $K$ such that for all $k,\ell \geq K$ we have $\|v_k -v_\ell\|_{L^2(B_N^\circ)}^2 \leq \eps/3$, and therefore
\begin{align*}
    \|v_k - v_\ell\|_{L^2(\manifold)}^2 &= \|v_k -v_\ell\|_{L^2(B_N^\circ)}^2 + \int_{\manifold \setminus B_N^\circ} |v_k-v_\ell|^2 d\vol_g \\
    &\leq \eps/3 + \int_{\manifold \setminus B_N^\circ} 2(|v_k|^2+|v_\ell|^2) d\vol_g \\
    &\leq \eps/3 + 2(\eps/6+\eps/6) = \eps.
\end{align*}
This shows that the sequence $\{v_k\}_{k \in \N}$ is Cauchy (and thus convergent) in $L^2(\manifold)$.
\end{proof}

We rely on the following version of the Rellich--Kondrachov theorem. For~$n \geq 3$ we refer the reader to~\cite[Thm.~2.34]{aubinNonlinearAnalysisManifolds1982}; the result can also be established for~$n=1,2$ by using the right Euclidean analog of the statement, which can be found in~\cite[Thm.~9.16]{brezisFunctionalAnalysisSobolev2011}.
\begin{theorem}[{Rellich--Kondrachov on compact manifolds with boundary}] \label{thrm:rellich}
Let $(\manifold,g)$ be a smooth, compact, $n$-dimensional Riemannian manifold with boundary and~$\manifold^\circ = \manifold \setminus \partial \manifold$. Then $H^1(\manifold^\circ)$ is compactly embedded in $L^2(\manifold^\circ)$.
\end{theorem}
Here $H^1(\manifold^\circ)$ is the completion of $C^\infty(\manifold^\circ)$ with respect to $\|\cdot\|_{H^1(\manifold^\circ)}$, where
\[
\|u\|_{H^1(\manifold^\circ)}^2 = \int_{\manifold^\circ} \|\nabla u\|^2 d\vol_g + \int_{\manifold^\circ} |u|^2 d\vol_g
\]

\section{Proof of Proposition~\ref{prop:comm2LB}}
\label{sec:comm2LB proof}
\begin{proof}[Proof of~\cref{prop:comm2LB}.]
    Let~$\phi \in C_c^\infty(\domain)$.
    Then
    \begin{align*}
        [h, [h, L]] \phi & = h ([h, L] \phi) - [h,L] (h \phi) \\
        & = h (h L(\phi) - L(h \phi)) - h (L(h \phi)) + L(h^2 \phi) \\
        & = h^2 L(\phi) - 2 h L(h \phi) + L(h^2 \phi).
    \end{align*}
    From this expression we deduce that for~$\phi_1, \phi_2 \in C_c^\infty(\domain)$, 
    \begin{align}
    \label{eq:hhL action}
        \ipriem{\phi_2}{[h, [h, L]] \phi_1} & = \ipriem{\phi_2}{h^2 L(\phi_1) - 2 h L(h \phi_1) + L(h^2 \phi_1)} = \ipriem{[h, [h, L]] \phi_2}{\phi_1}
    \end{align}
    using self-adjointness of~$L$ and~$\ipriem{\phi_2}{h^2 L(\phi_1)} = \ipriem{h^2 \phi_2}{ L(\phi_1)}$, $\langle {\phi_2}{h L(h \phi_1)}= \ipriem{h \phi_2}{L(h \phi_1)}$.
    The defining property of~$L$ then yields
    \begin{align*}
        \ipriem{\phi_2}{h^2 L(\phi_1)}
        & = - \int_\domain g_x^*(d(h^2 \phi_2), d\phi_1) \, d\vol_g \\
        & = - \int_\domain g_x^*(h^2 \, d\phi_2 + 2 h \phi_2 \, dh, d\phi_1) \, d\vol_g \\
        & = - \int_\domain h^2 g_x^*(d\phi_2, d\phi_1) + 2 h \phi_2 g_x^*(dh, d\phi_1) \, d\vol_g
    \end{align*}
    and
    \begin{align*}
        \ipriem{h \phi_2}{L(h \phi_1)}
        & = - \int_\domain g_x^*(d(h \phi_2), d (h \phi_1)) \, d\vol_g \\
        & = - \int_\domain g_x^*(h \, d\phi_2 + \phi_2 \, dh, h \, d\phi_1 + \phi_1 \, dh) \, d\vol_g \\
        & = - \int_\domain h^2 g_x^*(d\phi_2, d\phi_1) + \phi_2 g_x^*(dh, d\phi_1) + \phi_1 g_x^*(d\phi_2, dh) + \phi_1 \phi_2 g_x^*(dh, dh) \, d\vol_g.
    \end{align*}
    Using this to expand~\cref{eq:hhL action} leads to the identity
    \begin{align*}
        \ipriem{\phi_2}{[h, [h, L]] \phi_1} = -2 \int_\domain \phi_2 \phi_1 g_x^*(dh, dh) \, d\vol_g =  -2 \int_D \phi_2 \phi_1 \normriem{dh}^2 \, d\vol_g.
    \end{align*}
    From this it follows that $[h, [h, L]]$ acts as multiplication by~$-2\normriem{dh}^2$.
\end{proof}

\end{document}